\documentclass[11pt,letterpaper]{article}

\usepackage[letterpaper,margin=1in]{geometry}
\usepackage{microtype}
\usepackage{amsmath,amssymb,amsthm,mathtools}
\usepackage{bm}
\usepackage{enumitem}
\usepackage{booktabs,longtable,array}
\usepackage{graphicx}
\usepackage{tikz}

\usepackage{xcolor}
\usepackage{colortbl}
\usepackage[colorlinks=true,linkcolor=blue!60!black,citecolor=blue!60!black,urlcolor=blue!60!black,bookmarksnumbered=true,bookmarksdepth=3]{hyperref}
\usepackage[capitalize,nameinlink]{cleveref}
\usepackage{etoolbox}

\newif\ifdraft
\draftfalse

\theoremstyle{plain}
\newtheorem{theorem}{Theorem}[section]
\newtheorem{lemma}[theorem]{Lemma}
\newtheorem{proposition}[theorem]{Proposition}
\newtheorem{corollary}[theorem]{Corollary}

\newtheorem{introthm}{Theorem}

\theoremstyle{definition}

\newtheorem{fact}[theorem]{Fact}
\theoremstyle{remark}

\crefname{claim}{Claim}{Claims}
\Crefname{claim}{Claim}{Claims}
\crefname{fact}{Fact}{Facts}
\Crefname{fact}{Fact}{Facts}
\crefname{introthm}{Theorem}{Theorems}
\Crefname{introthm}{Theorem}{Theorems}

\numberwithin{equation}{section}

\newcommand{\R}{\mathbb{R}}

\newcommand{\Z}{\mathbb{Z}}
\newcommand{\E}{\mathbb{E}}
\renewcommand{\Pr}{\operatorname{Pr}}
\newcommand{\Rev}{\operatorname{Rev}}
\newcommand{\OPT}{\operatorname{OPT}}
\newcommand{\supp}{\operatorname{supp}}
\newcommand{\pos}[1]{\left[#1\right]_{+}}
\newcommand{\one}{\bm{1}}
\newcommand{\cF}{\mathcal{F}}
\newcommand{\ind}[1]{\mathbf{1}\{#1\}}
\newcommand{\cC}{\mathcal{C}}
\newcommand{\cT}{\mathcal{T}}

\newcommand{\cY}{\mathcal{Y}}

\newcommand{\cN}{\mathcal{N}}
\newcommand{\cZ}{\mathcal{Z}}
\newcommand{\cV}{\mathcal{V}}
\newcommand{\cA}{\mathcal{A}}
\newcommand{\cB}{\mathcal{B}}

\newcommand{\cJ}{\mathcal{J}}

\newcommand{\cW}{\mathcal{W}}
\newcommand{\cX}{\mathcal{X}}
\newcommand{\sharpP}{\#\mathrm{P}}

\DeclareMathOperator{\Capac}{Cap}
\DeclareMathOperator{\loss}{loss}
\DeclareMathOperator{\bit}{bit}
\DeclareMathOperator{\rev}{rev}
\newcommand{\ptie}{p_{\mathrm{tie}}}

\newif\ifshowchanges
\showchangestrue

\title{Optimal Multi-Item Auctions with I.I.D.\ Values}
\author{%
  Boyu Liu\thanks{Renmin University of China. Email: \texttt{liuboyu03@ruc.edu.cn}}
  \and
  Zhengyang Liu\thanks{Beijing Institute of Technology. Email: \texttt{zhengyang@bit.edu.cn}}
  \and
  Zihe Wang\thanks{Renmin University of China. Email: \texttt{wang.zihe@ruc.edu.cn}}
}
\date{}

\begin{document}

\pagenumbering{roman}
\maketitle
\thispagestyle{empty}
\begin{abstract}
We study revenue-maximizing auctions for additive buyers whose item values are independent and identically distributed (i.i.d.) draws from a finite distribution, under dominant-strategy incentive compatibility (DSIC) and ex-post individual rationality (IR). For any two-point distribution and any numbers of buyers and items, we construct an explicit deterministic auction that achieves the optimal expected revenue, even among randomized DSIC, ex-post IR mechanisms. Furthermore, its expected revenue can be computed in polynomial time. In contrast, we show that the problem becomes intractable once the number of support points is part of the input: even with a single buyer, and even with two items, both constructing an optimal mechanism and computing the optimal revenue are $\sharpP$-hard.

\end{abstract}
\clearpage
\begingroup
\hypersetup{linkcolor=black}
\makeatletter
\patchcmd{\l@section}{\addvspace{1.0em \@plus\p@}}{\addvspace{0.25em \@plus\p@}}{}{}
\makeatother
\setcounter{tocdepth}{2}
\tableofcontents
\endgroup
\clearpage
\pagenumbering{arabic}
\section{Introduction}\label{sec:intro}
Designing revenue-optimal auctions is a central challenge in mechanism design. For a single item, Myerson~\cite{Mye81} famously characterized the optimal auction via virtual values, providing an explicit description that can be computed efficiently. In contrast, no comparable unified theory exists for multiple items. Multidimensional screening problems generally lack closed-form solutions~\cite{MM88,Arm96,RC98}. Even in the basic setting of a single additive buyer, optimal mechanisms exhibit notorious complexity: they can require bundling~\cite{AY76,MMW89}, randomization~\cite{MV07,Pav11}, and even a continuum of lotteries~\cite{DDT13}, and optimal revenue can strictly decrease under stochastic value increases~\cite{HR15}. Because exact analytical solutions are known only for highly structured priors, attention has naturally turned to the computational problem: computing the optimal revenue exactly and constructing an optimal mechanism.

The computational complexity of this problem fundamentally depends on the solution concept adopted. Under Bayesian incentive compatibility (BIC), truth-telling maximizes a buyer's expected utility given that rivals report truthfully and their valuations follow the common prior. Dominant-strategy incentive compatibility (DSIC) is strictly more robust: truth-telling is an optimal strategy regardless of other buyers' reports, freeing agents from needing distributional knowledge about their peers~\cite{BM05,CE07}. For discrete valuation distributions, the optimal auction can be formulated as a linear program~\cite{CS02}, but its naive representation is exponential in both the number of buyers and items. Under BIC, however, a buyer's incentives depend only on her interim allocation rule; leveraging these reduced-form representations~\cite{Bor91}, an optimal mechanism can be computed in time polynomial in the total number of buyer types~\cite{CDW12a,CDW12b,AFHHM19}. When each item takes $K$ possible values, a buyer has $K^m$ types, yielding a polynomial-time algorithm whenever the number of items $m$ is fixed. Under DSIC, this dimensionality reduction fails: incentive constraints must hold pointwise at every report profile of the remaining buyers, and no general compact formulation is known. While BIC and DSIC achieve identical optimal revenues in single-parameter settings~\cite{MV10,GGKMS13}, DSIC can extract strictly less revenue in multi-item environments~\cite{TW16,Yao17}. The complexity of exact revenue optimization under DSIC thus requires separate investigation.

Existing hardness results, however, heavily exploit asymmetry across items. For a single additive buyer, computing an optimal mechanism is $\sharpP$-hard even when valuations across items are mutually independent and drawn from two-point supports~\cite{DDT14}. Crucially, this hardness depends on items having distinct distributions, along with analogous intractability results for optimal deterministic mechanisms~\cite{CMPY18} and grand-bundle pricing~\cite{DDT12}. In stark contrast, imposing symmetry dramatically alters the landscape. When item valuations are drawn i.i.d. from a two-point prior, explicit optimal mechanisms emerge: for a single buyer~\cite{CMPY18}, for two items~\cite{Yao17}, and under BIC for arbitrary numbers of buyers and items~\cite{MT26}. Beyond two-point distributions, symmetry remains powerful: when the support size and either the number of buyers or items are fixed, exploiting symmetry among items or buyers yields polynomial-time algorithms~\cite{DW12}.

Motivated by this contrast, we examine the canonical model endowed with complete symmetry across both buyers and items: $n$ additive buyers compete for $m$ items, with all $nm$ values drawn independently from a shared discrete distribution $F$ supported on $K$ points. The mechanism must satisfy DSIC and ex-post individual rationality (IR). Because the input to the algorithm is simply the tuple $(F, n, m)$, whereas the full type space spans $K^{nm}$ valuation profiles, the challenge is sharp:
\begin{quote}
Does joint buyer and item symmetry suffice to make exact DSIC optimization computationally tractable?
\end{quote}
\subsection{Our results}\label{sec:results}

The answer depends critically on the support size $K$ (\Cref{tab:results}). Formally, we specify the prior $F$ by its rational support points and associated probabilities, and provide the parameters $n$ and $m$ in unary. We denote by $\OPT(F; n, m)$ the maximum expected revenue achievable by any randomized, DSIC, and ex-post IR mechanism, referring to any revenue-maximizing mechanism simply as \emph{optimal}. By \emph{constructing} an optimal mechanism, we mean producing a representation that enables evaluating the allocation probabilities and payments at any given report profile in polynomial time (\Cref{sec:conventions} formalizes the corresponding notion of $\sharpP$-hardness).

\begin{table}[t]
\centering
\caption{Exact DSIC revenue maximization with $n$ buyers, $m$ items and i.i.d.\ values with $K$ support points. Each entry applies both to computing the optimal revenue and to constructing an optimal mechanism. Shaded: this paper.}
\label{tab:results}
\small
\renewcommand{\arraystretch}{1.3}
\setlength{\tabcolsep}{6pt}
\begin{tabular}{>{\raggedright\arraybackslash}p{3.1cm}>{\raggedright\arraybackslash}p{5.7cm}>{\raggedright\arraybackslash}p{5.7cm}}
\toprule
 & $K$ fixed & $K$ part of the input\\
\midrule
$n$ and $m$ fixed & polynomial time~\cite{CS02} & polynomial time~\cite{CS02}\\
$n$ fixed, any $m$ & polynomial time~\cite{DW12} & \cellcolor{black!9}$\sharpP$-hard, already for $n=1$ (\Cref{thm:introB})\\
$m$ fixed, any $n$ & polynomial time~\cite{DW12} & \cellcolor{black!9}$\sharpP$-hard, already for $m=2$ (\Cref{thm:introC});\newline polynomial time for $m=1$~\cite{Mye81}\\
any $n$ and $m$ & \cellcolor{black!9}explicit deterministic optimal auction for $K=2$ (\Cref{thm:introA}) & \cellcolor{black!9}$\sharpP$-hard\\
\bottomrule
\end{tabular}
\end{table}

For two-point distributions, there is an explicit optimal auction for all $n$ and $m$.

\begin{introthm}[Two-point values; \Cref{sec:A}]\label{thm:introA}
If $F$ has two support points, then for all $n$ and $m$ there is an explicit deterministic auction that is optimal among all randomized DSIC, ex-post IR mechanisms, and its expected revenue can be computed in polynomial time.
\end{introthm}

\Cref{thm:introA} frees both $n$ and $m$ under DSIC and ex-post IR, and shows that randomization then brings no additional revenue. Together with the BIC optimum of~\cite{MT26}, both optimal revenues for two-point values are now explicit.

Exact optimization also takes polynomial time once any two of $n$, $m$, $K$ are fixed: the finite-type linear program~\cite{CS02} then has polynomial size, directly or after compression by the symmetries among items or among buyers~\cite{DW12}. Once $K$ is part of the input, however, fixing $n$ or $m$ no longer suffices: the problem is $\sharpP$-hard already for a single buyer, where DSIC and BIC coincide, and already for two items.

\begin{introthm}[A single buyer; \Cref{sec:B}]\label{thm:introB}
Let $n=1$. Computing $\OPT(F;1,m)$ exactly and constructing an optimal mechanism are both $\sharpP$-hard.
\end{introthm}

\begin{introthm}[Two items; \Cref{sec:two}]\label{thm:introC}
Let $m=2$. Computing $\OPT(F;n,2)$ exactly and constructing an optimal mechanism are both $\sharpP$-hard.
\end{introthm}

In the symmetric model, the hardness thus comes from the number of support points rather than from asymmetry among the items. Nor does it come from complicated optimal mechanisms: in the instances of \Cref{thm:introB}, posting the grand bundle at the better of two known prices is optimal (\Cref{cor:twopricesB}), so the hardness also holds when only deterministic mechanisms are allowed. For two items, the optimal BIC revenue can be computed in polynomial time~\cite{CDW12a}, whereas computing the optimal DSIC revenue is $\sharpP$-hard even within an additive error $e$ that is part of the input (\Cref{thm:C}); so, unless every function in $\sharpP$ can be computed in polynomial time, DSIC and BIC are separated not only in revenue~\cite{Yao17} but also in complexity.


\subsection{Technical overview}\label{sec:overview}

All three theorems rest on the duality framework of Cai, Devanur and Weinberg~\cite{CDW16}: a flow on the incentive constraints assigns a virtual value to every type and every item, and the revenue of every DSIC mechanism is at most its expected virtual surplus, the expected sum of the virtual values weighted by the allocation probabilities. For \Cref{thm:introA}, we construct explicit flows and show that our auction attains this bound.

The two hardness results follow the same plan. We construct priors under which the best mechanism in a known family is optimal, and both which mechanism is best and what the optimal revenue is depend on a $\sharpP$-hard count. Finding this mechanism is not enough, for two reasons. First, the optimal revenue is a maximum over all mechanisms; to determine it exactly and read off the count, we must show that no mechanism outside the family does better. Second, in the construction problem, the reduction may receive the allocations of any optimal mechanism, so the answer must be visible in every optimal mechanism.

Stability resolves both issues. Each of our priors is a small modification of a base prior. Under the base prior, the mechanisms of the family attain the revenue bound given by a flow, and the amount by which the revenue of a mechanism falls short of this bound is its \emph{loss}. Stability says that a mechanism with a small loss is close to some mechanism of the family: at every type of the base prior, their allocations and payments differ by at most a constant times the loss. The modification ensures that, under the modified prior, a mechanism earns at most a small multiple of this difference more than that mechanism of the family: in \Cref{thm:introB}, because the types that the modification adds have small probability; in \Cref{thm:introC}, because the modification makes the capacity gained by a deviation almost worthless. So what a mechanism gains by deviating from the family is always less than the loss that the deviation causes, and the optimal revenue is that of the best mechanism in the family. Moreover, every optimal mechanism has zero loss, so at an explicit report profile, the allocations of every optimal mechanism answer a comparison involving the count. Binary search over the threshold of this comparison recovers the count.

\paragraph{Two-point values (\Cref{thm:introA}).} Let $F$ take a low value $a$ with probability $p$ and a high value $b$ with probability $q=1-p$, and let $\Delta=b-a$. A type is the set of items that the buyer values high, so types are still multidimensional. An item that some buyer reports high, however, can be sold at price $b$ to a buyer who reports it high, which extracts its full value; what has to be decided is only the treatment of the items that every buyer reports low, the \emph{all-low} items: whether to sell them, to whom, and at what price. To see how they should be sold, fix the reports of the other buyers, and let $D$ be the set of items that no opponent reports high. The items of $D$ that buyer~$i$ reports high are reported high by her alone and go to her in any case; the others are exactly the all-low items. Selling her the all-low items therefore amounts to selling her the whole of $D$, whose value $a|D|+\Delta e_i$ to her depends only on her \emph{exclusive count} $e_i$, the number of items that only she reports high. Given the opponents' reports, $e_i\sim\mathrm{Bin}(|D|,q)$, and the multidimensional problem becomes one-dimensional in the count.

To apply Myerson's method to the count, we first need the virtual values of the all-low items. Types with more high items can pose as types with fewer, so the price of $D$ must leave them an information rent. When there are $c$ all-low items, this gives each of them the virtual value
\[
\phi_{c,r}=a-\Delta\,\frac{\Pr[R>r]}{c\,\Pr[R=r]},\qquad R\sim\mathrm{Bin}(c+r,q),
\]
for a buyer with exclusive count $r$, where $R$ is the number of the $c+r$ items of $D$ that she values high. In other words, $c\,\phi_{c,r}$ is the discrete Myerson virtual value at count $r$ of the value of $D$, minus the value $br$ of her high items. For each buyer, all $c$ all-low items have the same virtual value $\phi_{c,r}$, which is strictly increasing in her exclusive count $r$. Thus, the buyer with the largest exclusive count has the highest virtual value for every all-low item. Since virtual surplus is additive across items, all these items should be allocated to this buyer if her virtual value is positive, and otherwise kept unsold. This amounts to a reserve price on the exclusive count. The auction therefore ranks the buyers by exclusive count, breaking ties by a fixed priority order: an item that some buyer reports high goes to the highest-ranked buyer who reports it high, and the all-low items go together to the top-ranked buyer if her virtual value is positive and are kept otherwise. For a single buyer, this is the item pricing with a discounted grand bundle that~\cite{CMPY18} show to be optimal.

Since the ranking depends on the reports, a buyer may want to report more items high, to raise her own count or to lower those of her opponents; so truthfulness must be shown. For fixed reports of the opponents, her choices are simple: she can buy the items that she reports high at $b$ each, or the whole of $D$ at $a|D|+\Delta\Theta$, where the threshold $\Theta$ is the least count with which she obtains $D$, one that both reaches the reserve and ranks her ahead of every opponent; an item that an opponent also reports high she can buy only at $b$. Within $D$, this amounts to a posted price on the count, so truthful reporting is optimal. Her only other option is to report high an item that an opponent reports high, which lowers that opponent's exclusive count: each such item lowers $\Theta$ by at most one and saves her at most $\Delta$, but she pays $b$ for an item worth $a$ to her, which costs her $\Delta$. So the auction is DSIC and ex-post IR (\Cref{lem:menuA,prop:dsicA}).

It remains to show that no mechanism, randomized or not, earns more. Since the DSIC constraints hold separately at each profile of the opponents' reports, the bound can use a separate flow for each buyer and each profile of her opponents' reports; the revenue of every DSIC, ex-post IR mechanism is then at most the expectation of the sum, over the items, of the largest virtual value of the item among the buyers, or zero if all of them are negative (\Cref{cor:multi}). Our flow uses only two kinds of constraints: the downward IC constraints that a type does not gain by reporting one of its high items low, and the IR constraint of the empty set. On $D$, it is the Myerson flow on counts: the probability of the types with count at least $r$ flows from count $r$ to count $r-1$. The virtual value of an all-low item is then exactly $\phi_{c,r}$, that of an item reported high is $b$, and that of an item that an opponent reports high is at most $a$. In the auction, every constraint used by the flow is tight, and every item goes to a buyer with the largest virtual value if this value is positive; so the auction attains the bound. The bound is an expectation over the distribution of the exclusive counts and can be computed in polynomial time; the flow computation is related to the symmetric computations of~\cite{MT26}.

\paragraph{A single buyer (\Cref{thm:introB}).} For a single buyer, DSIC and BIC coincide, and the optimal mechanism solves a linear program, but there are $K^m$ types; the hardness must come from the number of items. To see how many items produce a count, fix positive integers $b_1,\dots,b_N$ and $T_0$, and let $\cZ(T_0)$ be the number of index tuples $(i_1,\dots,i_m)\in\{1,\dots,N\}^m$ with $\sum_{j=1}^mb_{i_j}\ge T_0$. With a standard digit encoding, the difference of two such counts is $m!$ times the number of solutions of a subset-sum instance (\Cref{lem:tailB}), so computing $\cZ(T_0)$ is $\sharpP$-hard. This count arises in grand-bundle pricing: if the value of each item is drawn independently and uniformly from $b_1,\dots,b_N$, the mechanism that posts the grand bundle at price $T_0$ earns $T_0\,\cZ(T_0)/N^m$, so even the revenue of such a mechanism is $\sharpP$-hard to compute. The optimal revenue, however, is a maximum over all mechanisms; an optimal mechanism may use lotteries and have a complicated structure, and its revenue need not reflect the count. We construct priors under which, of the two mechanisms that post the grand bundle at one of two known prices, the one with the larger revenue is optimal, and which one it is depends on the count. The optimal revenue is then the larger of the two revenues, and every optimal mechanism shows which revenue is larger, so the count can be read off.

To prove that the better of these two mechanisms is optimal, we must control all mechanisms, and we can do so only near a prior whose optimal mechanisms are known. We therefore start from such a prior and perturb it slightly; for a small perturbation to decide which of the two mechanisms earns more, they must tie before the perturbation, that is, earn the same revenue. Take the prior $F_0$ on the two values $a$ and $a+1$, with $\Pr[a]=p$. Posting the grand bundle at price $ma$ sells to every type and earns $ma$; posting it at price $ma+1$ loses only the lowest type $a\one$, whose values are all $a$, and earns $(1-p^m)(ma+1)$. The two revenues are equal when $ma=p^{-m}-1$. Let $M_z$ be the random mixture of these two posted prices that posts $ma$ with probability $z$ and $ma+1$ otherwise; thus $M_1$ and $M_0$ are the two posted-price mechanisms themselves. Every $M_z$ earns the same revenue, and under $M_z$ the lowest type buys the grand bundle with probability $z$. All these mechanisms are optimal under $F_0$. Under the Myerson flow on counts of \Cref{thm:introA}, for a single buyer, the virtual values of the lowest type are zero, which is the tie in another form: the allocation of the lowest type does not affect the revenue bound. The virtual values of all other types are positive, so the revenue of every mechanism is at most the expected virtual surplus obtained by giving the grand bundle to all other types. Every $M_z$ allocates in this way, and every constraint used by the flow is tight under $M_z$; so every $M_z$ attains this bound, which is therefore $ma$.

The perturbation must make the difference between the revenues of $M_0$ and $M_1$ reflect the count. The perturbed prior $F$ adds $N$ intermediate values $a+\zeta_1,\dots,a+\zeta_N$, each of probability $\eta/N$, where $\eta$ is small. The mechanism $M_1$ still sells the grand bundle to every type, while the new types to which $M_0$ sells are determined by the intermediate values. Each $\zeta_\ell$ is close to $1/m$, so a type whose values are all intermediate has a value sum close to $ma+1$; the small differences among the $\zeta_\ell$ encode the $b_\ell$, so that the type $(a+\zeta_{i_1},\dots,a+\zeta_{i_m})$ buys at price $ma+1$ exactly when $\sum_jb_{i_j}\ge T_0$. Each such type has probability $(\eta/N)^m$, and $\cZ(T_0)$ of them buy, with total probability $(\eta/N)^m\cZ(T_0)$. Hence the revenue of $M_0$ minus that of $M_1$ equals $ma+1$ times this probability, minus a known constant. The reduction for the construction problem can read off only which of $M_0$ and $M_1$ earns more, that is, the sign of this difference, and must therefore binary search over a threshold to find $\cZ(T_0)$. So we adjust $p$ and $a$ so that the constant makes $M_0$ and $M_1$ tie exactly when $\cZ(T_0)=c+\frac12$, for an integer $c$ chosen by the reduction: $M_1$ earns more if $\cZ(T_0)\le c$, and $M_0$ otherwise.

Under $F_0$, every $M_z$ is optimal, including $M_1$ and $M_0$. After the perturbation, the revenue of $M_z$ is still linear in $z$, so the best mechanism in this family is the better of $M_1$ and $M_0$. Whether it is optimal among all mechanisms is not obvious: the perturbation adds new types, and other mechanisms may exploit them to earn more. Two things remain to be shown: under the perturbed prior, no mechanism earns more than the better of $M_1$ and $M_0$; and every optimal mechanism gives the lowest type the same allocation as this mechanism. To this end, we compare an arbitrary mechanism $M$ with $M_z$, where $z$ is the probability with which $M$ allocates item~$1$ to the lowest type. Under $F_0$, $M$ earns $\varepsilon(M)=ma-\Rev_{F_0}(M)$ less than $M_z$; we call this its \emph{loss}. We show that after the perturbation, what $M$ earns in excess of $M_z$ is always less than this loss. This takes two steps: first, at the types with all values in $\{a,a+1\}$, $M$ differs from $M_z$ by at most a constant times the loss; second, what $M$ collects from the new types in excess of $M_z$ is controlled by this difference, and the new types have small total probability.

The first step uses the identity of weak duality: the revenue of a mechanism equals its expected virtual surplus minus the slacks of the constraints used by the flow, weighted by the flow. Under $F_0$, the expected virtual surplus is at most $ma$, and the shortfall comes from not giving the grand bundle to the other types; so the loss equals the virtual surplus lost in this way plus the weighted slacks (\Cref{prop:lossB}). Both are at most the loss, so the allocations of $M$ at the other types differ from the grand bundle by at most a constant times the loss, and the constraints used by the flow hold with equality up to an error of the same order. A mechanism that gives the grand bundle to the other types and satisfies these constraints with equality must coincide with $M_z$ at the types with all values in $\{a,a+1\}$, where $z$ is determined by the allocation of the lowest type. Hence $M$ differs from $M_z$ at these types by at most a constant times the loss.

The second step concerns the new types. A new type can report a type with all values in $\{a,a+1\}$, and such types can report the new type; through the incentive constraints in both directions, a new type pays under $M$ at most a constant times the difference of the first step more than under $M_z$. The exception is the types whose value sums lie just below $ma+1$: under $M_z$, they are almost indifferent between buying only when the posted price is $ma$ and buying at either price, so even a mechanism that differs little from $M_z$ can sell them more and charge them more. 

Our intermediate values ensure that every type whose value sum is below $ma+1$ has value sum at most $ma+1-\nu$; the extra payment of a new type is then at most a constant times the loss divided by $\nu$, where the constant depends only on $m$ and $a$. The perturbation changes the probabilities of the types by at most $2m\eta$ in total, and we take $\eta$ small enough relative to $\nu$ that what $M$ gains from the perturbation in excess of $M_z$ is at most a quarter of its loss (\Cref{thm:stabB}). This cannot make up for a positive loss; so no mechanism earns more than the mechanism $M_z$ with which it is compared, and hence than the better of $M_1$ and $M_0$ (\Cref{cor:twopricesB}).

An optimal mechanism has zero loss, and the mechanism $M_z$ with which it is compared is also optimal, so $M_z$ is this better mechanism, and the optimal mechanism agrees with it at the lowest type.

This lets us read off the count. For the optimal revenue, take $c=0$: if $\cZ(T_0)\ge1$, the optimal revenue exceeds $ma$ by a known multiple of $\cZ(T_0)-\frac12$, and otherwise it equals $ma$; so one query gives $\cZ(T_0)$, and two values of $\cZ$ give the number of subset-sum solutions. In the construction problem, the reduction receives allocations rather than revenues, and every optimal mechanism allocates each item to the lowest type with probability $1$ if $\cZ(T_0)\le c$ and $0$ otherwise, as $M_1$ or $M_0$ does; so binary search over $c$ finds $\cZ(T_0)$.

\paragraph{Two items (\Cref{thm:introC}).} With two items, a single buyer has only $K^2$ types, and for a fixed number of buyers an optimal mechanism can be computed by a linear program of polynomial size. The hardness of \Cref{thm:introC} therefore comes from the growing number of buyers and the exponentially many report profiles that come with it. Under BIC, the problem can be compressed to polynomial size through interim allocations~\cite{CDW12a}; under DSIC, the incentive constraints must hold separately at every profile of the opponents' reports, while the allocations of different buyers share the supply.

Our construction uses this competition for supply to produce a count. Consider a buyer who wants both items. In the limit that we consider, she receives equal amounts of the two items, so her allocation is determined by the item of which less supply remains, the \emph{bottleneck} item. Which item is the bottleneck depends on how much the buyers who want a single item take; in our construction, each of them takes an amount proportional to her coefficient, and, as we show, the bottleneck is decided by which of the two groups, the opponents who want item~$1$ and those who want item~$2$, has the larger sum of coefficients. As a function of one coefficient, her expected allocation therefore has kinks where the two sums are equal: there its slope jumps, by an amount proportional to the total probability of the profiles of opponents whose two sums are equal at that point. By choosing the prior suitably, we give these profiles the same probability, so the jump is proportional to their number, which is the number of solutions of a subset-sum instance; the other profiles contribute only smooth terms, which finite differences remove. This counting relation holds exactly only in the limit. The constants of the actual construction are finite, but they can be chosen so that the error is small, and rounding to the nearest integer recovers the exact count. The key is to carefully construct priors that force every optimal mechanism into this structure, so that the count can be recovered from the optimal revenue or from the allocations of an arbitrary optimal mechanism.

\textit{1.~Canonical menus: capacities describe the allocation of a buyer.} Fixing the reports of the other buyers, we represent the allocation and payment of a buyer as the menu that she faces, from which she chooses her preferred option, including the option of buying nothing. The incentive constraints involve only one menu at a time, so a mechanism is DSIC and ex-post IR if and only if all its menus are IC and IR; the menus of different buyers are linked through the supply constraints. Let
\[
y_1=\max_{\bm t}x_1(\bm t),\qquad y_2=\max_{\bm t}x_2(\bm t)
\]
be the largest probabilities with which a menu allocates the two items; we call them its \emph{capacities}. By weak duality, the revenue of a menu is at most a linear function of its capacities.

We construct priors under which this bound is $\gamma(y_1+y_2)$, where $\gamma$ is a constant, and the menus that attain it are exactly a family of \emph{canonical menus}. The types of the prior fall into three classes: types that buy nothing, \emph{single-item types} that buy only one item, and \emph{bundle types} that buy both items. In a canonical menu, a bundle type receives $(y_1,y_2)$ and pays $y_1+y_2$, and a single-item type that wants item~$1$ receives
\[
(y_1-\beta y_2,\ 0);
\]
item~$2$ is symmetric. The \emph{slope} $\beta$ depends on her value for item~$1$: the higher the value, the smaller the slope, and the more of the item she receives. The slopes of the different values are determined by the positions of the support points, and we encode the instance by choosing these positions.

The capacities of a canonical menu must also satisfy
\[
y_1\ge\kappa y_2,\qquad y_2\ge\kappa y_1,
\]
where $\kappa<1$ is the largest slope, which can be arbitrarily close to~$1$; we take it closer to~$1$ than the other slopes, so that the difference between the two capacities is much smaller than the allocation of a single-item type. As $\kappa\to1$, the two capacities therefore become equal. Writing $y$ for their common value, a single-item type receives $(1-\beta)y$, and we call $1-\beta$ the \emph{coefficient} of her type.

\textit{2.~A greedy method: the canonical menus at each context.} With canonical menus only, the revenue is $\gamma$ times the total capacity, so the problem becomes to maximize the total capacity subject to supply. By symmetry, the capacities may be taken to depend only on the \emph{context} of a menu, the classes of the opponents. We give a greedy method: in increasing order of the number of single-item opponents, deduct what the single-item buyers take, and let the bundle buyers split the remaining supply equally. We show that it yields the unique optimal averaged capacities, and the stability argument shows that the capacities of every optimal mechanism average to them; so they determine the optimal revenue.

In the limit $\kappa\to1$, the greedy method is a simple recursion, and the bottleneck item is decided exactly by which of the two sums of coefficients is larger; so, as functions of a coefficient, the capacities have kinks where the two sums are equal. The actual construction has $\kappa<1$; the greedy method still yields the optimal capacities, and they differ from the values of the limit recursion by at most a constant times $1-\kappa$. We can therefore analyze the kinks in the limit recursion and take $\kappa$ close enough to~$1$ that the error of the finite case can be removed by rounding; the reduction for the construction problem, in turn, uses exactly this difference to break a tie.

\textit{3.~Stability: forcing optimal mechanisms to use canonical menus.} The analysis so far does not exclude non-canonical menus. A mechanism may give less to the single-item types and use the saved supply to raise capacities in other menus. Some menus then earn less than $\gamma(y_1+y_2)$, but the revenue from the additional capacity may make up for this loss.

To exclude this, we add to the prior a very high value $H$ of small probability, and still define capacities by the types without the value $H$. A buyer with the value $H$ can mimic the ordinary type that receives the largest allocation probability, so raising the capacities of the ordinary types increases the information rent that must be left to her. With a suitable probability of $H$, the net revenue contribution of a unit of capacity becomes a small positive number.

On the other hand, we prove a stability result. Call the amount by which the revenue of a menu falls short of the bound its \emph{loss}. A menu with nonnegative payments can be lowered to a canonical menu without raising any allocation probability, and its capacities decrease by at most a constant times its loss. Lowering all menus at once preserves supply; since the net contribution of a unit of capacity is small, the revenue lost with the capacity is less than the loss removed. Hence an optimal mechanism has no loss: at the report profiles in which no buyer reports $H$, it uses only canonical menus, and their capacities average to the greedy capacities.

\textit{4.~Hardness of the optimal revenue: extracting the target kink by finite differences.} We first encode the subset-sum instance as a list of numbers such that the solutions to be counted all choose exactly $k$ of them. For each item, we then introduce ordinary single-item classes, one for each of these numbers, and a special class whose coefficient is controlled by a parameter $R$. With $n=k+2$ buyers and a sufficiently large common offset in the coefficients, the bottleneck item can depend on $R$ only when the opponents are one special type and $k$ ordinary types that want the other item; the comparison of the two sums of coefficients then becomes the comparison of $R$ with the sum of these $k$ numbers.

Hence every report profile of this form whose $k$ numbers sum to the target produces a switch of the bottleneck at the target. These profiles have the same probability, and each of them produces the same jump in slope, so the jump of the total capacity at the target, as a function of $R$, is proportional to the number of solutions.

The total capacity, however, also contains the contributions of all other report profiles, so the count cannot be read off from its value at one parameter. The encoding keeps the other kinks away from the target, and we show that the remaining contributions near the target are polynomials of degree less than the number $n$ of buyers. The $n$-th finite difference near the target therefore removes these polynomials and leaves only the contribution of the jump at the target.

Finally, we approximate this difference from the optimal revenues at $n+1$ values of $R$, and rounding gives the number of solutions.

\textit{5.~Hardness of constructing a mechanism: reading the answer from the allocations at one report profile.} That the optimal revenue is hard to compute does not mean that an optimal mechanism is hard to construct: even if the allocations of some optimal mechanism could be computed efficiently at every report profile, computing its revenue from them would still require an expectation over exponentially many report profiles, which need not be efficient. The hardness of construction therefore needs a separate argument: we turn the answer to the hard problem into a difference of allocations that can be observed at one report profile.

To this end, we choose a group of single-item opponents for which, in the limit, the remaining supplies of the two items are exactly equal, and break this tie by small perturbations. After the perturbation, which item has more remaining supply depends on whether a quantity $\cJ$ exceeds a given threshold, provided that the two differ by a known margin. This quantity is produced by the same capacity recursion and also has kinks at subset sums. So if such threshold comparisons can be answered, binary search approximates $\cJ$, and finite differences recover the number of solutions.

The bias of the average capacities, however, does not yet let us read the answer from one report profile, because different menus may favor different items. Here the high value $H$ is used again: it brings an additional revenue that is concave and penalizes such disagreement, and under our priors every relevant menu of an optimal mechanism favors the item that the average capacities favor (\Cref{fig:gamma}).

Finally, let two buyers report a bundle type and the other buyers form the group of single-item opponents above. The allocations of these two buyers both show the same direction, and supply guarantees an explicit positive lower bound on their combined bias. So the four allocation probabilities of the two buyers for the two items decide the comparison. This makes the reduction work for any optimal mechanism.

\subsection{Related work}\label{sec:related}

Explicit two-item solutions rely on analytic properties of the distribution~\cite{GK18b,TSN19,TW17}, and some independent priors have no finite optimal menu~\cite{DDT13}. With more items, the Straight-Jacket Auction is optimal for up to six i.i.d.\ values uniform on $[0,1]$~\cite{GK18a}, pure bundling is optimal for any number of i.i.d.\ exponential values~\cite{Gia15}, and other bundling results require distributional conditions~\cite{MV06,MHJ15,HH21} or a sufficiently large common shift~\cite{DDT17,MT26}. For two independent distributions with positive, nondecreasing, continuously differentiable densities, raising the lower end of at least one support makes a deterministic mechanism globally optimal~\cite{Thi26}.

Duality has also been applied to DSIC directly~\cite{CDW16,Zuo17,JPW26}. For a single item with correlated values and an explicitly given joint distribution, an optimal randomized DSIC mechanism is computed by a linear program of size polynomial in the support of the joint distribution~\cite{DFK11}, whereas computing an optimal deterministic auction is NP-hard for three or more buyers~\cite{PP11}. Our input is only the marginal distribution, and the support of the joint distribution is exponentially large.

The symmetry algorithms of Daskalakis and Weinberg explicitly cover DSIC and ex-post IR mechanisms~\cite[full version, Remark~1 and Appendices~D and~K]{DW12}, and the fixed-parameter algorithms above come from these methods. Symmetric priors do not force symmetric deterministic optima~\cite{BNR18}. For a single additive buyer with independent items, exact optimization over deterministic partition mechanisms is strongly NP-hard~\cite{Rub16}.

Pricing the grand bundle is already $\sharpP$-hard~\cite{DDT12}, and item pricing and lottery pricing for a unit-demand buyer are hard as well~\cite{CDPSY18,CDOPSY22}. All our hardness results reduce from counting subset-sum solutions~\cite{FH09,HK16} through digit encodings (\Cref{lem:markers,lem:blocks}); the encoding in the single-buyer reduction is related to~\cite[Section~6]{CDPSY18}.

\section{Model and preliminaries}\label{sec:model}

\subsection{Model}\label{sec:model-def}

A seller has $m$ indivisible items, indexed by $[m]$, where $[r]=\{1,\dots,r\}$ for a positive integer $r$, and faces $n$ additive buyers. All $nm$ values are independent draws from a common finite distribution $F$ with support $\supp(F)=\{a_1<\dots<a_K\}$, where $a_1>0$, and probability mass function $f$. A buyer's type is a vector $\bm t\in\cT:=\supp(F)^m$, where $t_j$ is her value for item~$j$; it has probability $f(\bm t)=\prod_jf(t_j)$, where we use $f$ also for the probability of a type. A profile of reports is $\bm v=(\bm v_1,\dots,\bm v_n)\in\cT^n$, and $\bm v_{-i}$ denotes the reports of the buyers other than~$i$.

By the revelation principle, we consider direct mechanisms. A mechanism $M=(\bm x,\bm s)$ assigns to every profile $\bm v$ allocation probabilities $x_{ij}(\bm v)\in[0,1]$ and expected payments $s_i(\bm v)\in\R$ that satisfy the supply constraint
\begin{equation}\label{eq:supply}
\sum_{i=1}^n x_{ij}(\bm v)\le 1\qquad(j\in[m],\ \bm v\in\cT^n).
\end{equation}
Probabilities and expectations are over the internal randomization of the mechanism, with the reports fixed. Since the buyers are additive, a randomized allocation affects their utilities only through these marginal probabilities, and every $\bm x$ is realized by giving each item~$j$ to buyer~$i$ with probability $x_{ij}(\bm v)$. The mechanism is \emph{deterministic} if every $x_{ij}(\bm v)$ lies in $\{0,1\}$. We write $\bm x_i(\bm v)=(x_{i1}(\bm v),\dots,x_{im}(\bm v))$.

Let $u_i(\bm t;\bm v_{-i})=\bm t\cdot\bm x_i(\bm t,\bm v_{-i})-s_i(\bm t,\bm v_{-i})$ be the expected utility of buyer~$i$ of type $\bm t$ who reports truthfully when the others report $\bm v_{-i}$. The mechanism is DSIC and ex-post IR if, for all $i$, $\bm v_{-i}$ and $\bm t,\bm t'\in\cT$,
\begin{equation}\label{eq:dsic}
u_i(\bm t';\bm v_{-i})\ \ge\ \bm t'\cdot\bm x_i(\bm t,\bm v_{-i})-s_i(\bm t,\bm v_{-i}),\qquad u_i(\bm t;\bm v_{-i})\ \ge\ 0.
\end{equation}
The first inequality is the \emph{IC constraint} of the pair $(\bm t',\bm t)$: type $\bm t'$ of buyer~$i$ does not gain by reporting~$\bm t$. The second is the \emph{IR constraint} of~$\bm t$.

The expected revenue of $M$ is $\Rev_F(M)=\E_{\bm v}\sum_is_i(\bm v)$, where $\bm v_1,\dots,\bm v_n$ are independent with probability mass function $f$; we write $\Rev(M)$ when $F$ is clear. For a distribution $G$ on $\supp(F)$, $\Rev_G(M)$ is defined in the same way, with the values drawn from $G$. Let
\[
\OPT(F;n,m)=\sup_M\Rev_F(M),
\]
where the supremum is over all mechanisms that satisfy~\eqref{eq:supply} and~\eqref{eq:dsic}. The supremum is attained; to see this, we first make the IR constraint of the lowest type $a_1\one$ bind, where $\one$ is the all-ones vector.

\begin{lemma}[Binding IR at the lowest type]\label{lem:norm}
Let $M$ be DSIC and ex-post IR, and fix a buyer $i$ and reports $\bm v_{-i}$. Adding the utility $u_i(a_1\one;\bm v_{-i})\ge0$ of the lowest type to every payment $s_i(\bm t,\bm v_{-i})$ preserves~\eqref{eq:supply} and~\eqref{eq:dsic} and makes the IR constraint of $a_1\one$ bind. If this constraint binds, then for every $\bm t\in\cT$, with $\bm x=\bm x_i(\bm t,\bm v_{-i})$ and $s=s_i(\bm t,\bm v_{-i})$,
\begin{equation}\label{eq:bounded}
0\ \le\ a_1\|\bm x\|_1\ \le\ s\ \le\ \bm t\cdot\bm x\ \le\ m\,a_K,
\end{equation}
where $\|\cdot\|_1$ is the $\ell_1$ norm.
\end{lemma}

\begin{proof}
Write $u(\bm t)=u_i(\bm t;\bm v_{-i})$. Type $\bm t$ does not gain by reporting $a_1\one$, so $u(\bm t)\ge u(a_1\one)+(\bm t-a_1\one)\cdot\bm x_i(a_1\one,\bm v_{-i})\ge u(a_1\one)$, and $u(a_1\one)\ge0$ by IR. Adding $u(a_1\one)$ to all payments $s_i(\bm t,\bm v_{-i})$ changes neither~\eqref{eq:supply} nor any IC constraint, and turns the utilities into $u(\bm t)-u(a_1\one)\ge0$, which vanishes at $a_1\one$. If the IR constraint of $a_1\one$ binds, then $s\ge a_1\one\cdot\bm x=a_1\|\bm x\|_1$ because $a_1\one$ does not gain by reporting $\bm t$, $s\le\bm t\cdot\bm x$ because $u(\bm t)\ge0$, and $\bm t\cdot\bm x\le ma_K$ because $\bm t\le a_K\one$ and $\bm x\le\one$.
\end{proof}

Applying \Cref{lem:norm} to every $i$ and $\bm v_{-i}$ turns every DSIC, ex-post IR mechanism into one with no smaller revenue in which the IR constraint of $a_1\one$ binds for every $i$ and $\bm v_{-i}$. By~\eqref{eq:bounded}, the mechanisms with this property form a polytope, on which $\Rev_F$ is linear; hence $\OPT(F;n,m)$ is attained, and it is rational when the support points and the probabilities of $F$ are rational. Since every profile of reports has positive probability, every optimal mechanism has this property: otherwise the operation of \Cref{lem:norm} would strictly increase its revenue.

With a single buyer, we drop the index~$i$: a mechanism is a pair $(\bm x,s)$ of maps on $\cT$, and DSIC and ex-post IR mean that it satisfies all IC and IR constraints; we then call it IC and IR. Apart from general notation such as point masses and indicators, the symbols introduced in each of \Cref{sec:A,sec:B,sec:two} are local to that section and its appendix.

\subsection{Flows and virtual values}\label{sec:flows}

Following Cai, Devanur and Weinberg~\cite{CDW16}, we bound the revenue of DSIC, ex-post IR mechanisms using nonnegative combinations of the IC and IR constraints.

Fix a buyer $i$ and reports $\bm v_{-i}$ of the others, and abbreviate $\bm x(\bm t)=\bm x_i(\bm t,\bm v_{-i})$, $s(\bm t)=s_i(\bm t,\bm v_{-i})$ and $u(\bm t)=u_i(\bm t;\bm v_{-i})$. Since $\bm v_i$ has probability mass function $f$ and is independent of $\bm v_{-i}$, the expected payment of buyer~$i$ given $\bm v_{-i}$ is
\[
\sum_{\bm t}f(\bm t)\,s(\bm t)=\sum_{\bm t}f(\bm t)\,\bm t\cdot\bm x(\bm t)-\sum_{\bm t}f(\bm t)\,u(\bm t).
\]
The utility of type $\bm t'$ from reporting $\bm t$ is $\bm t'\cdot\bm x(\bm t)-s(\bm t)=u(\bm t)+(\bm t'-\bm t)\cdot\bm x(\bm t)$, so the IC constraint of $(\bm t',\bm t)$ reads $u(\bm t')-u(\bm t)\ge(\bm t'-\bm t)\cdot\bm x(\bm t)$, and the IR constraint of $\bm t$ reads $u(\bm t)\ge0$. Give each IC constraint $(\bm t',\bm t)$ a weight $\lambda(\bm t',\bm t)\ge0$ and each IR constraint a weight $\lambda(\bm t,\bot)\ge0$, and add:
\[
\sum_{\bm t',\bm t}\lambda(\bm t',\bm t)\big(u(\bm t')-u(\bm t)\big)+\sum_{\bm t}\lambda(\bm t,\bot)\,u(\bm t)\ \ge\ \sum_{\bm t',\bm t}\lambda(\bm t',\bm t)\,(\bm t'-\bm t)\cdot\bm x(\bm t).
\]
If the weights satisfy
\begin{equation}\label{eq:conserv}
\sum_{\bm t''}\lambda(\bm t,\bm t'')+\lambda(\bm t,\bot)=f(\bm t)+\sum_{\bm t'}\lambda(\bm t',\bm t)\qquad(\bm t\in\cT),
\end{equation}
then the left side of the summed inequality is the expected utility $\sum_{\bm t}f(\bm t)u(\bm t)$, so
\[
\sum_{\bm t}f(\bm t)\,u(\bm t)\ \ge\ \sum_{\bm t',\bm t}\lambda(\bm t',\bm t)\,(\bm t'-\bm t)\cdot\bm x(\bm t).
\]
Substituting this lower bound into the expected payment and collecting the terms in each $x_j(\bm t)$ gives
\begin{equation}\label{eq:certbound}
\sum_{\bm t}f(\bm t)\,s(\bm t)\ \le\ \sum_{\bm t}f(\bm t)\sum_j\Phi_j(\bm t)\,x_j(\bm t),
\end{equation}
where
\begin{equation}\label{eq:Phidef}
\Phi_j(\bm t)=t_j-\frac1{f(\bm t)}\sum_{\bm t'}\lambda(\bm t',\bm t)\,(t'_j-t_j)\qquad(\bm t\in\cT,\ j\in[m]).
\end{equation}

Condition~\eqref{eq:conserv} is flow conservation: read $\lambda(\bm t',\bm t)$ as the flow along an edge $\bm t'\to\bm t$ and $\lambda(\bm t,\bot)$ as the flow from $\bm t$ to a sink $\bot$; every type $\bm t$ is a source of $f(\bm t)$ units, and all flow ends at~$\bot$. Nonnegative weights that satisfy~\eqref{eq:conserv} are called a \emph{flow}; $\Phi_j(\bm t)$ is the \emph{virtual value} of item~$j$ at type $\bm t$ under the flow, and $\sum_{\bm t}f(\bm t)\sum_j\Phi_j(\bm t)x_j(\bm t)$ is the \emph{expected virtual surplus}. Flows and virtual values depend only on $\cT$ and $f$, not on the mechanism. For example, with a single item the types are the values $a_1<\dots<a_K$. Letting each $a_{k+1}$ send $\sum_{l>k}f(a_l)$, the probability that the value exceeds $a_k$, to $a_k$, and letting $a_1$ send $1$ to $\bot$, gives a flow whose virtual values $\Phi(a_k)=a_k-(a_{k+1}-a_k)\sum_{l>k}f(a_l)/f(a_k)$, for $k<K$, and $\Phi(a_K)=a_K$, are Myerson's virtual values in discrete form~\cite{Mye81}.

\begin{lemma}[Weak duality~\cite{CDW16}]\label{lem:duality}
Let $\lambda$ be a flow with virtual values $\Phi$. For every mechanism $M$, buyer $i$ and reports $\bm v_{-i}$, with the abbreviations above,
\begin{equation}\label{eq:slackid}
\begin{aligned}
\sum_{\bm t}f(\bm t)\sum_j\Phi_j(\bm t)\,x_j(\bm t)-\sum_{\bm t}f(\bm t)\,s(\bm t)
={}&\sum_{\bm t',\bm t}\lambda(\bm t',\bm t)\big(u(\bm t')-u(\bm t)-(\bm t'-\bm t)\cdot\bm x(\bm t)\big)\\
&+\sum_{\bm t}\lambda(\bm t,\bot)\,u(\bm t).
\end{aligned}
\end{equation}
In particular, if $M$ is DSIC and ex-post IR, then \eqref{eq:certbound} holds, with equality if and only if every constraint of positive weight is tight.
\end{lemma}

\begin{proof}
In the weighted sum of the constraints, the left side minus the right side is the right side of~\eqref{eq:slackid}, and by~\eqref{eq:conserv} the left side is $\sum_{\bm t}f(\bm t)u(\bm t)$. Substituting into the expected payment as above gives~\eqref{eq:slackid}. If $M$ is DSIC and ex-post IR, every term on the right of~\eqref{eq:slackid} is nonnegative, and the sum vanishes if and only if every constraint of positive weight is tight.
\end{proof}

\begin{corollary}[Bounded allocations]\label{cor:cap}
Let $\lambda$ be a flow with virtual values $\Phi$, and let $(\bm x,s)$ be IC and IR with $0\le x_j(\bm t)\le y_j$ for all $\bm t$ and $j$. Then, with $\pos{z}=\max\{z,0\}$,
\begin{equation}\label{eq:certcap}
\sum_{\bm t}f(\bm t)\,s(\bm t)\ \le\ \sum_jy_j\sum_{\bm t}f(\bm t)\pos{\Phi_j(\bm t)},
\end{equation}
with equality if and only if every constraint of positive weight is tight, $x_j(\bm t)=y_j$ whenever $\Phi_j(\bm t)>0$, and $x_j(\bm t)=0$ whenever $\Phi_j(\bm t)<0$.
\end{corollary}

\begin{proof}
Since $0\le x_j(\bm t)\le y_j$, we have $\Phi_j(\bm t)x_j(\bm t)\le\pos{\Phi_j(\bm t)}\,y_j$ for all $\bm t$ and $j$; as $f(\bm t)>0$, the expected virtual surplus is therefore at most the right side of~\eqref{eq:certcap}, with equality if and only if the last two conditions hold. Combine this with \Cref{lem:duality}.
\end{proof}

\begin{corollary}[Several buyers]\label{cor:multi}
Let $M$ be DSIC and ex-post IR. Choose a flow for every buyer $i$ and reports $\bm v_{-i}$, and let $\Phi_{ij}(\bm v)$ be the virtual value of item~$j$ at type $\bm v_i$ under the flow chosen for $(i,\bm v_{-i})$. Then
\begin{equation}\label{eq:multibound}
\Rev(M)\le\E_{\bm v}\sum_{j=1}^m\max\{0,\Phi_{1j}(\bm v),\dots,\Phi_{nj}(\bm v)\}.
\end{equation}
Equality holds if, for every $i$ and $\bm v_{-i}$, every constraint of positive weight in the flow chosen for $(i,\bm v_{-i})$ is tight, and $\sum_i\Phi_{ij}(\bm v)x_{ij}(\bm v)=\max\{0,\max_i\Phi_{ij}(\bm v)\}$ for every $\bm v$ and $j$.
\end{corollary}

\begin{proof}
By \Cref{lem:duality}, taking the expectation of~\eqref{eq:certbound} over $\bm v_{-i}$ and summing over $i$ gives
\[
\Rev(M)\le\E_{\bm v}\sum_j\sum_i\Phi_{ij}(\bm v)\,x_{ij}(\bm v).
\]
Since $x_{ij}\ge0$ and $\sum_ix_{ij}\le1$, we have $\sum_i\Phi_{ij}(\bm v)x_{ij}(\bm v)\le\max\{0,\max_i\Phi_{ij}(\bm v)\}$ for every $\bm v$ and $j$.
\end{proof}

\subsection{Computational conventions}\label{sec:conventions}

Let $\bit(F)$ be the total binary length of the input list of~$F$ (\Cref{sec:results}); polynomial time means time polynomial in $n+m+\bit(F)$. Computing $\OPT(F;n,m)$ means outputting its exact value, which is rational (\Cref{sec:model-def}). Hardness is with respect to deterministic polynomial-time Turing reductions; for a problem with several correct outputs, such as approximating $\OPT(F;n,m)$, the reduction must succeed whichever correct outputs the oracle returns. We call constructing an optimal mechanism \emph{$\sharpP$-hard} if every function in $\sharpP$ can be computed by a polynomial-time algorithm with access to an oracle of the following kind, whichever such oracle it is given: for each instance, the oracle fixes an optimal mechanism and, given a report profile $\bm v$ and a rational $e>0$, returns the allocation probabilities and payments of this mechanism at $\bm v$, each rounded down to a multiple of $e$. A deterministic polynomial-time construction algorithm provides such an oracle, whatever description of the mechanism it outputs.

All our hardness results reduce from the counting version of \textsc{SubsetSum}: given positive integers $w_1,\dots,w_r$ and a target $T$ in binary, compute the number of sets $A\subseteq[r]$ with $\sum_{i\in A}w_i=T$. This problem is $\sharpP$-complete~\cite{FH09,HK16}.

\section{Optimal auctions for two-point values}\label{sec:A}

With two-point values, an item that some buyer reports high can always be sold at the high value, so the difficulty lies only in the items that every buyer reports low. We allocate these items according to the exclusive counts of the buyers (defined in \Cref{sec:A-idea}), as in a Myerson auction in the count dimension, and prove with an explicit flow that the resulting auction is optimal among all randomized DSIC, ex-post IR mechanisms.

Throughout this section, $F=p\delta_a+q\delta_b$, where $\delta_a$ is the distribution concentrated at $a$, with $0<a<b$, $0<p<1$ and $q=1-p$, and $\Delta=b-a$. Every value is either high ($b$) or low ($a$), so a type is determined by the set of items that the buyer values high. We identify the type of buyer $i$ with this set $S_i\subseteq[m]$: its value vector is $\bm t(S_i)$ with $t_j(S_i)=a+\Delta\ind{j\in S_i}$, where $\ind{\cdot}$ is the indicator function, and its probability is $f(S_i)=p^{m-|S_i|}q^{|S_i|}$. We write $\bm S=(S_1,\dots,S_n)$ for a profile of reports.

\subsection{All-low items and exclusive counts}\label{sec:A-idea}

At a profile $\bm S$, the items that no buyer reports high are the \emph{all-low} items; we write
\[
L=[m]\setminus\bigcup_iS_i,\qquad c=|L|.
\]
Let $E_i=S_i\setminus\bigcup_{\ell\ne i}S_\ell$ be the set of items that only buyer $i$ reports high; the sets $E_i$ are pairwise disjoint. We call $e_i=|E_i|$ the \emph{exclusive count} of buyer $i$, and write $e^*=\max_ie_i$. For $c\ge1$ and $r\ge0$, let
\begin{equation}\label{eq:phiA}
\phi_{c,r}=a-\Delta\,\frac{\Pr[R>r]}{c\,\Pr[R=r]},\qquad R\sim\mathrm{Bin}(c+r,q),
\end{equation}
be the virtual value of each all-low item for a buyer with exclusive count $r$ when there are $c$ all-low items; when there are none, we put $\phi_{0,r}=0$. In other words, $c\,\phi_{c,r}$ is the discrete Myerson virtual value at $R=r$ of the value $a(c+r)+\Delta R$ of $c+r$ items, of which the buyer values $R$ high, minus the value $br$ of her $r$ high items; the flow in \Cref{sec:A-proof} makes this precise.

\begin{lemma}\label{lem:monoA}
For fixed $c\ge1$, $\phi_{c,r}$ is strictly increasing in $r$. For fixed $d\ge2$, $\phi_{d-k,k}$ is strictly increasing in $k\in\{0,\dots,d-1\}$.
\end{lemma}

The proof is in \Cref{app:A-mono}. By \Cref{lem:monoA}, the all-low items should go to the buyer with the largest exclusive count, provided that her virtual value is positive. Let $\theta_0=0$ and, for $d\ge1$,
\begin{equation}\label{eq:thetaA}
\theta_d=\min\Big(\{k\in\{0,\dots,d-1\}:\phi_{d-k,k}>0\}\cup\{d\}\Big).
\end{equation}
By \Cref{lem:monoA}, for $0\le k<d$,
\begin{equation}\label{eq:thresholdA}
k\ge\theta_d\iff\phi_{d-k,k}>0 .
\end{equation}
Thus $\theta_d$ is a reserve in the count dimension: the all-low items are worth selling to a buyer only if her exclusive count reaches $\theta_d$.

\subsection{The optimal auction}\label{sec:A-auction}

Fix the priority order $1<2<\dots<n$. At a profile $\bm S$, rank the buyers by decreasing exclusive count, breaking ties by priority, and let $i^*$ be the first buyer.
\begin{itemize}[nosep]
\item \emph{Allocation.} Each item that some buyer reports high goes to the first buyer in the ranking who reports it high. If $e_{i^*}\ge\theta_{c+e_{i^*}}$, every item of $L$ goes to $i^*$; otherwise the items of $L$ are kept. Let $A_i$ be the set allocated to buyer $i$.
\item \emph{Payments.} For buyer $i$, let $d=c+e_i$, the number of items that no opponent reports high, and
\begin{equation}\label{eq:ThetaA}
\Theta_i=\max\Big(\{\theta_d\}\cup\{e_\ell+\ind{\ell<i}:\ell\ne i\}\Big).
\end{equation}
Buyer $i$ pays $s_i=\sum_{j\in A_i}t_j(S_i)-\Delta\pos{e_i-\Theta_i}$.
\end{itemize}
Here $\Theta_i$ is the least count that reaches the reserve $\theta_d$ and ranks buyer $i$ ahead of every opponent. Hence $e_i\ge\Theta_i$ if and only if $i=i^*$ and $e_{i^*}\ge\theta_{c+e_{i^*}}$; when $c\ge1$, this is the event that buyer $i$ receives the all-low items (\Cref{app:A-menu}). She pays her reported value for each item she receives, minus the discount $\Delta(e_i-\Theta_i)$ when $e_i\ge\Theta_i$.

We now describe the options that the auction offers to one buyer when the reports of the others are fixed; they show that the auction is truthful. Fix buyer $i$ and the reports of the others. Let $H$ be the set of items that some opponent reports high, $D=[m]\setminus H$ and $d=|D|$. Each report of buyer $i$ can be written uniquely as $S=Q\cup T$ with $Q=S\cap H$ and $T=S\cap D$; let $r=|T|$. At the resulting profile, $E_i=T$ and $L=D\setminus T$. For $\ell\ne i$, let $C_\ell$ be the set of items that, among the buyers other than $i$, only $\ell$ reports high. When $i$ reports $Q\cup T$, the exclusive count of opponent $\ell$ is $e_\ell(Q)=|C_\ell\setminus Q|$, so the threshold \eqref{eq:ThetaA} depends only on $Q$; we denote it by $\Theta(Q)$.

\begin{lemma}[Bundle prices]\label{lem:menuA}
Let buyer $i$ report $Q\cup T$ with $r=|T|$.
\begin{enumerate}[label=(\arabic*),nosep]
\item If $r<\Theta(Q)$, buyer $i$ receives $T$ together with the items of $Q$ that no buyer ranked before her reports high, and pays $b$ for each of them.
\item If $r\ge\Theta(Q)$, buyer $i$ receives $D\cup Q$ and pays $ad+b|Q|+\Delta\,\Theta(Q)$.
\end{enumerate}
Moreover, for all $Q,Q'\subseteq H$,
\begin{equation}\label{eq:LipA}
\Theta(Q)\le\Theta(Q')+|Q'\setminus Q| .
\end{equation}
\end{lemma}

Thus, on $D$, buyer $i$ faces bundle prices: she buys her high items at $b$ each, or the whole of $D$ at $ad+\Delta\,\Theta(Q)$, where the threshold $\Theta(Q)$ combines the reserve with the exclusive counts of her opponents; items of $H$ she can buy only at $b$ each. The proof is in \Cref{app:A-menu}. The description of $\Theta_i$, \Cref{lem:menuA} and their proofs use no property of the reserves, so they hold for any integer reserves $\theta_0,\dots,\theta_m$; only the last step of the proof of \Cref{thm:A} uses \eqref{eq:thetaA}.

\begin{theorem}[Theorem A]\label{thm:A}
For all $n,m\ge1$, the auction above is deterministic, feasible, DSIC and ex-post IR, has nonnegative payments, and its expected revenue is
\begin{equation}\label{eq:revA}
\OPT(F;n,m)=mb(1-p^n)+\E_{\bm S}\Big[c\,\pos{\phi_{c,e^*}}\Big].
\end{equation}
It is optimal among all randomized DSIC, ex-post IR mechanisms. If $a$, $b$ and $p$ are rational, the outcome of the auction at any profile and the value \eqref{eq:revA} can be computed in time polynomial in $n$, $m$ and the bit length of $(a,b,p)$.
\end{theorem}

For $n=1$, we have $H=\emptyset$ and $\Theta=\theta_m$: a type with $r$ high items buys the grand bundle at price $ma+\Delta\theta_m$ if $r\ge\theta_m$, and otherwise buys her $r$ high items at $b$ each. For $m=2$, the optimal DSIC revenue was determined by Yao~\cite{Yao17}; \Cref{thm:A} gives it for every $m$. Choosing the priority order uniformly at random before the reports yields an anonymous randomized optimal auction, since every fixed order is DSIC and ex-post IR and all orders have the same expected revenue.

\subsection{Proof of Theorem A}\label{sec:A-proof}

The proof has three steps: a flow gives an upper bound on the revenue of every mechanism (\Cref{lem:upperA}); the bundle prices of \Cref{lem:menuA} show that the auction is DSIC and ex-post IR (\Cref{prop:dsicA}); and the auction attains the bound. We then discuss the computation.

We begin with the upper bound. Fix buyer $i$ and the reports of the others, and use the notation of \Cref{sec:A-auction}. Write $f_H(Q)=q^{|Q|}p^{|H|-|Q|}$, so that $f(Q\cup T)=f_H(Q)\,q^rp^{d-r}$, and let $R\sim\mathrm{Bin}(d,q)$. The flow uses only two kinds of constraints: the downward IC constraints ``$S\cup\{j\}$ does not gain by reporting $S$'', and the IR constraint of $\emptyset$.
\begin{itemize}[nosep]
\item On $D$, it is the one-dimensional Myerson flow on counts: for each $Q\subseteq H$, the mass $f_H(Q)\Pr[R\ge r]$ of the types $Q\cup T'$ with $|T'|\ge r$ flows from count $r$ to count $r-1$, spread evenly over the $\binom dr$ sets $Q\cup T$ with $|T|=r$ and their $r$ edges. That is, for $|T|=r\ge1$ and $j\in T$,
\begin{equation}\label{eq:flowA}
\lambda(Q\cup T,\,Q\cup T\setminus\{j\})=\frac{f_H(Q)\Pr[R\ge r]}{r\binom dr}.
\end{equation}
\item On $H$, the flow only carries mass downward, from the reports with $T=\emptyset$: in decreasing order of $|Q|$, each report $Q\subseteq H$ with $Q\ne\emptyset$ splits all the mass that it holds, its own probability plus all the flow entering it, evenly among the $|Q|$ edges that delete one element of $Q$.
\item All other edges have weight zero, and the IR constraint of $\emptyset$ absorbs $f(\emptyset)$ plus all the flow entering $\emptyset$.
\end{itemize}
All weights are nonnegative; conservation is verified in \Cref{app:A-flow}. By \eqref{eq:Phidef}, an item that buyer $i$ reports high has virtual value $b$; for $j\in D\setminus T$, \eqref{eq:flowA} gives $\Phi_j(S)=\phi_{c,r}$ with $c=d-r$; and for $j\in H\setminus Q$, $\Phi_j(S)\le a$.

\begin{lemma}[Upper bound]\label{lem:upperA}
Every DSIC, ex-post IR mechanism $M$, randomized or not, satisfies
\[
\Rev(M)\le mb(1-p^n)+\E_{\bm S}\Big[c\,\pos{\phi_{c,e^*}}\Big].
\]
\end{lemma}

\begin{proof}
Apply \Cref{cor:multi} with the flow above for every buyer and every profile of the others' reports. Fix a profile $\bm S$ and an item $j$. If $j\notin L$, a buyer who reports $j$ high has virtual value $b$, and no buyer has a larger one. If $j\in L$, then for every buyer $i$ the item $j$ lies in $D_i\setminus T_i$, where $D_i$ and $T_i$ are the sets $D$ and $T$ for buyer $i$; since $|D_i\setminus T_i|=c$ and $|T_i|=e_i$, we get $\Phi_{ij}=\phi_{c,e_i}$; by \Cref{lem:monoA}, $\max\{0,\max_i\Phi_{ij}\}=\pos{\phi_{c,e^*}}$. Summing over items, the right side of \eqref{eq:multibound} at $\bm S$ is $b(m-c)+c\pos{\phi_{c,e^*}}$, also when $c=0$. Each item is all-low with probability $p^n$, so $\E[b(m-c)]=mb(1-p^n)$.
\end{proof}

Next, the bundle prices of \Cref{lem:menuA} show that the auction is truthful, whatever the reserves.

\begin{proposition}\label{prop:dsicA}
For every choice of integer reserves with $\theta_0=0$ and $0\le\theta_d\le d$, the auction is deterministic, feasible, DSIC and ex-post IR, with nonnegative payments.
\end{proposition}

\begin{proof}
Feasibility and determinism are immediate from the definition. Fix buyer $i$, the reports of the others, and her true type $S^\circ=Q^\circ\cup T^\circ$ with $r^\circ=|T^\circ|$. By \Cref{lem:menuA}, her truthful utility is $u^\circ=\Delta\pos{r^\circ-\Theta(Q^\circ)}\ge0$, which gives ex-post IR; every payment in \Cref{lem:menuA} is nonnegative. Consider a report $Q'\cup T'$. If it yields outcome~(1) of \Cref{lem:menuA}, every item she receives is one that she reports high and pays $b$ for, so her true utility is at most $0\le u^\circ$. If it yields outcome~(2), she receives $D\cup Q'$, worth $ad+\Delta r^\circ+a|Q'|+\Delta|Q'\cap Q^\circ|$ to her, for $ad+b|Q'|+\Delta\,\Theta(Q')$, so her true utility is
\[
\Delta\big(r^\circ-|Q'\setminus Q^\circ|-\Theta(Q')\big)\le\Delta\big(r^\circ-\Theta(Q^\circ)\big)\le u^\circ,
\]
where the first inequality is \eqref{eq:LipA} with $Q=Q^\circ$.
\end{proof}

Finally, we show that with the reserves \eqref{eq:thetaA} the auction attains the bound. It is DSIC and ex-post IR by \Cref{prop:dsicA}, and \Cref{lem:upperA} bounds the revenue of every mechanism by the right side of \eqref{eq:revA}; it remains to check the two equality conditions of \Cref{cor:multi}.
\begin{itemize}[nosep]
\item \emph{Every constraint of positive weight is tight.} An edge \eqref{eq:flowA} joins $Q\cup T\cup\{j\}$ and $Q\cup T$ with $j\in D\setminus T$ and $r=|T|$. By \Cref{lem:menuA}, the utilities of these reports are $\Delta\pos{r+1-\Theta(Q)}$ and $\Delta\pos{r-\Theta(Q)}$; since $\Theta(Q)$ is an integer, they differ by $\Delta\ind{r\ge\Theta(Q)}$, which equals $\Delta\,x_j(Q\cup T)$ because $j$ is all-low at the report $Q\cup T$. An edge on $H$ joins two reports with utility zero, and at the lower report the item goes to an opponent who reports it high. The utility of $\emptyset$ is zero.
\item \emph{The allocation attains the largest virtual value.} An item that some buyer reports high goes to a buyer with virtual value $b$. For the all-low items, $i^*$ has the largest virtual value $\phi_{c,e^*}$, and by \eqref{eq:thresholdA} they go to $i^*$ exactly when it is positive.
\end{itemize}
Hence the revenue of the auction equals the bound, which proves \eqref{eq:revA} and optimality.

As for the computation, the outcome of the auction at a profile can be read off from the reports once the counts $e_i$ and $c$, the ranking and the reserves $\theta_d$ for $d\le m$ are known. Since items are independent, the expectation in \eqref{eq:revA} is a sum of polynomially many terms, each computable in polynomial time (\Cref{app:A-compute}).

\section{Hardness with a single buyer}\label{sec:B}

We begin with the formal statement of \Cref{thm:introB}.

\begin{theorem}[Theorem B]\label{thm:B}
\leavevmode
\begin{enumerate}[label=(\roman*),nosep]
\item Computing $\OPT(F;1,m)$ exactly is $\sharpP$-hard, even when $F$ has exactly $2m+2$ support points, all in $(0,1]$.
\item Constructing an optimal mechanism is $\sharpP$-hard for the same priors.
\end{enumerate}
\end{theorem}

To prove it, we first construct a two-point prior under which two grand-bundle prices earn the same revenue and all their mixtures are optimal, and then add intermediate values of small probability, which encode a $\sharpP$-hard count in the difference between the revenues of the two prices. The main part of the proof shows that after this perturbation no mechanism can exploit the new types.

Throughout this section, there is a single additive buyer and there are $m\ge2$ items, and we use the single-buyer notation of \Cref{sec:model-def}: a mechanism is a pair $(\bm x,s)$ of maps on $\cT=\supp(F)^m$, and $u(\bm t)=\bm t\cdot\bm x(\bm t)-s(\bm t)$.

\subsection{Encoding hardness in grand-bundle prices}\label{sec:B-idea}

If the grand bundle is posted at price $T_0$, the buyer buys it exactly when $\sum_jt_j\ge T_0$. If the value of each item is drawn independently and uniformly from a list of positive integers $b_1,\dots,b_N$, the revenue is therefore $T_0\,\cZ(T_0)/N^m$, where
\[
\cZ(T_0)=\#\Big\{(i_1,\dots,i_m)\in[N]^m:\ \sum_{j=1}^mb_{i_j}\ge T_0\Big\}
\]
is the number of index tuples whose values $b_{i_j}$ sum to at least $T_0$. The next lemma, proved in \Cref{app:B-tail}, shows that two values of $\cZ$ determine the number of subset-sum solutions.

\begin{lemma}[Tuples whose sum reaches a threshold]\label{lem:tailB}
Let $m\ge2$, let $w_1,\dots,w_m$ and $T$ be positive integers with $T\le\sum_iw_i$, and let $\cN$ be the number of sets $A\subseteq[m]$ with $w(A)=T$, where $w(A)=\sum_{i\in A}w_i$. Put $N=2m$ and $W=1+m\sum_iw_i$; let $b_1,\dots,b_N$ be the $N$ numbers
\[
b_{i,0}=W(m+1)^{i-1},\qquad b_{i,1}=W(m+1)^{i-1}+w_i\qquad(i\in[m])
\]
in any order; and put $R=W\sum_{i=1}^m(m+1)^{i-1}+T$ and $b_{\max}=\max_\ell b_\ell$. Then, for this list,
\[
\cN=\frac{\cZ(R)-\cZ(R+1)}{m!},
\]
and $R,R+1\in[1,mb_{\max}]$. All these integers have bit length polynomial in the input length.
\end{lemma}

\subsection{Constructing a prior that encodes the count}\label{sec:B-construction}

Consider first a prior with the two values $a$ and $a+1$, where $a$ is chosen so that the grand-bundle prices $ma$ and $ma+1$ earn the same revenue. For $p\in[1/3,2/3]$, let
\begin{equation}\label{eq:tieB}
F_0=p\,\delta_a+(1-p)\,\delta_{a+1},\qquad ma=p^{-m}-1,\qquad P=ma+1=p^{-m}.
\end{equation}
Write a type as $\bm t=a\one+\bm\xi$, where $\xi_j=t_j-a\in[0,1]$ is the excess of the value of item~$j$ over~$a$. At price $ma$, every type buys the grand bundle, and the revenue is $ma$. At price $P$, the buyer buys exactly when $\|\bm\xi\|_1\ge1$; under $F_0$, these are all types except the lowest type $a\one$, which has probability $p^m$, so the revenue is $(1-p^m)P=P-1=ma$.

The two prices tie, and so do their mixtures. For $z\in[0,1]$, let $M_z$ be the mechanism that posts the grand bundle at price $ma$ with probability $z$ and at price $P$ with probability $1-z$. Not buying, buying only at price $ma$, and buying at whichever price is posted yield the outcomes (allocation, payment) $(\bm 0,0)$, $(z\one,maz)$ and $(\one,P-z)$, which give type $\bm t$ the utilities $0$, $z\|\bm\xi\|_1$ and $\|\bm\xi\|_1-1+z$. Since $(\|\bm\xi\|_1-1+z)-z\|\bm\xi\|_1=(1-z)(\|\bm\xi\|_1-1)$, breaking ties in favor of the larger payment gives
\begin{equation}\label{eq:MzB}
M_z(\bm t)=\begin{cases}(z\one,\ maz),&\|\bm\xi\|_1<1,\\ (\one,\ P-z),&\|\bm\xi\|_1\ge1.\end{cases}
\end{equation}
Write $s_z(\bm t)$ for the payment of $M_z$ at $\bm t$. On the whole cube $[a,a+1]^m$, $M_z$ is IC and IR, and
\begin{equation}\label{eq:tiedrevB}
\Rev_{F_0}(M_z)=p^m\,maz+(1-p^m)(P-z)=ma\qquad(z\in[0,1]).
\end{equation}

We now perturb $F_0$ by adding $N$ intermediate values of small total probability $\eta$, whose excesses over $a$ encode a list of integers. Fix distinct positive integers $b_1,\dots,b_N$, for instance those of \Cref{lem:tailB}, write $b_{\max}=\max_\ell b_\ell$, and fix integers $T_0\in[1,mb_{\max}]$ and $c\in[0,N^m)$. The intermediate values are $a+\zeta_\ell$, where
\[
\Lambda=4m^2b_{\max},\qquad \zeta_\ell=\frac1m+\frac{b_\ell-(T_0-\frac12)/m}{\Lambda}\quad(\ell\in[N]).
\]
Each $\zeta_\ell$ is close to $1/m$, so when all $m$ values are intermediate, their excesses over $a$ sum to about $1$, near the threshold $\|\bm\xi\|_1=1$ of the price $P$ (\Cref{fig:Bconstruction}); the small offsets encode the numbers $b_\ell$: $\sum_j\zeta_{i_j}\ge1$ if and only if $\sum_jb_{i_j}\ge T_0$. The integer $c$ is the threshold of the comparison that the prior encodes: we choose $P$ so that the two prices tie exactly when $\cZ(T_0)=c+\frac12$. To this end, put
\[
\vartheta=\frac{c+\frac12}{N^m},\qquad d_0=2^{-m}-\eta^m\vartheta,
\]
where $\eta>0$ is a small rational number that depends only on $m$ and $b_{\max}$ and is given in \eqref{eq:etaB}, and let $a$, $P$ and $F_0$ be given by \eqref{eq:tieB} with $p=d_0^{1/m}$; thus $P=1/d_0$ and $a=(P-1)/m$ are rational, while $p$ may be irrational and enters only the analysis. Define
\begin{equation}\label{eq:FB}
F=\big(\tfrac12-\eta\big)\,\delta_a+\tfrac12\,\delta_{a+1}+\frac{\eta}{N}\sum_{\ell=1}^N\delta_{a+\zeta_\ell},
\end{equation}
a distribution with rational data whose support consists of $a$, $a+1$ and the $N$ intermediate values $a+\zeta_\ell$. It is a small perturbation of $F_0$: with $f_0$ the probability mass function of $F_0$,
\[
\|f-f_0\|_1:=\sum_{a'\in\supp(F)}\big|f(a')-f_0(a')\big|=2\eta .
\]
These properties, and $p\in[1/3,2/3]$, are verified in \Cref{app:B-encoding}.

\begin{figure}[t]
\centering
\begin{tikzpicture}[scale=4.2,every node/.style={font=\small}]
  \def\nuv{0.05}
  \fill[black!10] (1,0) -- (1,1) -- (0,1) -- cycle;
  \draw[black!35] (0,0) rectangle (1,1);
  \draw[->] (0,0) -- (1.12,0) node[below] {$\xi_1$};
  \draw[->] (0,0) -- (0,1.12) node[left] {$\xi_2$};
  \draw[thick] (0,1) -- node[pos=0.16,above,sloped,font=\footnotesize] {$\|\bm\xi\|_1=1$} (1,0);
  \draw[densely dotted] (0,1-\nuv) -- (1-\nuv,0);
  \node[below] at (0.5,0) {$\tfrac12$};
  \node[below] at (1,0) {$1$};
  \node[left] at (0,0.5) {$\tfrac12$};
  \node[left] at (0,1) {$1$};
  \foreach \x in {0,0.35,0.40,0.55,0.72,1} {
    \foreach \y in {0,0.35,0.40,0.55,0.72,1} {
      \pgfmathparse{(\x+\y>=1)?1:0}
      \ifnum\pgfmathresult=1
        \fill (\x,\y) circle (0.42pt);
      \else
        \draw[fill=white] (\x,\y) circle (0.42pt);
      \fi
    }
  }
  \node[below left,inner sep=1.5pt,font=\footnotesize] at (0,0) {$\bm t=a\one$};
\end{tikzpicture}
\caption{The construction of \Cref{sec:B-construction} for $m=2$, in the coordinates $\bm\xi=\bm t-a\one$ (not to scale). The price $P$ sells the grand bundle to the types on or above the line $\|\bm\xi\|_1=1$ (filled dots), and the price $ma$ sells it to all types. All vertices except the lowest type $a\one$ lie on or above the line. The types with two intermediate values lie near $(\frac12,\frac12)$, and such a type lies above the line if and only if the two corresponding numbers $b_\ell$ sum to at least $T_0$. The types below the line (open dots) satisfy $\|\bm\xi\|_1\le1-\nu$ (dotted line), as required in~\eqref{eq:sepB}.}
\label{fig:Bconstruction}
\end{figure}

Let $\tau=\cZ(T_0)/N^m$. Under $F$, $\|\bm\xi\|_1\ge1$ holds exactly when some item has value $a+1$, or all items have intermediate values $a+\zeta_{i_1},\dots,a+\zeta_{i_m}$ with $\sum_jb_{i_j}\ge T_0$. Hence $\Pr_F[\|\bm\xi\|_1<1]=2^{-m}-\eta^m\tau$, and \eqref{eq:MzB} gives
\begin{equation}\label{eq:encrevB}
\Rev_F(M_z)=ma+(1-z)\,P\eta^m(\tau-\vartheta).
\end{equation}
The choice of $P$ thus makes the two prices tie exactly when $\tau=\vartheta$. Since $\tau-\vartheta=\big(\cZ(T_0)-c-\frac12\big)/N^m\ne0$, the best mechanism in the family $\{M_z\}$ is $M_1$, the price $ma$, when $\cZ(T_0)\le c$, and $M_0$, the price $P$, when $\cZ(T_0)\ge c+1$.

Hence, if no mechanism did better than the best $M_z$, the optimal revenue would be $ma+P\eta^m\pos{\tau-\vartheta}$. If, moreover, every optimal mechanism agreed with the best $M_z$ at the lowest type, it would allocate every item to the lowest type with probability $\ind{\cZ(T_0)\le c}$. \Cref{sec:B-proof} shows that these two properties suffice to prove \Cref{thm:B}. Other mechanisms might, however, exploit the new types created by the perturbation; \Cref{sec:B-stability} proves both properties.

\subsection{No mechanism does better}\label{sec:B-stability}

We first show that every $M_z$ is optimal under the unperturbed prior $F_0$, and then show that under priors close to $F_0$ the optimal revenue is still attained in the family $\{M_z\}$. In this subsection, $p\in[1/3,2/3]$ is arbitrary, $a$, $P$ and $F_0$ are given by \eqref{eq:tieB}, and $F$ is a finite distribution on $[a,a+1]$ with positive probability at $a$ and at $a+1$; we return to the prior of \eqref{eq:FB} at the end of the subsection.

The type set $\cT=\supp(F)^m$ contains every vertex $\bm t^A=a\one+\one_A$ ($A\subseteq[m]$, where $\one_A$ is the indicator vector of $A$) of the cube $[a,a+1]^m$, that is, every type whose values lie in $\{a,a+1\}$. We identify the vertex $\bm t^A$ with the set $A$ and write, for example, $u(A)=u(\bm t^A)$ and $x_j(A)=x_j(\bm t^A)$; let $f_0(A)=p^{m-|A|}(1-p)^{|A|}$ be the probability of $\bm t^A$ under $F_0$. By \Cref{lem:norm}, it suffices to consider mechanisms in which the IR constraint of the lowest type $\bm t^\emptyset=a\one$ binds, that is, $u(\emptyset)=0$, and every optimal mechanism has this property (\Cref{sec:model-def}).

We use the flow of \Cref{sec:A-proof} on the vertices, for the prior $F_0$ and a single buyer. For $j\in A$ with $|A|=k\ge1$, the IC constraint ``$\bm t^A$ does not gain by reporting $\bm t^{A\setminus\{j\}}$'' has weight
\[
\lambda_k=\frac{\Pr_{F_0}[\|\bm\xi\|_1\ge k]}{k\binom mk},
\]
the IR constraint of $\bm t^\emptyset$ has weight $f_0(\emptyset)+m\lambda_1$, and all other constraints have weight zero. Thus the probability that at least $k$ values equal $a+1$ flows from the vertices $\bm t^A$ with $|A|=k$ to those with $|A|=k-1$, spread evenly over the $k\binom mk$ edges between them. By \eqref{eq:Phidef}, the virtual values are $\Phi_j(A)=a+1$ for $j\in A$ and $\Phi_j(A)=a-\lambda_{|A|+1}/f_0(A)$ for $j\notin A$.

\begin{lemma}\label{lem:PhiB}
$\Phi_j(\emptyset)=0$ for every $j$, and $\Phi_j(A)\ge1/4$ for every $A\ne\emptyset$ and every $j$.
\end{lemma}

The first part is the tie in another form: the choice of $ma$ in \eqref{eq:tieB} is exactly what makes the virtual values of the lowest type vanish. Write the slacks of the IC constraints above as
\[
\sigma(A,j)=u(A)-u(A\setminus\{j\})-x_j(A\setminus\{j\})\qquad(j\in A);
\]
they are nonnegative by IC. For this flow, the identity \eqref{eq:slackid} of \Cref{lem:duality} gives the following exact expression for the revenue lost relative to $ma$; it is proved, together with \Cref{lem:PhiB}, in \Cref{app:B-loss}.

\begin{proposition}[Revenue-loss identity]\label{prop:lossB}
Let $M$ be an IC, IR mechanism on $\cT$ with $u(\emptyset)=0$, and let $\varepsilon(M)=ma-\Rev_{F_0}(M)$ be its \emph{loss}. Then
\begin{equation}\label{eq:lossB}
\varepsilon(M)=\sum_{A\ne\emptyset}f_0(A)\sum_{j}\Phi_j(A)\big(1-x_j(A)\big)+\sum_{A\ne\emptyset}\sum_{j\in A}\lambda_{|A|}\,\sigma(A,j),
\end{equation}
and every summand is nonnegative. In particular (take $F=F_0$), $\OPT(F_0;1,m)=ma$. Moreover, $\varepsilon(M)=0$ if and only if $M$ allocates the grand bundle at every vertex other than the lowest type and all slacks vanish.
\end{proposition}

The first sum in \eqref{eq:lossB} is the virtual surplus lost by not allocating the grand bundle at the vertices other than the lowest type; the second is the slack of the constraints used by the flow, weighted by their flow. Both vanish for every $M_z$, so the loss measures how far a mechanism is from the family $\{M_z\}$.

Write $z(M)=x_1(\emptyset)$. A mechanism $M$ with zero loss agrees with $M_{z(M)}$ at the vertices. The next lemma, proved in \Cref{app:B-pointwise}, makes this quantitative and also bounds the payments of the new types, that is, of the types in $\cT$ other than the vertices. Types near the threshold $\|\bm\xi\|_1=1$ are almost indifferent between the outcomes $(z\one,maz)$ and $(\one,P-z)$ of $M_z$, so a mechanism with a small loss could charge them much more than $M_z$ does. The bound for the new types therefore requires the types with $\|\bm\xi\|_1<1$ to stay away from the threshold: for a parameter $\nu\in(0,1]$,
\begin{equation}\label{eq:sepB}
\bm t=a\one+\bm\xi\in\cT,\ \ \|\bm\xi\|_1<1\ \Longrightarrow\ \|\bm\xi\|_1\le1-\nu .
\end{equation}

\begin{lemma}[Pointwise estimates]\label{lem:pointwiseB}
Let $M$ be an IC, IR mechanism on $\cT$ with $u(\emptyset)=0$ and loss $\varepsilon$, let $z=z(M)$, and let $V=m(a+1)$. There are explicit constants $C_m$, depending only on $m$, and $B_m$, depending also on $a$ and $\nu$, both given in \eqref{eq:constsB}, such that the following hold.
\begin{enumerate}[label=(\alph*),nosep]
\item At every vertex, each allocation probability of $M$ differs from that of $M_z$ by at most $C_m\varepsilon$, and the payment by at most $(V+1)C_m\varepsilon$. In particular, if $\varepsilon=0$, then $M$ agrees with $M_z$ at every vertex.
\item If \eqref{eq:sepB} holds, then $s(\bm t)\le s_z(\bm t)+B_m\varepsilon$ for every $\bm t\in\cT$.
\end{enumerate}
\end{lemma}

Combining the two estimates shows that, after a small perturbation of $F_0$, the optimal revenue is still attained in the family $\{M_z\}$.

\begin{theorem}[Two-price stability]\label{thm:stabB}
Suppose that $\cT$ satisfies \eqref{eq:sepB} with parameter $\nu$ and that
\[
\chi:=1-mB_m\|f-f_0\|_1>0.
\]
Then every IC, IR mechanism $M$ on $\cT$ with $u(\emptyset)=0$ satisfies
\begin{equation}\label{eq:stabB}
\Rev_F(M)\le\Rev_F(M_{z(M)})-\chi\,\varepsilon(M).
\end{equation}
Consequently, $\OPT(F;1,m)=\max_{0\le z\le1}\Rev_F(M_z)$; and if $M$ is optimal, then $\varepsilon(M)=0$, $M_{z(M)}$ is optimal, and $M$ agrees with $M_{z(M)}$ at every vertex.
\end{theorem}

The proof, in \Cref{app:B-stab}, splits $\Rev_F(M)-\Rev_F(M_{z(M)})$ into three parts. The sum over the vertices weighted by $f_0$ equals $-\varepsilon(M)$. The sum over the vertices weighted by the differences between the probabilities under $F$ and under $F_0$ is controlled by \Cref{lem:pointwiseB}(a), and the sum over the new types by \Cref{lem:pointwiseB}(b); together, these two sums are at most $mB_m\|f-f_0\|_1\,\varepsilon(M)$.

We now return to the perturbed prior. \Cref{app:B-encoding} verifies that the prior $F$ of \eqref{eq:FB} satisfies all conditions of \Cref{thm:stabB}, with $p=d_0^{1/m}$, and that $\chi\ge3/4$. This yields the following.

\begin{proposition}\label{prop:encodeB}
Let $F$ be the prior of \eqref{eq:FB}. Every optimal mechanism satisfies, for all $j\in[m]$,
\[
x_j(a\one)=\begin{cases}1,&\cZ(T_0)\le c,\\ 0,&\cZ(T_0)\ge c+1.\end{cases}
\]
In these two cases, $M_1$ and $M_0$, respectively, are optimal, and
\begin{equation}\label{eq:OPTB}
\OPT(F;1,m)=ma+P\eta^m\pos{\tau-\vartheta}.
\end{equation}
\end{proposition}

\begin{proof}
By \Cref{thm:stabB}, $\OPT(F;1,m)=\max_z\Rev_F(M_z)$. By \eqref{eq:encrevB}, $\Rev_F(M_z)$ is a linear function of $z$ with nonzero slope, so it has a unique maximizer on $[0,1]$, namely $z=1$ if $\cZ(T_0)\le c$ and $z=0$ if $\cZ(T_0)\ge c+1$; the maximum is \eqref{eq:OPTB}. Let $M$ be an optimal mechanism. By \Cref{thm:stabB}, $M_{z(M)}$ is optimal, so $z(M)$ is the unique maximizer, and $M$ agrees with $M_{z(M)}$ at the vertex $a\one$, that is, $x_j(a\one)=z(M)$ for all $j$.
\end{proof}

\begin{corollary}\label{cor:twopricesB}
Let $F$ be the prior of \eqref{eq:FB}. The mechanisms that post the grand bundle at price $ma$ and at price $P$ earn different revenues, and the better of the two is optimal among all mechanisms, randomized or not. Moreover, every optimal mechanism allocates every item to the lowest type with the same probability $z\in\{0,1\}$.
\end{corollary}

\begin{proof}
These two mechanisms are $M_1$ and $M_0$, whose revenues differ by $P\eta^m|\tau-\vartheta|>0$ by \eqref{eq:encrevB}. The other claims are \Cref{prop:encodeB}.
\end{proof}

Since the better of these two mechanisms is deterministic, the largest revenue of a deterministic mechanism is $\OPT(F;1,m)$, and every deterministic mechanism that attains it is optimal; so the proof of \Cref{thm:B} in \Cref{sec:B-proof} also applies when only deterministic mechanisms are allowed.

\subsection{Proof of Theorem B}\label{sec:B-proof}

Let $(w_1,\dots,w_m,T)$ be an instance of the counting version of \textsc{SubsetSum}; we may assume that $m\ge2$ and $T\le\sum_iw_i$, since the other instances are trivial. Form $b_1,\dots,b_N$, $b_{\max}$ and $R$ as in \Cref{lem:tailB}. Since $R,R+1\in[1,mb_{\max}]$, the construction of \Cref{sec:B-construction} applies to $T_0\in\{R,R+1\}$ and to every integer $c\in[0,N^m)$.

The reductions use a rescaled prior. Dividing all values by $a+1$ turns $F$ into a distribution on $N+2=2m+2$ rational points of $(0,1]$ with the same probabilities, and the map $(\bm x,s)\mapsto(\bm x,s/(a+1))$ is a bijection between the mechanisms for the two priors that preserves IC, IR and allocations and divides revenue by $a+1$. In particular, under the rescaled prior the two prices become $ma/(a+1)$ and $P/(a+1)$, and \Cref{cor:twopricesB} holds for the rescaled prior. All numbers involved have bit length polynomial in the input length (\Cref{app:B-bits}), so the rescaled prior can be computed in polynomial time.

For~(i), take $c=0$, so that $\vartheta=1/(2N^m)$. By \eqref{eq:OPTB}, $\OPT(F;1,m)=ma$ if $\cZ(T_0)=0$, and otherwise
\[
\cZ(T_0)=\frac12+\frac{N^m\big(\OPT(F;1,m)-ma\big)}{P\eta^m}.
\]
Hence one query for the optimal revenue under the rescaled prior, multiplied by $a+1$, determines $\cZ(T_0)$, and the queries for $T_0=R$ and $T_0=R+1$ determine $\cN$ by \Cref{lem:tailB}.

For~(ii), we find $\cZ(T_0)\in[0,N^m]$ for $T_0\in\{R,R+1\}$ by binary search. To test whether $\cZ(T_0)\le c$ for an integer $c\in[0,N^m)$, query the oracle, with $e=\frac12$, for the allocation probabilities at the lowest type under the rescaled prior. By \Cref{prop:encodeB}, and since the rescaling preserves allocations, the probability for item~$1$ is $\ind{\cZ(T_0)\le c}$, whichever optimal mechanism the oracle fixes, and rounding it down to a multiple of $\frac12$ does not change it. After $O(m\log N)$ tests, $\cZ(T_0)$ is known, and \Cref{lem:tailB} gives $\cN$.

Thus the procedures for~(i) and~(ii) are polynomial-time Turing reductions from the counting version of \textsc{SubsetSum} to the exact computation of $\OPT(F;1,m)$ and to the oracle of~(ii), respectively. Since this problem is $\sharpP$-complete (\Cref{fact:subsetsum}), \Cref{thm:B} follows.

\section{Hardness with two items}\label{sec:two}

We begin with the formal statement of \Cref{thm:introC}.

\begin{theorem}[Theorem C]\label{thm:C}
Let $m=2$, and let the prior $F$ be supported in $(0,1]$.
\begin{enumerate}[label=(\roman*),nosep]
\item Given $F$, $n$ and a rational $\vartheta$, deciding whether $\OPT(F;n,2)\ge\vartheta$ is $\sharpP$-hard. In particular, computing $\OPT(F;n,2)$ exactly is $\sharpP$-hard. Moreover, given $F$, $n$ and a rational $e>0$, computing a number within $e$ of $\OPT(F;n,2)$ is $\sharpP$-hard.
\item Constructing an optimal mechanism is $\sharpP$-hard.
\end{enumerate}
\end{theorem}

The reduction for~(i) reads the optimal revenue only through comparisons ``$\OPT(F';n,2)\ge\vartheta$?'', which it uses to approximate optimal revenues within an additive error of polynomial bit length; this gives the last claim of~(i). Since $\OPT(F;n,2)\le2$ for $F$ supported in $(0,1]$, the claim also holds for a relative error $e$. The instances of this reduction have $8r+6$ support points and $2r+2$ buyers, and those of the reduction for~(ii) at most $8r+8$ support points and $4r+4$ buyers, where $r$ is the number of integers in the \textsc{SubsetSum} instance. Both numbers grow with the instance, whereas fixing either of them makes both problems solvable in polynomial time.

Part~(ii) does not follow from part~(i): the oracle returns allocations, not revenues, and even a full description of a mechanism gives its revenue only as a sum over exponentially many report profiles. Moreover, the oracle may use any optimal mechanism, randomized or asymmetric. As for \Cref{thm:B}(ii), the reduction therefore computes inputs at which every optimal mechanism reveals the answer: here one report profile, at which four allocation probabilities decide each query.

The proof reduces the study of optimal mechanisms to a problem about capacities. For the priors that we construct, the revenue of a menu is at most a constant times the sum of its capacities, that is, of its largest allocation probabilities of the two items, and the menus that attain this bound are determined by their capacities (\Cref{sec:two-real}); among the tables of capacities that respect supply, the one with the largest total capacity is given by a greedy recursion (\Cref{sec:two-tables}); and a very high value of small probability forces every optimal mechanism to use such menus (\Cref{sec:two-high}). In the limit, the greedy recursion becomes a scalar recursion that compares sums of coefficients (\Cref{sec:two-scalar}), and both reductions are built on it (\Cref{sec:C}).

\subsection{Menus, capacities and the realization lemma}\label{sec:two-real}

Throughout this section there are two items. For a mechanism $M$, the \emph{menu} $M_i(\bm v_{-i})$ of buyer~$i$ at reports $\bm v_{-i}$ of the others is the map $\bm t\mapsto\big(\bm x_i(\bm t,\bm v_{-i}),s_i(\bm t,\bm v_{-i})\big)$ on $\cT=\supp(F)^2$. More generally, a menu $Q=(\bm x,s)$ on $\cT$ is a mechanism for a single buyer with prior $F$, with revenue $\Rev_F(Q)$. By~\eqref{eq:dsic}, $M$ is DSIC and ex-post IR if and only if all its menus are IC and IR, that is, at every profile each buyer receives her preferred option from her menu, with nonnegative utility (the taxation principle~\cite{Ham79,Roc85}). The \emph{capacities} of a menu $(\bm x,s)$ are
\[
y_1=\max_{\bm t}x_1(\bm t),\qquad y_2=\max_{\bm t}x_2(\bm t);
\]
we write $\bm y=(y_1,y_2)$, so that $\|\bm y\|_1=y_1+y_2$. For a DSIC, ex-post IR mechanism $M$ with $n$ buyers, the menu $M_i(\bm v_{-i})$ has capacities $\bm y_i(\bm v_{-i})=\big(y_{i1}(\bm v_{-i}),y_{i2}(\bm v_{-i})\big)$, and
\[
\Capac(M)=\sum_{i=1}^n\sum_{k=1}^2\E_{\bm v_{-i}}\big[y_{ik}(\bm v_{-i})\big]
\]
is the \emph{total capacity} of $M$.

For a general prior, the menus that earn the most for given capacities can be complicated. We construct priors for which the revenue of a menu is at most a constant times the sum of its capacities, and the menus that attain this bound have three properties: they are determined by their capacities; their capacity pairs form exactly the set
\[
\cC_\kappa=\big\{(y_1,y_2)\in[0,1]^2:\ y_1\ge\kappa y_2,\ y_2\ge\kappa y_1\big\}
\]
for an adjustable $\kappa$; and some types buy a single item and receive an allocation that depends on the two capacities through an adjustable slope. The supply constraint then becomes linear in the capacities, with $\kappa$ and the slopes among its coefficients, and these coefficients encode the instances.

The prior is concentrated near~$1$, with support $\ell_J<\dots<\ell_0<1<c_0<\dots<c_J$. Given slopes $\beta_0=\kappa>\beta_1>\dots>\beta_J$, and with the convention $c_{-1}=\ell_{-1}=1$, the types fall into the following classes:
\begin{itemize}[nosep]
\item the \emph{exit class} $Z$: $t_1<1$ and $t_2<1$; it buys nothing;
\item the \emph{single-item classes} $X_j$ ($0\le j\le J$): $t_1=c_{j-1}$ and $t_2<\ell_{j-1}$; they buy only item~$1$;
\item the single-item classes $Y_j$, symmetric to $X_j$; they buy only item~$2$;
\item the \emph{bundle class} $B$: all other types; it buys both items.
\end{itemize}
For $\bm y\in\cC_\kappa$, the \emph{canonical menu} with capacities $\bm y$ gives the classes the following outcomes:
\begin{equation}\label{eq:canon}
\begin{array}{lll}
\text{class} & \text{allocation} & \text{payment}\\ \hline
Z & (0,0) & 0\\
X_j & (y_1-\beta_jy_2,\ 0) & y_1+(1-A_j)\,y_2\\
Y_j & (0,\ y_2-\beta_jy_1) & y_2+(1-A_j)\,y_1\\
B & (y_1,y_2) & y_1+y_2
\end{array}
\qquad\text{where }A_j=\ell_{j-1}+\beta_jc_{j-1}.
\end{equation}
The bundle class pays $1$ per unit of capacity, and the payment of $X_j$ makes the types of $B$ next to $X_j$ indifferent between the two outcomes: when $y_2>0$, a type $\bm t$ is indifferent between the outcomes of $X_j$ and of $B$ if and only if $t_2=A_j-\beta_jt_1$, and this line, line~$j$, passes through the consecutive corners $(c_{j-1},\ell_{j-1})$ and $(c_j,\ell_j)$ of $B$. The upper envelope of these lines separates the classes $X_j$ from $B$ (\Cref{fig:classes}).

\begin{figure}[t]
\centering
\begin{tikzpicture}[scale=0.64,every node/.style={font=\small}]
  \def\va{-3.12}\def\vb{-2.40}\def\vc{-1.36}\def\vd{0}\def\ve{1.6}\def\vf{3.2}\def\vg{4.8}
  \fill[black!14] (-3.40,-3.40) rectangle (-1.08,-1.08);
  \fill[black!6] (-0.26,-3.40) rectangle (0.26,-1.08);
  \fill[black!6] (1.34,-3.40) rectangle (1.86,-2.12);
  \fill[black!6] (2.94,-3.40) rectangle (3.46,-2.84);
  \fill[black!6] (-3.40,-0.26) rectangle (-1.08,0.26);
  \fill[black!6] (-3.40,1.34) rectangle (-2.12,1.86);
  \fill[black!6] (-3.40,2.94) rectangle (-2.84,3.46);
  \foreach \x in {\va,\vb,\vc,\vd,\ve,\vf,\vg} \foreach \y in {\va,\vb,\vc,\vd,\ve,\vf,\vg} \fill (\x,\y) circle (1.5pt);
  \draw[thick] (0,0) -- (1.6,-1.36) -- (3.2,-2.40) -- (4.8,-3.12) -- (5.4,-3.39);
  \draw[thick] (0,0) -- (-1.36,1.6) -- (-2.40,3.2) -- (-3.12,4.8) -- (-3.39,5.4);
  \foreach \x/\lab in {\va/{\ell_2},\vb/{\ell_1},\vc/{\ell_0},\vd/{1},\ve/{c_0},\vf/{c_1},\vg/{c_2}} {
    \node at (\x,-4.0) {$\lab$};
    \node at (-4.0,\x) {$\lab$};
  }
  \node at (-1.88,-2.76) {$Z$};
  \node[font=\scriptsize] at (0.62,-3.12) {$X_0$};
  \node[font=\scriptsize] at (2.22,-3.12) {$X_1$};
  \node[font=\scriptsize] at (3.82,-3.12) {$X_2$};
  \node[font=\scriptsize] at (-3.12,0.62) {$Y_0$};
  \node[font=\scriptsize] at (-3.12,2.22) {$Y_1$};
  \node[font=\scriptsize] at (-3.12,3.82) {$Y_2$};
  \node at (2.6,2.6) {$B$};
  \node at (5.6,-4.0) {$t_1$};
  \node at (-4.0,5.6) {$t_2$};
\end{tikzpicture}
\caption{The classes for $J=2$ (not to scale). The axes show the support points $\ell_2<\ell_1<\ell_0<1<c_0<c_1<c_2$, and the dots are the types. The polygonal line through $(1,1)$, $(c_0,\ell_0)$, $(c_1,\ell_1)$ and $(c_2,\ell_2)$, with slopes $-\kappa$, $-\beta_1$ and $-\beta_2$, is the part with $t_1\ge1$ of the upper envelope of the lines $t_2=A_j-\beta_jt_1$; the other polygonal line is its mirror image. The classes $X_0,X_1,X_2$ consist of the types on the vertical lines $t_1=1,c_0,c_1$ below the envelope; the $Y_j$ are their mirror images, $Z$ is the lower left square, and all other types form $B$.}
\label{fig:classes}
\end{figure}

\begin{lemma}[Realization lemma]\label{lem:real}
Let $J\ge1$, and let $\kappa,\underline\beta,g,\beta_1,\dots,\beta_J$ be rationals with
\[
\tfrac12\le\kappa<1,\qquad 0<\underline\beta\le\beta_J<\beta_{J-1}<\dots<\beta_1\le\kappa-g,\qquad g>0,
\]
and put $\beta_0=\kappa$. There are positive rationals $d$, $q_0$ and $\theta$, depending only on $(J,\underline\beta,g)$, and, for every rational $p\in[d/2,3d/2]$, a prior $F_p$ on $2J+3$ rational points $\ell_J<\dots<\ell_0=1-d<1<c_0<\dots<c_J<1+3d$, which do not depend on $p$, such that the following hold for the partition above.
\begin{enumerate}[label=(\alph*),nosep]
\item Every IC, IR menu $Q$ on $\cT$ with capacities $\bm y$ satisfies $\Rev_{F_p}(Q)\le\gamma_p\|\bm y\|_1$, where
\[
\gamma_p=\Pr(B)+(1-\kappa)\,pq_0+\sum_{j=1}^J(2-A_j)\,\mu p\,\theta^j,\qquad \mu=\frac{1-p-q_0}{J+1},
\]
with equality if and only if $\bm y\in\cC_\kappa$ and $Q$ gives every type the outcome of the canonical menu~\eqref{eq:canon}. In particular, $\OPT(F_p;1,2)=2\gamma_p$. Moreover, $1-d/4\le\gamma_p<1$.
\item The probability mass function $f_p$ of $F_p$ satisfies $f_p(\ell_j)=p(1-\theta)\theta^j$ for $j<J$, $f_p(\ell_J)=p\theta^J$, $f_p(1)=q_0$ and $f_p(c_j)=\mu$. Hence the classes have probabilities
\[
\Pr(Z)=p^2,\qquad \Pr(X_0)=\Pr(Y_0)=pq_0,\qquad \Pr(X_j)=\Pr(Y_j)=\mu p\,\theta^j\quad(1\le j\le J),
\]
and $\Pr(B)$ is the rest. Conditionally on the type not being in $Z$, the single-item classes have total probability at most $1/100$.
\item $d$, $q_0$, $\theta$, the support, $F_p$ and $\gamma_p$ can be computed in time polynomial in $J$ and the bit length of the input and of $p$.
\end{enumerate}
\end{lemma}

By weak duality (\Cref{cor:cap}), it suffices to construct a flow under which every constraint of positive weight is tight at every canonical menu and whose virtual values are positive in both coordinates on $B$, zero in the first and negative in the second coordinate on $X_j$ (symmetrically on $Y_j$), and negative in both coordinates on $Z$. Then \eqref{eq:certcap} gives the bound $\gamma_p\|\bm y\|_1$, which the canonical menu attains.

The flow is explicit. Let $\bm P_j=(c_j,\ell_j)$ for $0\le j\le J$ and $\bm P_{-1}=(1,1)$; these corners belong to $B$. Every type $\bm t$ of $X_j$ receives flow from the two corners next to it, $\bm P_{j-1}$ (which has the same first coordinate as $\bm t$) and $\bm P_j$, along the IC constraints that these corners do not gain by reporting $\bm t$. Both constraints are tight at every canonical menu, because $\bm t$ obtains the utility of $\bm P_{j-1}$ and $\bm P_j-\bm P_{j-1}$ is a multiple of $(1,-\beta_j)$. The inflow from $\bm P_{j-1}$ makes $\Phi_2(\bm t)$ negative, and the inflow from $\bm P_j$, together with the flow that the types of $X_{j+1}$ pass on to the type $(c_{j-1},\ell_j)$, makes $\Phi_1(\bm t)=0$. The other types of $B$ send their mass to the corners, the single-item types pass their inflow down to $X_0$ and $Y_0$, and the flow ends in $Z$, whose IR constraints absorb it (\Cref{app:realflow}).

The two constraints also determine the allocation of $\bm t$: subtracting them shows that $(c_j-c_{j-1})\big(x_1(\bm t)-\beta_jx_2(\bm t)\big)$ equals $u(\bm P_j)-u(\bm P_{j-1})$ up to their slacks. By the identity~\eqref{eq:slackid}, the revenue that a menu loses relative to the bound controls these slacks, and with them its distance from a canonical menu. For an IC, IR menu $Q=(\bm x,s)$ on $\cT$ with capacities $\bm y$, let
\[
\loss(Q)=\gamma_p\|\bm y\|_1-\Rev_{F_p}(Q)\ \ge0
\]
be its \emph{loss}. For $0\le U_1,U_2\le1$, let
\begin{equation}\label{eq:G}
G_\kappa(U_1,U_2)=\big(\min\{U_1,U_2/\kappa\},\ \min\{U_2,U_1/\kappa\}\big),
\end{equation}
the coordinatewise largest point of $\cC_\kappa$ below $(U_1,U_2)$ (\Cref{app:capacity}).

\begin{lemma}[Stability]\label{lem:stab}
There is a rational $K_0\ge1$ of polynomial bit length, depending only on $(J,\underline\beta,g)$ and computable in polynomial time, such that the following holds for every $p\in[d/2,3d/2]$ and every IC, IR menu $Q=(\bm x,s)$ on $\cT$ with capacities $\bm y$.
\begin{enumerate}[label=(\alph*),nosep]
\item Let $\hat{\bm y}=G_\kappa(\bm y)$. Every allocation probability, payment and capacity of $Q$ differs from that of the canonical menu with capacities $\hat{\bm y}$ by at most $K_0\loss(Q)$.
\item If all payments of $Q$ are nonnegative, there is $\bm y'\in\cC_\kappa$ such that, for the canonical menu $(\bm x',s')$ with capacities $\bm y'$, the differences $x_k(\bm t)-x'_k(\bm t)$, $s(\bm t)-s'(\bm t)$ and $y_k-y'_k$ lie in $[0,K_1\loss(Q)]$ for all $\bm t\in\cT$ and $k=1,2$, where $K_1=5K_0/\min\{1-\kappa^2,g\}$.
\end{enumerate}
\end{lemma}

Part~(a) contains the equality statement of \Cref{lem:real}(a). Part~(b) lowers a menu to a canonical menu without raising any allocation probability or payment; \Cref{sec:two-high} applies it to all menus of a mechanism at once, which preserves supply. The proofs of \Cref{lem:real,lem:stab} are in \Cref{app:realization}. The mass $p$ of the values below~$1$ may vary in $[d/2,3d/2]$ with the same constant $K_0$; \Cref{sec:C-reveal} uses this freedom to make the prices $\ell_0$ and $1$ earn almost equal revenues.

\subsection{Capacity tables}\label{sec:two-tables}

Fix a prior $F=F_p$ of the realization lemma and write $\gamma=\gamma_p$. Summing the bound of \Cref{lem:real}(a) over the buyers and the profiles of the other reports gives $\Rev(M)\le\gamma\Capac(M)$ for every DSIC, ex-post IR mechanism $M$, and we call $M$ \emph{tight} if equality holds. Since every profile of reports has positive probability, $M$ is tight if and only if all its menus attain the bound of \Cref{lem:real}(a), that is, are canonical. The allocation of a tight mechanism at a profile $\bm v$ is therefore $\bm x_i(\bm v)=C_{\mathrm{class}(\bm v_i)}\big(\bm y_i(\bm v_{-i})\big)$, with the \emph{allocation maps}
\begin{equation}\label{eq:Cmaps}
C_Z(\bm y)=\bm 0,\qquad C_B(\bm y)=\bm y,\qquad C_{X_j}(\bm y)=(y_1-\beta_jy_2,\ 0),\qquad C_{Y_j}(\bm y)=(0,\ y_2-\beta_jy_1),
\end{equation}
whose coordinates are nonnegative on $\cC_\kappa$. A tight mechanism is thus described by its \emph{capacity table} $\bm y=(\bm y_i(\bm v_{-i}))$, which takes values in $\cC_\kappa$ and is \emph{feasible}:
\begin{equation}\label{eq:capsupply}
\sum_{i=1}^nC_{\mathrm{class}(\bm v_i)}\big(\bm y_i(\bm v_{-i})\big)\le(1,1)\qquad\text{for every report profile }\bm v.
\end{equation}
Conversely, every feasible table with values in $\cC_\kappa$ is the capacity table of the tight mechanism whose menus are the canonical menus with these capacities. The total capacity is the linear function $\cY(\bm y)=\sum_i\E\big\|\bm y_i(\bm v_{-i})\big\|_1$, and $\Capac^*_n$ denotes its maximum over feasible tables, the largest total capacity of a tight mechanism.

Two averaging operations preserve feasibility and $\cY$. Conditionally on their classes, the reports of different buyers remain independent, and the maps $C_\tau$ are linear; so averaging $\bm y_i(\bm v_{-i})$ over the reports $\bm v_{-i}$ with given classes preserves~\eqref{eq:capsupply}. The buyers are symmetric, so averaging over permutations of the buyers preserves it as well. A table is \emph{anonymous} if $\bm y_i(\bm v_{-i})$ depends only on the multiset of the classes of $\bm v_{-i}$.

Buyers of the exit class receive nothing, since $C_Z\equiv\bm 0$, so the problem splits according to the number of buyers outside~$Z$. For an anonymous table $\bm w$, let $\bm w(l,S)$ be its value when $n-l$ opponents are in $Z$ and $S$ is the multiset of single-item classes among the other $l-1$. The constraints~\eqref{eq:capsupply} at the profiles with exactly $n-l$ buyers in $Z$ involve only $\bm w(l,\cdot)$, and they form a problem of the same kind with $l$ buyers and without the class $Z$, in which the classes are drawn conditionally on not being in $Z$; we call it the \emph{subproblem} with $l$ buyers. Likewise, an objective that sums, over the buyers, the expectation of a fixed function of the entries of $\bm w$ is a positive combination of the corresponding objectives of the subproblems (\Cref{app:averaging}).

An anonymous table of the subproblem with $l$ buyers assigns a capacity pair $\bm w(S)\in\cC_\kappa$ to every \emph{context} $S$ for $l$ buyers, a multiset of at most $l-1$ single-item classes (those among the opponents); we write $n_B(S)=l-|S|$, the number of buyers of class $B$ in a class profile of the subproblem at which a buyer of class $B$ faces $S$. The \emph{greedy table} $\bm w^*_l$ fixes these pairs in increasing order of $|S|$: $\bm w^*_l(\emptyset)=(1/l,1/l)$ and
\begin{equation}\label{eq:greedy}
\big(U_1(S),U_2(S)\big)=\frac{(1,1)-\sum_{\tau\in S}C_\tau\big(\bm w^*_l(S-\tau)\big)}{n_B(S)},\qquad \bm w^*_l(S)=G_\kappa\big(U_1(S),U_2(S)\big),
\end{equation}
where the sum counts multiplicity, $S-\tau$ denotes $S$ with one copy of $\tau$ removed, and $G_\kappa$ is the map~\eqref{eq:G}. In the class profile consisting of $S$ and $n_B(S)$ buyers of class $B$, a buyer of class $\tau\in S$ faces $S-\tau$ and receives $C_\tau(\bm w^*_l(S-\tau))$, and the $n_B(S)$ buyers of class $B$ split the remaining supply equally, each receiving the largest point of $\cC_\kappa$ below her share (\Cref{fig:greedy}). Let $\bm w^{\mathrm{ex}}$ be the table that uses the greedy table in every subproblem: $\bm w^{\mathrm{ex}}_i(\bm v_{-i})=\bm w^*_l(S)$, where $l-1$ is the number of opponents of $i$ not in $Z$ and $S$ is the multiset of their single-item classes.

\begin{figure}[t]
\centering
\begin{tikzpicture}[scale=3.3,every node/.style={font=\small}]
  \def\kap{0.72}
  \fill[black!12] (0,0) -- (1,\kap) -- (1,1) -- (\kap,1) -- cycle;
  \draw[black!35] (1,0) -- (1,1) -- (0,1);
  \draw[dashed,black!55] (0,0) -- (1,1);
  \draw[->] (0,0) -- (1.13,0) node[below] {$y_1$};
  \draw[->] (0,0) -- (0,1.13) node[left] {$y_2$};
  \node[below] at (1,0) {$1$};
  \node[left] at (0,1) {$1$};
  \node at (0.70,0.90) {$\cC_\kappa$};
  \draw[dotted,thick] (0.9,0) -- (0.9,0.5) -- (0,0.5);
  \fill (0.9,0.5) circle (0.45pt);
  \node[right] at (0.9,0.5) {$\bm U$};
  \draw[->] (0.875,0.5) -- (0.725,0.5);
  \fill (0.694,0.5) circle (0.45pt);
  \node[above left,inner sep=1.5pt] at (0.694,0.5) {$G_\kappa(\bm U)$};
\end{tikzpicture}
\caption{One step of the greedy recursion~\eqref{eq:greedy} (not to scale). The capacity pairs of canonical menus form the narrow cone $\cC_\kappa$ around the diagonal (shaded). At a context, a buyer of class $B$ receives $G_\kappa(\bm U)$, the largest point of $\cC_\kappa$ below her share $\bm U$ of the remaining supply.}
\label{fig:greedy}
\end{figure}

\begin{lemma}[Narrow cone]\label{lem:narrow}
Let $\chi=\max_{0\le j\le J}(1-\kappa\beta_j)$, and suppose that $n\chi\le1/4$. Then, for every $l\le n$, the greedy table $\bm w^*_l$ is feasible, and it is the unique feasible anonymous table $\bm w$ of the subproblem with $l$ buyers that maximizes $\E\|\bm w(S)\|_1$, where $S$ is the multiset of single-item classes among $l-1$ independent buyers with the class probabilities conditional on not being in $Z$. Consequently, $\Capac^*_n=\cY(\bm w^{\mathrm{ex}})=n\,\E\|\bm w^*_l(S)\|_1$, where now $l-1$ is the number of buyers not in $Z$ among $n-1$ independent buyers and $S$ is the multiset of their single-item classes.
\end{lemma}

The narrow cone, where $\kappa$ and all slopes are close to $1$, keeps the greedy table feasible: $\cC_\kappa$ is close to the diagonal, and every single-item buyer receives at most a fraction $\chi$ of her capacity. Optimality reflects the rarity of single-item classes. The capacity of any feasible table at $S$ is at most the residual supply computed from its own allocations at smaller contexts, so deviations at smaller contexts propagate upwards, and averaging over the random context contracts them: the capacity that a feasible anonymous table loses relative to the greedy table is at least a constant times its expected distance from it (\Cref{lem:narrowquant}). The proofs are in \Cref{app:capacity}.

\subsection{A very high value}\label{sec:two-high}

Two questions remain. An optimal mechanism need not be tight, so its revenue need not equal $\gamma\Capac^*_n$; and part~(ii) of \Cref{thm:C} concerns every optimal mechanism, not only the optimal revenue. We answer both by adding to the prior a very high value of small probability; this device is related to~\cite[Section~6.2]{KMSSW19}.

Let $F=F_p$ be a prior of the realization lemma, with largest support point $c_J$, and $\gamma=\gamma_p$. For a rational $\sigma\in(0,\gamma)$, fixed below, and a rational $H>c_J$, put
\begin{equation}\label{eq:FH}
\underline\gamma=\gamma-\sigma,\qquad \pi=\frac{\underline\gamma}H,\qquad F_H=(1-\pi)F+\pi\,\delta_H .
\end{equation}
Reports whose two values lie in $\supp(F)$ are \emph{base reports}, and the restriction of a mechanism for $F_H$ to the profiles of base reports is its \emph{base restriction}. For an IC, IR menu $Q$ on $\cT$, let $V_H(Q)$ be the largest revenue $\Rev_{F_H}(Q')$ of an IC, IR menu $Q'$ on $\supp(F_H)^2$ that agrees with $Q$ on $\cT$; these menus form a nonempty polytope, so the largest revenue is attained (\Cref{app:high}). For a price $z$, let $\rev(z)=z\Pr_{a\sim F}[a\ge z]$ be its revenue from a single value drawn from $F$, and let $\rev^*=\max_z\rev(z)$.

A buyer who values item~$1$ at $H$ can mimic the report at which item~$1$ is allocated with probability $y_1$, so the mechanism collects at most about $H(1-y_1)$ from her; item~$2$ is symmetric. Since the value $H$ has probability $\pi=\underline\gamma/H$, each unit of capacity costs about $\underline\gamma$ of the revenue from the value $H$, while a canonical menu earns $\gamma$ per unit of capacity; the difference $\sigma$ rewards capacity. Three lemmas of \Cref{app:high} make this accounting exact.

First, the buyers can be treated separately. Suppose that $(1-\pi)\rev^*+\pi c_J<\underline\gamma$; since $\pi H=\underline\gamma$, this makes $H$ the unique optimal price for a single value drawn from $F_H$. Then, given a mechanism on the base reports, some optimal extension has the following form: if exactly one buyer reports a value $H$, it serves only her, through an extension of her menu that attains $V_H$; if several buyers do, it sells each item reported at $H$ at the price $H$. Hence the largest revenue of an extension is a fixed multiple of the expected sum of the values $V_H$ of the menus, plus a term that depends only on $n$, $H$ and $\pi$ (\Cref{lem:separation}).

Second, $V_H$ has a closed form at canonical menus. A type $(H,a)$ buys item~$1$ for sure, and what remains to be sold to it is the unused capacity $1-y_2$ of item~$2$, offered as a menu to the value $a$; such a menu is a mixture of posted prices. For $r,T\ge0$, let
\begin{equation}\label{eq:frontier}
\cF_r(T)=\max\Big\{\sum_z\lambda_z\rev(z):\ \lambda\ge0,\ \sum_z\lambda_z\le r,\ \sum_z\lambda_z(H-z)\le T\Big\},
\end{equation}
where $z$ ranges over $\supp(F)$: the largest revenue of a mixture of posted prices of total quantity at most $r$ that leaves utility at most $T$ to the value $H$. The type $(H,H)$ can imitate the offers to $(H,a)$ and to $(a,H)$, and its probability $\pi^2$ is $\pi/(1-\pi)$ times the probability $\pi(1-\pi)$ of each of the two families $\{(H,a)\}_a$ and $\{(a,H)\}_a$; so the utility $T$ that the offers leave to the value $H$ is charged at this rate:
\begin{equation}\label{eq:Gamma}
\Gamma_H(\bm r)=\max_{T\ge0}\Big\{\cF_{r_1}(T)+\cF_{r_2}(T)-\frac\pi{1-\pi}\,T\Big\}\qquad\big(\bm r=(r_1,r_2)\ge\bm 0\big).
\end{equation}
As the value of a linear program, $\Gamma_H$ is concave and nondecreasing in each argument, and it is $\rev^*$-Lipschitz for $\|\cdot\|_1$ (\Cref{lem:frontierprops}). Under the condition of the first lemma and if $H\ge3c_J$, for every $\bm y\in\cC_\kappa$ the value $V_H(Q_{\bm y})$ of the canonical menu $Q_{\bm y}$ with capacities $\bm y$ is a constant plus $\sigma_H\|\bm y\|_1$ plus $\pi(1-\pi)\,\Gamma_H(\one-\bm y)$, where $\sigma_H=\sigma+(1-\gamma)(2\pi-\pi^2)$ (\Cref{lem:canonext}). The first two terms collect the payments of the base types and of the types with the value $H$ for the capacities $\bm y$, and the last term is the revenue from the unused capacities $\one-\bm y$.

Third, lowering a menu $Q$ with nonnegative payments to a canonical menu as in \Cref{lem:stab}(b) removes its loss, which is worth $(1-\pi)^2\loss(Q)$, and decreases the term $\sigma_H\|\bm y\|_1$ and the revenue from the value $H$ by at most a small multiple of the loss (\Cref{lem:lowered}). So lowering raises $V_H$ when the loss is positive and $\sigma$ and $\pi$ are small compared with $1/K_1$. From now on, $\sigma=\min\{1/(100K_1),\ d/100\}$, with $K_1$ as in \Cref{lem:stab}.

For $l\le n$ and a context $S$ for $l$ buyers, let $\cA_{l,S}$ be the set of pairs $(i,\bm v_{-i})$ of a buyer and base reports of her opponents in which exactly $l-1$ opponents are not in $Z$ and the multiset of single-item classes among them is $S$; the menus $M_i(\bm v_{-i})$ of a mechanism $M$ with $(i,\bm v_{-i})\in\cA_{l,S}$ are its menus at~$S$. We write $\E_{l,S}$ for the expectation over $(i,\bm v_{-i})\in\cA_{l,S}$, with $i$ uniform and $\bm v_{-i}\sim F^{\otimes(2n-2)}$ conditioned on $(i,\bm v_{-i})\in\cA_{l,S}$; every pair in $\cA_{l,S}$ has positive probability.

\begin{theorem}[Optimal mechanisms]\label{thm:structure}
Suppose that $n\chi\le1/4$, and let $H_0=2/\sigma$. Let $H\ge H_0$ and $p\in[d/2,3d/2]$ be rational, and let $M$ be an optimal mechanism for $(F_H;n,2)$, where $F_H$ is formed from $F=F_p$ by~\eqref{eq:FH}. Then:
\begin{enumerate}[label=(\alph*),nosep]
\item every menu of the base restriction of $M$ is canonical, and the capacities $\bm y_i(\bm v_{-i})$ of these menus form a feasible table $\bm y$;
\item $\E_{l,S}\,\bm y_i(\bm v_{-i})=\bm w^*_l(S)$ for all $1\le l\le n$ and all contexts $S$ for $l$ buyers;
\item $\E_{l,S}\,\Gamma_H\big(\one-\bm y_i(\bm v_{-i})\big)=\Gamma_H\big(\one-\bm w^*_l(S)\big)$ for all such $l$ and $S$;
\item the optimal revenue is
\[
\begin{aligned}
\OPT(F_H;n,2)={}&2H\big[1-(1-\pi)^n\big]\\
&+(1-\pi)^{2n-2}\sum_{i=1}^n\E\Big[\sigma_H\big\|\bm w^{\mathrm{ex}}_i(\bm v_{-i})\big\|_1+\pi(1-\pi)\,\Gamma_H\big(\one-\bm w^{\mathrm{ex}}_i(\bm v_{-i})\big)\Big].
\end{aligned}
\]
\end{enumerate}
\end{theorem}

The proof (\Cref{app:structure}) combines these three lemmas. In an optimal mechanism, all payments are nonnegative by~\eqref{eq:bounded}, since the lowest type has utility zero in every menu (\Cref{sec:model-def}); and lowering every menu of the base restriction as in \Cref{lem:stab}(b) preserves supply. By \Cref{lem:separation,lem:canonext,lem:lowered}, lowering raises the revenue unless every loss vanishes, which gives~(a). By \Cref{lem:separation,lem:canonext}, the optimal tables then maximize $\sigma_H\cY(\bm y)+\pi(1-\pi)\sum_i\E\,\Gamma_H(\one-\bm y_i(\bm v_{-i}))$. Averaging cannot decrease this concave objective, and since $\pi(1-\pi)\rev^*$ is small compared with $\sigma$, the contraction of \Cref{lem:narrowquant} makes the greedy table the unique optimal anonymous table of every subproblem; equality in Jensen's inequality then gives (b)--(d).

\subsection{The scalar recursion}\label{sec:two-scalar}

As $\kappa\to1$, the set $\cC_\kappa$ shrinks to the diagonal, and the greedy table tends to a recursion for a single number at each context. Both parts of the proof are built on this recursion.

We describe a class of slope $\beta$ by its \emph{coefficient} $a:=1-\beta$, and write $X(a)$ and $Y(a)$ for the single-item classes of coefficient $a$; thus $X_j=X(1-\beta_j)$, $Y_j=Y(1-\beta_j)$, and
\[
C_{X(a)}(\bm y)=\big(y_1-(1-a)y_2,\ 0\big),\qquad C_{Y(a)}(\bm y)=\big(0,\ y_2-(1-a)y_1\big).
\]
At a diagonal point $(y_1,y_1)$, the class $X(a)$ thus receives $a\,y_1$ of item~$1$. Fix a number $l\ge1$ of buyers; contexts are those for $l$ buyers (\Cref{sec:two-tables}), and $n_B(S)=l-|S|\ge1$. We regard a context as a multiset of occurrences of classes; for an occurrence $\tau$ in $S$, $S-\tau$ is $S$ with that occurrence removed. The sum $\sum_{X(a)\in S}$ runs over the occurrences of $X$-classes in $S$, $a$ denoting the coefficient of the occurrence; $\sum_{Y(b)\in S}$ runs over the occurrences of $Y$-classes, and $\sum_{X(a),Y(b)\in S}$ over the pairs of an $X$-occurrence and a $Y$-occurrence. The \emph{signed coefficient sum} of $S$ is
\[
\mathsf d(S)=\sum_{X(a)\in S}a-\sum_{Y(b)\in S}b .
\]

For $\kappa=1$, both coordinates of $G_1(U_1,U_2)$ in~\eqref{eq:G} equal $\min\{U_1,U_2\}$. The greedy recursion~\eqref{eq:greedy} then produces diagonal points $(W^{(l)}(S),W^{(l)}(S))$, so the \emph{scalar table} $W^{(l)}$ is given by $W^{(l)}(\emptyset)=1/l$ and
\begin{equation}\label{eq:scalar}
\begin{gathered}
n_B(S)\,W^{(l)}(S)=1-\max\{L^X(S),L^Y(S)\},\\
L^X(S)=\sum_{X(a)\in S}a\,W^{(l)}\big(S-X(a)\big),\qquad L^Y(S)=\sum_{Y(b)\in S}b\,W^{(l)}\big(S-Y(b)\big).
\end{gathered}
\end{equation}
The \emph{loads} $L^X(S)$ and $L^Y(S)$ are the amounts of items~$1$ and~$2$ taken by the single-item buyers, and the $n_B(S)$ buyers of class $B$ receive equal amounts of the two items, limited by the more heavily loaded item. Let $D(S)=L^X(S)-L^Y(S)$. We also allow the coefficient $0$, the limit as $\kappa\to1$ of the coefficient $1-\kappa$ of $X_0$ and $Y_0$.

\begin{lemma}[Signs of the scalar table]\label{lem:Dsign}
Suppose that all coefficients lie in $[0,\bar a]$, where $l\bar a\le1/4$. Then, for every context $S$, $\frac3{4n_B(S)}\le W^{(l)}(S)\le\frac1{n_B(S)}$, $\operatorname{sgn}D(S)=\operatorname{sgn}\mathsf d(S)$, and
\[
\frac{3\,|\mathsf d(S)|}{4(n_B(S)+1)}\ \le\ |D(S)|\ \le\ \frac{|\mathsf d(S)|}{n_B(S)+1}.
\]
\end{lemma}

So the more heavily loaded item at a context is decided by the sign of its signed coefficient sum alone, and the loads tie exactly when $\mathsf d(S)=0$. The proof (\Cref{app:scalar}) is an induction on $|S|$ based on the identity
\begin{equation}\label{eq:Dscalar}
\big(n_B(S)+1\big)D(S)=\mathsf d(S)-\sum_{X(a)\in S}a\pos{D\big(S-X(a)\big)}+\sum_{Y(b)\in S}b\pos{-D\big(S-Y(b)\big)}\qquad(|S|\ge1),
\end{equation}
which follows from~\eqref{eq:scalar} at the children of $S$: the contributions of the grandchildren cancel.

Both parts of the proof compare greedy tables, with $\kappa<1$ and slightly different coefficients, with the scalar table. The projection $G_\kappa$ keeps the smaller of the two residuals, so the smaller coordinate of a greedy table satisfies the scalar recursion up to the differences between the loads. These differences reach a context only through the coefficients of its single-item classes, whose sum is small, so they are not amplified along the recursion: the greedy table differs from the scalar table by an amount linear in $1-\kappa$ and in the changes of the coefficients (\Cref{lem:Wclose}).

\subsection{Proof of Theorem C}\label{sec:C}

Both parts reduce from the counting version of \textsc{SubsetSum} through one counting problem of fixed cardinality (\Cref{sec:C-blocks}). Part~(i) reads the total capacity of tight mechanisms for priors that encode this problem in the slopes, and extracts the count by a finite difference (\Cref{sec:C-revenue}); part~(ii) reads the sign of $\bm w^*_n(S^*)_1-\bm w^*_n(S^*)_2$ at one context $S^*$ (\Cref{sec:C-compare,sec:C-imbalance,sec:C-reveal}). \Cref{sec:C-proof} completes the proof. The symbols introduced in \Cref{sec:C-revenue}, and those introduced in \Cref{sec:C-compare,sec:C-imbalance,sec:C-reveal}, are local to these subsections and to the corresponding parts of the appendix.

\subsubsection{A counting problem with groups of four numbers}\label{sec:C-blocks}

Let $(w_1,\dots,w_r,T)$ be an instance of the counting version of \textsc{SubsetSum} with $r\ge1$ (instances with $r=0$ are trivial), and let $\cN$ be its number of solutions. Both reductions use an encoding in which all solutions correspond to sets of the same cardinality and the same index sum. Put $k=2r$, $L=4k+1$ and $\bar w=1+\max\{T,\,k\max_iw_i\}$, and
\begin{equation}\label{eq:blocks4}
z_{i,j}=\bar w\big(L^{2i-2}+jL^{2i-1}\big)+\ind{j=4}\,w_i\quad(i\in[r],\ j\in[4]),\qquad \hat T=\bar w\sum_{i=1}^r\big(2L^{2i-2}+5L^{2i-1}\big)+T .
\end{equation}
The numbers $z_{i,1},\dots,z_{i,4}$ form \emph{group}~$i$. These $2k$ numbers increase in the lexicographic order of $(i,j)$; we list them as $z_1<\dots<z_{2k}$, so that $z_{i,j}=z_{4i-4+j}$, and call $4i-4+j$ the \emph{index} of $z_{i,j}$.

These $2k$ numbers turn the counting problem into one of fixed cardinality: the sets $I\subseteq[2k]$ with $|I|=k$ and $\sum_{j\in I}z_j=\hat T$ correspond bijectively to the solutions, and all of them have the same index sum; every ordered $k$-tuple of these numbers with sum $\hat T$, repetitions allowed, has distinct entries, so each of these sets arises from exactly $k!$ such tuples (\Cref{lem:blocks}). Since class probabilities will decrease geometrically with the index (\Cref{lem:real}(b)), all solutions then have the same probability, which part~(i) uses; part~(ii) uses that they have the same cardinality.

\subsubsection{Part~(i): reading the total capacity}\label{sec:C-revenue}

Fix the numbers of \Cref{sec:C-blocks}. We take $n=k+2$ buyers, so that the $n-1$ opponents of a buyer can consist of one distinguished class and $k$ classes that form a solution; put
\begin{equation}\label{eq:uC}
J=2k+1,\qquad \Lambda=1+(k+1)(\hat T+1+z_{2k}),\qquad u_j=\Lambda+z_j\quad(j\in[2k]),
\end{equation}
\[
u_{\max}=k\Lambda+\hat T+1,\qquad \nu=\frac1{2000\,n^2u_{\max}},\qquad \underline\beta=\frac12,\qquad g=\frac{\nu u_1}2 .
\]
For a rational $R\in[\hat T-1,\hat T+1]$, give the \emph{distinguished} classes $X_J,Y_J$ the number $u_J(R)=k\Lambda+R$, and call the classes $X_j,Y_j$ with $j\in[2k]$ \emph{ordinary}; then $u_1<\dots<u_{2k}<u_J(R)\le u_{\max}$, since $k\ge2$ and $\Lambda>z_{2k}$. For $0<\varepsilon\le g/2$, let $F_{R,\varepsilon}$ be the prior $F_d$ of the realization lemma for $J$, $\kappa=1-\varepsilon$, $\underline\beta$, $g$ and the slopes $\beta_j=1-\nu u_j$ ($j\in[2k]$) and $\beta_J=1-\nu u_J(R)$. Its hypotheses hold: the slopes decrease strictly, $\beta_J\ge1-\nu u_{\max}>\frac12$, and $\kappa-\beta_1=\nu u_1-\varepsilon\ge g$. In the notation of \Cref{sec:two-scalar}, the classes $X_j,Y_j$ ($j\in[J]$) of $F_{R,\varepsilon}$ have coefficient $\nu u_j$, and $X_0,Y_0$ have coefficient $\varepsilon$. By \Cref{lem:real}(b), all these priors have the same class probabilities, which depend only on $J$, $\underline\beta$ and $g$; in particular, $\Pr(X_j)=\Pr(Y_j)=\mu d\theta^j$ for $j\in[J]$, where $\mu=(1-d-q_0)/(J+1)$.

Draw the classes of $n-1$ opponents independently with these probabilities, let $l$ be one plus the number of opponents not in $Z$, and let $S$ be the multiset of their single-item classes. Let $W^{(l)}_R$ be the scalar table~\eqref{eq:scalar} with $l$ buyers in which $X_j,Y_j$ have coefficient $\nu u_j$ for $j\in[2k]$, $\nu u_J(R)$ for $j=J$, and $0$ for $j=0$, and define
\begin{equation}\label{eq:WR}
\cW(R)=2n\,\E\big[W^{(l)}_R(S)\big].
\end{equation}
This is the limit, as $\varepsilon\to0$, of the largest total capacity $\Capac^*_n$ of a tight mechanism for $F_{R,\varepsilon}$.

The optimal revenue determines $\cW(R)$ approximately. Form the prior $F_H$ of~\eqref{eq:FH} from $F=F_{R,\varepsilon}$, with $\varepsilon$ small and $H$ large. By \Cref{thm:structure}(d), the optimal revenue minus $2H[1-(1-\pi)^n]$, divided by $\sigma(1-\pi)^{2n-2}$, differs from $\Capac^*_n$ by $O(n/(\sigma H))$; by \Cref{lem:narrow} and the comparison with the scalar table (\Cref{lem:Wclose}), $\Capac^*_n$ differs from $\cW(R)$ by $O(n^2\varepsilon)$. Hence binary search with comparisons ``$\OPT(F';n,2)\ge\vartheta$?'' approximates $\cW(R)$ to a given precision in time polynomial in the input length and the bit lengths of $R$ and of the precision; the same holds if sufficiently precise approximations of the optimal revenues are available instead (\Cref{lem:Cread}).

The offset $\Lambda$ makes the sign rule of \Cref{lem:Dsign} independent of $R$, except at the following contexts. A \emph{critical context} consists of one distinguished class and $k$ ordinary classes of the opposite orientation: $S=\{X_J,Y_{j_1},\dots,Y_{j_k}\}$ or $S=\{Y_J,X_{j_1},\dots,X_{j_k}\}$ with $j_1,\dots,j_k\in[2k]$. It has $k+1=n-1$ classes, so it occurs only for $l=n$, and then $n_B(S)=1$.

\begin{lemma}[Critical contexts]\label{lem:Ctarget}
Let $1\le l\le n$, and let $S$ be a context for $l$ buyers.
\begin{enumerate}[label=(\alph*),nosep]
\item If $S$ is not a critical context, then $W^{(l)}_R(S)$ is a polynomial function of $R\in[\hat T-1,\hat T+1]$ of degree at most $n-1$.
\item If $S=\{X_J,Y_{j_1},\dots,Y_{j_k}\}$, then $W^{(n)}_R(S-X_J)$ does not depend on $R$, and
\[
W^{(n)}_R(S)=1-\nu u_J(R)\,W^{(n)}_R(S-X_J)-\frac\nu2\Big[\sum_{a=1}^kz_{j_a}-R\Big]_+\qquad(R\in[\hat T-1,\hat T+1]);
\]
the mirror image holds for $S=\{Y_J,X_{j_1},\dots,X_{j_k}\}$.
\end{enumerate}
\end{lemma}

The signed coefficient sum of $S$ is $\nu$ times $\Lambda(kn_{\mathrm d}+n_{\mathrm o})+n_{\mathrm d}R+\Sigma_{\mathrm o}$, where $n_{\mathrm d}$ and $n_{\mathrm o}$ count the distinguished and the ordinary occurrences with sign $+1$ for $X$ and $-1$ for $Y$, and $\Sigma_{\mathrm o}$ is the corresponding signed sum of the numbers $z_j$. The last two terms are less than $\Lambda$ in absolute value, so the sign does not depend on $R$ unless $kn_{\mathrm d}+n_{\mathrm o}=0\ne n_{\mathrm d}$, which happens only at critical contexts. Elsewhere, \eqref{eq:scalar} is therefore a linear recursion whose coefficients are constant or affine in $R$, which gives~(a). At a critical context, the identity~\eqref{eq:Dscalar} gives $D(S)=\frac\nu2\big(R-\sum_az_{j_a}\big)$, and~(b) follows from~\eqref{eq:scalar} (\Cref{app:C-count}).

\begin{proposition}[Counting identity]\label{prop:Ccount}
Let $h_R=1/(4n)$ and $R_i=\hat T+(i-1)h_R$ for $0\le i\le n$. Then
\[
\sum_{i=0}^n(-1)^{n-i}\binom ni\,\cW(R_i)=C_{\mathrm{rev}}\,\cN,\qquad C_{\mathrm{rev}}=(-1)^{n+1}\,2n\nu h_R\,(k+1)\,k!\,(\mu d)^{k+1}\,\theta^{J+4r^2+r} .
\]
\end{proposition}

The left side is the $n$-th difference of $\cW$ with step $h_R$. By \Cref{lem:Ctarget}, it annihilates the contributions of all contexts other than the critical contexts, and at a critical context it leaves only the term with $\big[\sum_az_{j_a}-R\big]_+$. Since $\sum_az_{j_a}$ is an integer, this term contributes $(-1)^{n+1}\nu h_R/2$ if $\sum_az_{j_a}=\hat T$ and nothing otherwise. By \Cref{lem:blocks}, the ordered profiles of opponents that form a critical context with $\sum_az_{j_a}=\hat T$ consist of $2$ orientations, $k+1$ positions of the distinguished opponent and $k!\,\cN$ tuples of ordinary classes, and each of them has probability $\mu d\theta^J\cdot(\mu d)^k\theta^{4r^2+r}$ (\Cref{app:C-count}).

\subsubsection{Part~(ii): a hard comparison}\label{sec:C-compare}

In \Cref{sec:C-compare,sec:C-imbalance,sec:C-reveal}, $n=2k+4$ with $k\ge2$, so that the context $S^*$ of \Cref{sec:C-imbalance}, a multiset of $2k+2$ single-item classes, leaves room for two buyers of class $B$ in a report profile; all coefficients are at most $\bar a:=1/(500n^2)$, contexts are those for $n$ buyers, and $W=W^{(n)}$; thus \Cref{lem:Dsign} applies. For $0\le\varepsilon\le\bar a$, let $\bm w_\varepsilon$ be the greedy table~\eqref{eq:greedy} with $n$ buyers and $\kappa=1-\varepsilon$, on the contexts whose coefficients are at least $\varepsilon$, so that their slopes $1-a$ are at most $\kappa$; write $U^\varepsilon_1(S),U^\varepsilon_2(S)$ for its residuals $U_1(S),U_2(S)$. For $\varepsilon=0$, both coordinates of $\bm w_0$ equal the scalar table: $\bm w_0(S)=(W(S),W(S))$. The data of part~(ii) are coefficients: \Cref{sec:C-compare} uses only the scalar table, \Cref{sec:C-imbalance} compares it with $\bm w_\varepsilon$, and \Cref{sec:C-reveal} realizes $\bm w_\varepsilon$ by priors.

The function that carries the hardness comes from contexts with a single $X$-class. Fix coefficients $b_1,\dots,b_{2k}\in(0,\bar a]$ and a further coefficient $q\in(0,\bar a]$. For $I\subseteq[2k]$ write $\mathsf Y_I=\{Y(b_i):i\in I\}$ (a multiset if coefficients repeat) and
\[
W_\circ(I)=W(\mathsf Y_I),\qquad W_q(I)=W\big(\{X(q)\}\cup\mathsf Y_I\big),\qquad \phi_q(I)=\pos{-D\big(\{X(q)\}\cup\mathsf Y_I\big)} ,
\]
all computed by the scalar recursion~\eqref{eq:scalar}; these contexts have at least $n-2k-1=3$ buyers of class $B$. The quantity $W_\circ(I)$ does not depend on $q$.

This is the source of the hardness. By the scalar recursion, $\phi_q(I)$ satisfies a linear recursion over the subsets of $I$, with source term $\pos{\sum_{i\in I}b_i-q}$, and $W_q(I)$ is an affine function of $q$ and $\phi_q(I)$ (\Cref{lem:Dheavy}). So, as functions of $q\in(0,\bar a]$, $\phi_q(I)$ and $W_q(I)$ are continuous and piecewise affine, with kinks only at the subset sums $q=\sum_{i\in I'}b_i$, $I'\subseteq I$; and the jumps of the slope of $\phi_q(I)$ at a point $\hat q$ satisfy the same recursion, with the source term replaced by its slope jump $\ind{\sum_{i\in I}b_i=\hat q}$. They therefore count, with signs and weights, the subsets of $I$ whose sum is $\hat q$.

We now set up the comparison. \emph{Valid data} consist of an integer $k\ge2$ (so $n=2k+4$) and rationals $b_1,\dots,b_{2k}$ and $c$ such that, with $b_0:=1/(1000n^3)$, $\bm b=(b_1,\dots,b_{2k})$ and $\bar c:=\|\bm b\|_1-c$,
\begin{equation}\label{eq:Dvalid}
b_0\le b_i\le 1.01\,b_0\ (i\in[2k]),\qquad |c-kb_0|\le\frac{b_0}2,\qquad |\bar c-kb_0|\le\frac{b_0}2 .
\end{equation}
Then $c,\bar c\ge(k-\frac12)b_0\ge\frac32b_0>1.01b_0\ge\max_ib_i$ and $c,\bar c\le(k+\frac12)b_0\le nb_0=1/(1000n^2)<\bar a$. For valid data, the comparison concerns the first-order term of the capacity imbalance at the context $S^*$ of \Cref{sec:C-imbalance}: by \Cref{prop:Dimbalance}, the residuals at $S^*$ satisfy $U^\varepsilon_1(S^*)-U^\varepsilon_2(S^*)=\varepsilon(\vartheta-\cJ(c))/3$ to first order in the width $\varepsilon$ of the cone, where $\vartheta$ is a threshold of our choice and $\cJ$ is the following function, computed by the scalar recursion for the coefficients $b_1,\dots,b_{2k}$ and the $X$-coefficients $c$ and $\bar c$:
\begin{equation}\label{eq:Jdef}
\varrho_{\bm b}(q)=\frac32\,W_q([2k])+\sum_{j=1}^{2k}W_q([2k]\setminus\{j\}),\qquad \cJ_{\bm b}(c)=1-k-\frac{\|\bm b\|_1}2+c\,\varrho_{\bm b}(\bar c)+\bar c\,\varrho_{\bm b}(c);
\end{equation}
we write $\varrho$ and $\cJ$ when $\bm b$ is clear. The coefficient $\frac32$ comes from the two buyers of class $B$ at $S^*$. Since $0<W\le1$ (\Cref{lem:Dsign}), we have $0<\varrho\le2k+\frac32$ and $|\cJ(c)|\le k-1+\frac{\|\bm b\|_1}2+(2k+2)\|\bm b\|_1<n$.

Optimal mechanisms reveal how $\cJ(c)$ compares with a threshold only when the two differ by a known margin, so the reduction uses comparisons with a tolerance. For a rational $\tau>0$, a \emph{$\tau$-comparator} for $\cJ$ is an oracle that, given valid data and a rational $\vartheta$ with $|\vartheta|\le2n$, answers whether $\cJ(c)>\vartheta$, correctly whenever $|\cJ(c)-\vartheta|\ge\tau$.

\begin{proposition}[Hardness of the comparison]\label{prop:Dhard}
There is a deterministic polynomial-time algorithm that, given an instance of the counting version of \textup{\textsc{SubsetSum}}, computes a rational $\tau\in(0,1]$ and outputs the number of solutions of the instance, using any $\tau$-comparator for $\cJ$.
\end{proposition}

The algorithm encodes the numbers $z_1,\dots,z_{2k}$ and $\hat T$ of~\eqref{eq:blocks4} as the coefficients $b_i(\nu)=b_0+\nu z_i$ and $c(\nu)=kb_0+\nu\hat T$, for a small $\nu>0$; by \Cref{lem:blocks}, $\cN$ is the number of sets $I\subseteq[2k]$ with $|I|=k$ and $\sum_{i\in I}z_i=\hat T$. The common base $b_0$ separates the subset sums of different cardinalities, so near $c(\nu)$, $\cJ$ can have a kink only at $c(\nu)$ itself, coming from these $\cN$ sets and their complements, and is quadratic on each side of it. By the recursion above, the jump of the derivative of $\cJ$ at $c(\nu)$ is a combination of slope jumps that satisfy a linear recursion whose sources indicate these sets; replacing all $b_i(\nu)$ by $b_0$ keeps the sources, makes the recursion explicitly solvable, and shows that the jump differs from an explicit nonzero multiple of $\cN$ by $O(n\nu\bar z)$, where $\bar z=\sum_iz_i$. Three values of $\cJ$ near $c(\nu)$, each approximated by binary search with the comparator, give this jump through a second difference, and for suitable $\nu$, step and $\tau$ of polynomial bit length, rounding gives $\cN$ (\Cref{app:D-compare}).

\subsubsection{Part~(ii): a capacity imbalance}\label{sec:C-imbalance}

We now turn the comparison into an imbalance of capacities. For valid data, a threshold $\vartheta$ with $|\vartheta|\le2n$ and a tolerance $\tau\in(0,1]$, put
\begin{equation}\label{eq:Deps}
\varepsilon:=\min\Big\{\frac{b_0}{100n},\ \frac{\tau}{100n^2}\Big\},\qquad \kappa:=1-\varepsilon,\qquad \bar c_\varepsilon:=\bar c-2\varepsilon\vartheta,\qquad \eta:=\frac{\varepsilon\tau}4 .
\end{equation}
Then $|\bar c_\varepsilon-\bar c|=2\varepsilon|\vartheta|\le4n\varepsilon\le b_0/25$, so $\bar c_\varepsilon\in[(k-0.54)b_0,(k+0.54)b_0]$; in particular, $\bar c_\varepsilon>1.01b_0\ge\max_ib_i$ and $\bar c_\varepsilon<\bar a$. We compare the greedy table $\bm w_\varepsilon$ for the coefficients $c,\bar c_\varepsilon,b_1,\dots,b_{2k}$ at the context
\[
S^*=\{X(c),X(\bar c_\varepsilon)\}\cup\mathsf Y_{[2k]},\qquad |S^*|=2k+2=n-2,\qquad n_B(S^*)=2,
\]
with the scalar table for the coefficients $c,\bar c,b_1,\dots,b_{2k}$ at $S^*_0=\{X(c),X(\bar c)\}\cup\mathsf Y_{[2k]}$. Since $c+\bar c=\|\bm b\|_1$, the signed coefficient sum of $S^*_0$ vanishes, and the scalar loads at $S^*_0$ tie. The numbers in~\eqref{eq:Deps} have bit length polynomial in the input length and in that of $\tau$.

\begin{proposition}[A capacity imbalance at $S^*$]\label{prop:Dimbalance}
For valid data, a threshold $\vartheta$ with $|\vartheta|\le2n$ and $\tau\in(0,1]$, the perturbed greedy table satisfies
\[
\Big|\frac{U^\varepsilon_1(S^*)-U^\varepsilon_2(S^*)}{\varepsilon}-\frac{\vartheta-\cJ(c)}3\Big|\ \le\ 2n^2\varepsilon\ \le\ \frac\tau{50}.
\]
Consequently, $\bm w_\varepsilon(S^*)_1-\bm w_\varepsilon(S^*)_2\ge\eta$ if $\vartheta\ge\cJ(c)+\tau$, and $\bm w_\varepsilon(S^*)_2-\bm w_\varepsilon(S^*)_1\ge\eta$ if $\vartheta\le\cJ(c)-\tau$.
\end{proposition}

The proof (\Cref{app:D-imbalance}) first shows, by \Cref{lem:Wclose}, that the perturbed table stays within $O(n\varepsilon)$ of the scalar table at all sub-multisets of $S^*$, then expands an exact identity for the difference of the loads at $S^*$ to first order in $\varepsilon$, and finally identifies the first-order coefficient with $(\vartheta-\cJ(c))/3$ using the scalar recursion; the factor $\frac13$ comes from $n_B(S^*)=2$.

\subsubsection{Part~(ii): every optimal mechanism reveals the comparison}\label{sec:C-reveal}

Fix valid data, a threshold $\vartheta$ with $|\vartheta|\le2n$ and $\tau\in(0,1]$, with $\varepsilon$, $\bar c_\varepsilon$ and $\eta$ as in~\eqref{eq:Deps}, and let $\beta_1>\dots>\beta_J$ be the distinct values among $1-c,1-\bar c_\varepsilon,1-b_1,\dots,1-b_{2k}$. The realization lemma with these slopes, $\kappa=1-\varepsilon$, $\underline\beta=\frac12$ and $g=\frac{b_0}2$ applies, and $S^*$ is a context of $n-2$ single-item classes of its priors $F_p$. Their capacity problem satisfies $n\chi\le\frac14$, and at the sub-multisets of $S^*$, its greedy table $\bm w^*_n$ is the perturbed table $\bm w_\varepsilon$ (\Cref{app:D-reveal}). The revenues of the prices $1$ and $\ell_0$ are $\rev(1)=1-p$ and $\rev(\ell_0)=\ell_0(1-p\theta)$ (\Cref{lem:revcurve}), and they are equal when
\[
p=\ptie:=\frac d{1-\theta(1-d)} ;
\]
here $d<\ptie<\frac32d$, since $0<\theta<\frac13$. We choose $p$ slightly above $\ptie$, so that the prices $\ell_0$ and $1$ earn almost, but not exactly, the same revenue:
\begin{equation}\label{eq:pH}
p(H)=\frac{\ptie\,(4H-3)}{4H-4+\ptie},\qquad\text{so that}\qquad \rev(\ell_0)-\rev(1)=\frac{d\,\big(1-p(H)\big)}{4(H-1)}>0 .
\end{equation}
For every rational $H\ge2$, we have $\ptie<p(H)<\frac32d$ (\Cref{app:D-reveal}); so \Cref{thm:structure} applies to $F_{p(H)}$ for every rational $H\ge H_0$. Each price above $1$ earns less than the price $1$ per unit of utility left to the value $H$, and each price below $\ell_0$ earns less than $\ell_0$ while leaving more utility to the value $H$. So the frontier~\eqref{eq:frontier} is generated by the prices $1$ and $\ell_0$, and, for rational $H\ge H_0$ and $F_H$ formed from $F_{p(H)}$, $\Gamma_H$ is piecewise linear, with kinks only along the diagonal and the lines $r_2=\zeta r_1$ and $r_1=\zeta r_2$, where
\[
\zeta=\frac{H-\ell_0}{H-1}=1+\frac d{H-1};
\]
$\Gamma_H$ coincides with a linear function $A_Hr_1$ of $r_1$ alone exactly on the region $\{r_2\ge\zeta r_1\}$ and lies below it everywhere, and symmetrically with $A_Hr_2$ on $\{r_1\ge\zeta r_2\}$ (\Cref{lem:hinges} and \Cref{fig:gamma}).

By \Cref{thm:structure}(b), the unused capacities of the menus at $S^*$ average to $\one-\bm w^*_n(S^*)$; by \Cref{prop:Dimbalance}, for $\zeta$ close enough to~$1$, the comparison decides which of the two outer regions contains this average. By \Cref{thm:structure}(c), Jensen's inequality for $\Gamma_H$ holds with equality on these menus; since $\Gamma_H$ lies below the linear function of that region, with equality only on the region, the unused capacities of all of them lie in the same region (\Cref{lem:signs}). The two regions meet only at $\bm y=(1,1)$, where a menu reveals nothing. A report profile with two buyers of class $B$ avoids this point, because supply prevents both of them from receiving item~$1$ with probability close to~$1$.

\begin{figure}[t]
\centering
\begin{tikzpicture}[scale=3.3,every node/.style={font=\small}]
  \def\zet{1.35}
  \fill[black!12] (0,0) -- (1/\zet,1) -- (0,1) -- cycle;
  \fill[black!12] (0,0) -- (1,1/\zet) -- (1,0) -- cycle;
  \draw (0,0) -- (1/\zet,1);
  \draw (0,0) -- (1,1/\zet);
  \draw[dashed,black!55] (0,0) -- (1,1);
  \draw[black!35] (1,0) -- (1,1) -- (0,1);
  \draw[->] (0,0) -- (1.13,0) node[below] {$r_1$};
  \draw[->] (0,0) -- (0,1.13) node[left] {$r_2$};
  \node[below] at (1,0) {$1$};
  \node[left] at (0,1) {$1$};
  \node[font=\footnotesize] at (0.31,0.90) {$\Gamma_H=A_Hr_1$};
  \node[font=\footnotesize] at (0.75,0.12) {$\Gamma_H=A_Hr_2$};
  \node[font=\footnotesize,above] at (1/\zet,1) {$r_2=\zeta r_1$};
  \node[font=\footnotesize,right] at (1,1/\zet) {$r_1=\zeta r_2$};
  \foreach \p in {(0.32,0.58),(0.42,0.62),(0.34,0.50),(0.36,0.54)} \fill[black!50] \p circle (0.4pt);
  \fill (0.36,0.56) circle (0.55pt);
\end{tikzpicture}
\caption{The function $\Gamma_H$ on the plane of unused capacities $\bm r=\one-\bm y$ (not to scale). It is linear between consecutive lines among the diagonal and the lines $r_2=\zeta r_1$ and $r_1=\zeta r_2$; it equals $A_Hr_1$ exactly on the upper shaded region and $A_Hr_2$ exactly on the lower one. If the average of the unused capacities of the menus at $S^*$ (black dot) lies in a shaded region, equality in Jensen's inequality places all of them in this region (gray dots).}
\label{fig:gamma}
\end{figure}

\begin{lemma}[One report profile]\label{lem:readout}
Let $H\ge H_0$ be rational with $g_H:=d/(H-1)<\eta$, and let $M$ be an optimal mechanism for $(F_H;n,2)$, where $F_H$ is formed from $F_{p(H)}$ by~\eqref{eq:FH}. Let $\bm v$ be a profile of base reports in which buyers~$1$ and~$2$ report $(1,1)$, buyers~$3$ and~$4$ report types of the classes $X(c)$ and $X(\bar c_\varepsilon)$, respectively, and buyer~$4+i$ reports a type of the class $Y(b_i)$ ($i\in[2k]$). Then
\[
\psi=\big(x_{11}(\bm v)-x_{12}(\bm v)\big)+\big(x_{21}(\bm v)-x_{22}(\bm v)\big)
\]
satisfies $\psi\ge g_H$ if $\vartheta\ge\cJ(c)+\tau$, and $\psi\le-g_H$ if $\vartheta\le\cJ(c)-\tau$.
\end{lemma}

\begin{proof}
For $i\in\{1,2\}$, the opponents of buyer~$i$ are one buyer of class $B$ and $n-2$ buyers whose classes form $S^*$, so $(i,\bm v_{-i})\in\cA_{n,S^*}$. By \Cref{thm:structure}(a), the menus of $M$ at base reports are canonical, so buyer~$i$, of class $B$, receives her capacities: $\bm x_i(\bm v)=\bm y_i(\bm v_{-i})$. Suppose that $\vartheta\ge\cJ(c)+\tau$. Since $\zeta-1=g_H<\eta$, \Cref{lem:signs} gives $x_{i1}(\bm v)-x_{i2}(\bm v)\ge(\zeta-1)\big(1-x_{i1}(\bm v)\big)$ for $i=1,2$. Adding these two inequalities and using the supply constraint for item~$1$, $\big(1-x_{11}(\bm v)\big)+\big(1-x_{21}(\bm v)\big)\ge1$, gives $\psi\ge\zeta-1=g_H$. The case $\vartheta\le\cJ(c)-\tau$ is symmetric, with item~$2$ in place of item~$1$.
\end{proof}

\subsubsection{Completing the proof}\label{sec:C-proof}

\begin{proof}[Proof of \Cref{thm:C}]
(i) Given an instance $(w_1,\dots,w_r,T)$ with $r\ge1$, compute the numbers of~\eqref{eq:blocks4} and~\eqref{eq:uC}, the parameters $d$, $q_0$ and $\theta$ of \Cref{lem:real}, $\mu$, $C_{\mathrm{rev}}$ and $e_0=|C_{\mathrm{rev}}|/2^{n+3}$, which is less than $1$. All of them are rationals of polynomial bit length, computable in polynomial time, by \Cref{lem:blocks}, \Cref{lem:real}(c) and since $J+4r^2+r=O(r^2)$. For each of the $n+1$ points $R_i$ of \Cref{prop:Ccount}, compute by \Cref{lem:Cread} a number $\tilde\cW_i$ within $e_0$ of $\cW(R_i)$. By \Cref{prop:Ccount}, the number $\sum_i(-1)^{n-i}\binom ni\tilde\cW_i$ is within $2^ne_0=|C_{\mathrm{rev}}|/8$ of $C_{\mathrm{rev}}\cN$, so $\cN$ is the integer nearest to its quotient by $C_{\mathrm{rev}}$. The reduction runs in polynomial time and reads the optimal revenue only through comparisons ``$\OPT(F';n,2)\ge\vartheta$?'', where $n=2r+2$ and $F'$ has $2J+4=8r+6$ support points, all in $(0,1]$. Since the counting version of \textsc{SubsetSum} is $\sharpP$-complete (\Cref{fact:subsetsum}), deciding these comparisons is $\sharpP$-hard, and so is computing $\OPT(F';n,2)$ exactly. By \Cref{lem:Cread}, the comparisons can be replaced by approximations of the numbers $\OPT(F';n,2)$ within an additive error of polynomial bit length. The reduction then computes the number of solutions in polynomial time with any oracle that returns such approximations, which proves the last claim of~(i).

(ii) By \Cref{prop:Dhard,fact:subsetsum}, it suffices to implement, in polynomial time with the oracle, a $\tau$-comparator for $\cJ$, for a given rational $\tau\in(0,1]$. Let valid data and a threshold $\vartheta$ with $|\vartheta|\le2n$ be given. Compute $n$, $b_0$, $\varepsilon$, $\bar c_\varepsilon$ and $\eta$ by~\eqref{eq:Dvalid} and~\eqref{eq:Deps}; the slopes, $d$, $\theta$ and the support as in \Cref{sec:C-reveal}; the number $K_1$ of \Cref{lem:stab} and the numbers $\sigma$ and $H_0$ of \Cref{sec:two-high}; and $H=\max\{H_0,1+2d/\eta\}$, so that $g_H=d/(H-1)\le\eta/2$. All of them have polynomial bit length. Then compute $p(H)$, the prior $F_H$ formed from $F_{p(H)}$, the prior $\hat F_H$ obtained from $F_H$ by dividing all values by $H$, which has $2J+4\le4k+8$ support points, all in $(0,1]$, a profile $\bm v$ as in \Cref{lem:readout}, and the number $g_H$. Dividing all values and payments by $H$ maps the DSIC, ex-post IR mechanisms for $F_H$ bijectively to those for $\hat F_H$, preserving allocations and optimality. Query the oracle, with $e=g_H/16$, at the profile $\bm v/H$ for $(\hat F_H;n,2)$. Its answer gives the four allocation probabilities $x_{11},x_{12},x_{21},x_{22}$, and with them a number $\tilde\psi$ within $4e=g_H/4$ of the number $\psi$ of \Cref{lem:readout}, for the optimal mechanism for $(F_H;n,2)$ that corresponds to the one fixed by the oracle. Answer that $\cJ(c)>\vartheta$ if and only if $\tilde\psi<0$. If $\cJ(c)\ge\vartheta+\tau$, then $\psi\le-g_H$ by \Cref{lem:readout}, so $\tilde\psi<0$; if $\cJ(c)\le\vartheta-\tau$, then $\psi\ge g_H$ and $\tilde\psi>0$. So this is a $\tau$-comparator, and each comparison takes polynomial time and one oracle query. Hence the counting version of \textsc{SubsetSum}, and with it every function in $\sharpP$, can be computed in polynomial time with the oracle. Since the answer is used only within $4e$, the same argument applies to an oracle that returns arbitrary numbers within $e$ of these allocation probabilities.
\end{proof}

\section*{AI methodology}
\addcontentsline{toc}{section}{AI methodology}
We used the large language models GPT-6 Pro, GPT-6 Astra and Claude Opus 5.5 for literature search, translation, drafting and editing, and for part of the mathematical work. The authors conceived Theorems A and B and determined the methods and the key steps of their proofs; GPT-6 Pro wrote out the detailed proofs under the authors' direction. Theorem C grew out of the authors' ideas and was developed in intensive discussion between the authors and GPT-6 Pro, which carried out much of the technical work. Claude Opus 5.5 and GPT-6 Astra were used to draft the text, to improve its presentation, to run numerical experiments and to check the proofs; the experiments tested identities and inequalities on small instances, and no result of the paper relies on them. These tools materially affected the results and proofs of Sections~\ref{sec:A}--\ref{sec:two} and Appendices~\ref{app:A}--\ref{app:two}, and the text of the whole paper. The authors verified the correctness and originality of all content, including all proofs and references, and take full responsibility for the paper.

\clearpage
\appendix
\addtocontents{toc}{\protect\setcounter{tocdepth}{1}}
\section{Auxiliary results}\label{app:tools}

This appendix collects the digit encoding used in the reduction for \Cref{thm:B} (\Cref{sec:B}) and the known result on which all our hardness results rest.

\subsection{A digit encoding of subset sums}\label{app:tools-encoding}

For a set $A$ of indices and numbers $w_i$, write $w(A)=\sum_{i\in A}w_i$.

\begin{lemma}[Encoding subset sums by sums of $r$ numbers]\label{lem:markers}
Let $r\ge1$, let $w_1,\dots,w_r$ be positive integers, and let $W>r\max_iw_i$ and $0\le T<W$ be integers. Put
\[
b_{i,0}=W(r+1)^{i-1},\qquad b_{i,1}=W(r+1)^{i-1}+w_i\quad(i\in[r]),\qquad R=W\sum_{i=1}^{r}(r+1)^{i-1}+T .
\]
The $2r$ numbers $b_{i,e}$ are distinct. An ordered $r$-tuple of these numbers, repetitions allowed, has sum $R$ if and only if it contains exactly one number from each pair $\{b_{i,0},b_{i,1}\}$ and the set $A=\{i:b_{i,1}\text{ is chosen}\}$ satisfies $w(A)=T$; each such $A$ corresponds to exactly $r!$ tuples.
\end{lemma}

\begin{proof}
Each $w_i$ is positive and smaller than $W$, so $b_{i,1}-b_{i,0}=w_i>0$ and $\lfloor b_{i,e}/W\rfloor=(r+1)^{i-1}$; hence the $2r$ numbers are distinct. Consider a tuple, let $d_i$ be the number of its entries from pair $i$, so that $\sum_id_i=r$, and let $q$ be the sum of the offsets $w_i$ of its entries of the form $b_{i,1}$, counted with multiplicity. Its sum is $W\sum_id_i(r+1)^{i-1}+q$, with $0\le q\le r\max_iw_i<W$. If the sum equals $R$, reducing modulo $W$ gives $q=T$, and then $\sum_id_i(r+1)^{i-1}=\sum_i(r+1)^{i-1}$; since $0\le d_i\le r<r+1$, uniqueness of base-$(r+1)$ expansions gives $d_i=1$ for all $i$. Conversely, such tuples have sum $R$. The $r$ entries of such a tuple are distinct, so each $A$ corresponds to $r!$ orderings.
\end{proof}

\subsection{A known result}\label{app:tools-known}

The following result is used as a black box and is not proved here.

\begin{fact}[Counting subset sums~\cite{FH09,HK16}]\label{fact:subsetsum}
Given positive integers $w_1,\dots,w_r$ and $T$ in binary, computing $\#\{A\subseteq[r]:w(A)=T\}$ is $\sharpP$-complete, and $\sharpP$-hard under parsimonious many-one reductions.
\end{fact}

\section{Proofs for Section~\ref*{sec:A}}\label{app:A}

We keep the notation of \Cref{sec:A} and write $\rho=q/p$.

\subsection{An explicit formula for the virtual values and the proof of Lemma~\ref*{lem:monoA}}\label{app:A-mono}

For $c\ge1$ and $r\ge0$, let $d=c+r$ and $R\sim\mathrm{Bin}(d,q)$. For $1\le h\le c$,
\[
\frac{\Pr[R=r+h]}{\Pr[R=r]}=\frac{\binom d{r+h}}{\binom dr}\,\rho^h=\frac{r!\,c!}{(r+h)!\,(c-h)!}\,\rho^h .
\]
Hence
\begin{equation}\label{eq:GammaA}
\Gamma_{c,r}:=\frac{\Pr[R>r]}{c\,\Pr[R=r]}=\sum_{h=1}^{c}\alpha_{c,r,h}\,\rho^h,\qquad \alpha_{c,r,h}=\frac{(c-1)!\,r!}{(c-h)!\,(r+h)!}>0,
\end{equation}
and $\phi_{c,r}=a-\Delta\,\Gamma_{c,r}$. In particular, $\Gamma_{c,r}$ is a polynomial in $\rho$ with positive coefficients; if $p$ is rational, it is a rational number whose bit length is polynomial in $c+r$ and the bit length of $p$.

\begin{proof}[Proof of \Cref{lem:monoA}]
It suffices to prove the corresponding monotonicity of $\Gamma$. Fix $c\ge1$. For $h\ge1$, $\alpha_{c,r+1,h}/\alpha_{c,r,h}=(r+1)/(r+h+1)<1$; since $\rho>0$, each term of $\Gamma_{c,r+1}$ is strictly smaller than the corresponding term of $\Gamma_{c,r}$, so $\Gamma_{c,r+1}<\Gamma_{c,r}$.

Now fix $d\ge2$ and $0\le k\le d-2$, and let $c=d-k\ge2$ and $r=k$. For $1\le h\le c-1$,
\[
\frac{\alpha_{c-1,r+1,h}}{\alpha_{c,r,h}}=\frac{(r+1)(c-h)}{(c-1)(r+h+1)}<1,
\]
because $(c-1)(r+h+1)-(r+1)(c-h)=(h-1)(c+r)+(c-1)>0$. Moreover, the term $h=c$ of $\Gamma_{c,r}$ is positive and has no counterpart in $\Gamma_{c-1,r+1}$. Hence $\Gamma_{d-k-1,k+1}<\Gamma_{d-k,k}$.
\end{proof}

By the second part, $k\mapsto\phi_{d-k,k}$ is strictly increasing on $\{0,\dots,d-1\}$, so the set of $k$ with $\phi_{d-k,k}>0$ is an upper interval of $\{0,\dots,d-1\}$; this gives \eqref{eq:thresholdA}.

\subsection{The flow: definition, conservation and virtual values}\label{app:A-flow}

Fix buyer $i$ and the reports of the others, and use the notation $H$, $D$, $d$, $S=Q\cup T$ and $r=|T|$ of \Cref{sec:A-auction}, and $f_H$ and $R\sim\mathrm{Bin}(d,q)$ of \Cref{sec:A-proof}. Since $\Pr[R=r]=\binom drq^rp^{d-r}$,
\[
f(Q\cup T)=f_H(Q)\,q^rp^{d-r}=\frac{f_H(Q)\Pr[R=r]}{\binom dr}.
\]
All weights are placed on downward IC constraints that delete one element, and on the IR constraint of $\emptyset$:
\begin{enumerate}[label=(\alph*)]
\item (edges deleting an element of $D$) for $T\ne\emptyset$ with $|T|=r$ and $j\in T$,
\[
\lambda(Q\cup T,\,Q\cup T\setminus\{j\})=\frac{f_H(Q)\Pr[R\ge r]}{r\binom dr};
\]
\item (edges deleting an element of $H$; they leave only reports with $T=\emptyset$) process the nonempty sets $Q\subseteq H$ in decreasing order of cardinality, and put, for $j\in Q$,
\[
\lambda(Q,\,Q\setminus\{j\})=\frac1{|Q|}\Big(f_H(Q)+\sum_{j'\in H\setminus Q}\lambda(Q\cup\{j'\},\,Q)\Big);
\]
\item all other edges have weight zero, $\lambda(\emptyset,\bot)=f(\emptyset)+\sum_{j\in[m]}\lambda(\{j\},\emptyset)$, and $\lambda(S,\bot)=0$ for $S\ne\emptyset$.
\end{enumerate}
All weights are nonnegative. For fixed $Q$, the total flow from the sets of count $r$ to those of count $r-1$ is $\binom dr\cdot r\cdot\frac{f_H(Q)\Pr[R\ge r]}{r\binom dr}=f_H(Q)\Pr[R\ge r]$, the whole probability mass of counts at least $r$; this is the one-dimensional Myerson flow on counts described in \Cref{sec:A-proof}.

\paragraph{Conservation.} We verify \eqref{eq:conserv} at every type $S=Q\cup T$.
\begin{itemize}
\item If $r\ge1$, the flow leaves $S$ only along its $r$ edges of type (a), in total $f_H(Q)\Pr[R\ge r]/\binom dr$. If $r<d$, it enters $S$ from the $d-r$ sets $Q\cup T\cup\{j'\}$ with $j'\in D\setminus T$, each contributing
\[
\frac{f_H(Q)\Pr[R\ge r+1]}{(r+1)\binom d{r+1}}=\frac{f_H(Q)\Pr[R\ge r+1]}{(d-r)\binom dr},
\]
because $(r+1)\binom d{r+1}=(d-r)\binom dr$; if $r=d$, no flow enters $S$ along edges of type (a), and $\Pr[R\ge d+1]=0$. No edge of type (b) enters a set with $T\ne\emptyset$. In both cases the outflow minus the inflow is $f_H(Q)\Pr[R=r]/\binom dr=f(S)$.
\item If $r=0$ and $Q\ne\emptyset$, the inflow along edges of type (a) totals $f_H(Q)\Pr[R\ge1]=f_H(Q)(1-p^d)$, and $f(Q)=f_H(Q)p^d$, so the two sum to $f_H(Q)$. Rule (b) sends out exactly this amount plus the inflow along edges of type (b), evenly over the $|Q|$ edges.
\item At $S=\emptyset$, conservation holds by rule (c).
\end{itemize}

\paragraph{Virtual values.} Every edge entering $S$ comes from a set $S\cup\{j\}$ with $j\notin S$, and $\bm t(S\cup\{j\})-\bm t(S)=\Delta\one_{\{j\}}$. By \eqref{eq:Phidef},
\[
\Phi_j(S)=\begin{cases}b,& j\in S,\\[2pt] a-\Delta\,\lambda(S\cup\{j\},S)/f(S),& j\notin S.\end{cases}
\]
For $j\in D\setminus T$, with $c=d-r\ge1$,
\[
\frac{\lambda(S\cup\{j\},S)}{f(S)}=\frac{f_H(Q)\Pr[R\ge r+1]}{(r+1)\binom d{r+1}}\cdot\frac{\binom dr}{f_H(Q)\Pr[R=r]}=\frac{\Pr[R>r]}{c\,\Pr[R=r]},
\]
so $\Phi_j(S)=\phi_{c,r}$. For $j\in H\setminus Q$, $\lambda\ge0$ gives $\Phi_j(S)\le a$.

\subsection{Proof of Lemma~\ref*{lem:menuA}}\label{app:A-menu}

We first justify the description of $\Theta_i$ in \Cref{sec:A-auction}. In the ranking by decreasing exclusive count with ties broken by priority, buyer $i$ is first if and only if $e_i>e_\ell$ for every $\ell<i$ and $e_i\ge e_\ell$ for every $\ell>i$, that is, $e_i\ge e_\ell+\ind{\ell<i}$ for every $\ell\ne i$. Since $d=c+e_i$, it follows that $e_i\ge\Theta_i$ if and only if she is first and $e_i\ge\theta_{c+e_i}$; when $c\ge1$, this is exactly the event that she receives the all-low items.

\begin{proof}[Proof of \Cref{lem:menuA}]
Let buyer $i$ report $Q\cup T$. Every item of $Q$ is reported high by some opponent and no item of $T$ is, so $e_i=r$ and $c=d-r$. Since $T\subseteq D$ is disjoint from every $S_\ell$, the exclusive items of opponent $\ell$ are those of $C_\ell\setminus Q$; hence $\Theta_i=\Theta(Q)$.

If $r<\Theta(Q)$, buyer $i$ does not receive the all-low items. She receives the items that she reports high and for which she is the first high reporter in the ranking: all of $T$, which only she reports high, and the items of $Q$ that no buyer ranked before her reports high. Since $\pos{e_i-\Theta_i}=0$, she pays her reported value $b$ for each of them.

If $r\ge\Theta(Q)$, she is ranked first, so she is the first high reporter of every item of $Q\cup T$, and she receives the all-low items $D\setminus T$; thus $A_i=D\cup Q$. Her payment is
\[
b|Q|+br+a(d-r)-\Delta\big(r-\Theta(Q)\big)=ad+b|Q|+\Delta\,\Theta(Q).
\]

For \eqref{eq:LipA}, let $Q,Q'\subseteq H$. Since $C_\ell\setminus Q\subseteq(C_\ell\setminus Q')\cup(Q'\setminus Q)$, we have $e_\ell(Q)\le e_\ell(Q')+|Q'\setminus Q|$, while the reserve $\theta_d$ and the offsets $\ind{\ell<i}$ do not depend on $Q$. Taking maxima gives \eqref{eq:LipA}.
\end{proof}

\subsection{Computing the revenue}\label{app:A-compute}

Let $\varpi_{c,r}=\Pr[\,|L|=c,\ e^*=r\,]$. By \eqref{eq:revA},
\begin{equation}\label{eq:revA2}
\OPT(F;n,m)=mb(1-p^n)+\sum_{c=1}^{m}\sum_{r=0}^{m-c}c\,\varpi_{c,r}\pos{\phi_{c,r}} .
\end{equation}
Items are independent. A single item is all-low with probability $\pi_0=p^n$, is reported high by buyer $i$ alone with probability $\pi_1=qp^{n-1}$ for each $i$, and is reported high by at least two buyers with probability $\pi_2=1-p^n-nqp^{n-1}$. For nonnegative integers with $c+c_2+e_1+\dots+e_n=m$, the probability that $c$ items are all-low, $c_2$ items are reported high by at least two buyers, and $e_i$ items are reported high by buyer $i$ alone is
\[
\frac{m!}{c!\,c_2!\,e_1!\cdots e_n!}\,\pi_0^{c}\,\pi_2^{c_2}\,\pi_1^{e_1+\dots+e_n}.
\]
Let $w$ be a formal variable, $[w^k]$ the coefficient of $w^k$, $\mathcal E_r(w)=\sum_{k=0}^{r}(\pi_1w)^k/k!$ with $\mathcal E_{-1}=0$, and $\mathcal G(w)=\sum_{c_2=0}^{m}(\pi_2w)^{c_2}/c_2!$, with the convention $0^0=1$ (for $n=1$, $\pi_2=0$). Then
\begin{equation}\label{eq:varpiA}
\varpi_{c,r}=\frac{m!\,\pi_0^{c}}{c!}\,[w^{m-c}]\;\mathcal G(w)\big(\mathcal E_r(w)^n-\mathcal E_{r-1}(w)^n\big).
\end{equation}
Indeed, the coefficient of $w^{m-c}$ in $\mathcal G(w)\mathcal E_r(w)^n$ is the sum of
\[
\frac{\pi_2^{c_2}\,\pi_1^{e_1+\dots+e_n}}{c_2!\,e_1!\cdots e_n!}
\]
over all $(c_2,e_1,\dots,e_n)$ with $c_2+\sum_ie_i=m-c$ and $e_i\le r$ for every $i$; subtracting the same sum with $e_i\le r-1$ for every $i$ leaves exactly the count vectors with $\max_ie_i=r$. The polynomials may be truncated after degree $m$, and all coefficients are rationals of polynomial bit length, so \eqref{eq:varpiA} and \eqref{eq:revA2} are computable in polynomial time. Evaluating the auction at a profile requires the thresholds $\theta_d$ for $d\le m$; by \eqref{eq:GammaA}, each of them requires polynomially many rationals $\Gamma_{d-k,k}$ of polynomial bit length.

\section{Proofs for Section~\ref*{sec:B}}\label{app:B}

We keep the notation of \Cref{sec:B}. The constants of \Cref{lem:pointwiseB} are
\begin{equation}\label{eq:constsB}
C_m=8m^2\,3^m,\qquad B_m=\frac{(m+1)(V+1)\,C_m}{\nu},
\end{equation}
where $V=m(a+1)$ and $\nu$ is the parameter of \eqref{eq:sepB}. For the prior \eqref{eq:FB}, we take
\begin{equation}\label{eq:etaB}
\nu=\frac1{2\Lambda}=\frac1{8m^2b_{\max}},\qquad \eta=\frac{\nu}{256\,m^3(m+1)\,6^m}.
\end{equation}

\subsection{Proof of Lemma~\ref*{lem:tailB}}\label{app:B-tail}

The numbers $b_\ell$ are distinct (\Cref{lem:markers}), so index tuples correspond to tuples of numbers, and $\cZ(R)-\cZ(R+1)$ is the number of tuples whose sum is exactly $R$. As $W=1+m\sum_iw_i>m\max_iw_i$ and $T\le\sum_iw_i<W$, \Cref{lem:markers} with $r=m$ shows that each solution $A$ corresponds to exactly $m!$ such tuples, so $\cN=(\cZ(R)-\cZ(R+1))/m!$. Clearly $R\ge1$. Since every $w_i$ is smaller than $W$, the largest of the numbers is $b_{\max}=W(m+1)^{m-1}+w_m$. Since $m\ge2$,
\[
m(m+1)^{m-1}-\sum_{i=1}^m(m+1)^{i-1}=\sum_{i=1}^{m-1}\big((m+1)^{m-1}-(m+1)^{i-1}\big)\ge1,
\]
so $mb_{\max}-R\ge W+mw_m-T\ge W-T\ge1$, that is, $R+1\le mb_{\max}$. All these integers have $O(m\log m+\log W)$ bits.

\subsection{The perturbed prior}\label{app:B-encoding}

We verify the properties of $F$ stated in \Cref{sec:B-construction} and derive \eqref{eq:encrevB}.

\begin{lemma}\label{lem:incrB}
The numbers $\zeta_\ell$ are distinct, $|\zeta_\ell-1/m|\le1/(4m^2)$, and $0<\zeta_\ell<1$. The set $\cT=\supp(F)^m$ satisfies \eqref{eq:sepB} with the parameter $\nu$ of \eqref{eq:etaB}. Moreover, $\sum_{j=1}^m\zeta_{i_j}\ge1$ if and only if $\sum_{j=1}^mb_{i_j}\ge T_0$.
\end{lemma}

\begin{proof}
Both $b_\ell$ and $(T_0-\frac12)/m$ lie in $[0,b_{\max}]$, so $|\zeta_\ell-1/m|\le b_{\max}/\Lambda=1/(4m^2)$. Hence $\zeta_\ell\ge(4m-1)/(4m^2)>0$ and, as $m\ge2$, $\zeta_\ell\le\frac12+\frac1{16}<1$. Distinct $b_\ell$ give distinct $\zeta_\ell$.

Each excess $\xi_j=t_j-a$ lies in $\{0,1,\zeta_1,\dots,\zeta_N\}$. If some excess equals $1$, then $\|\bm\xi\|_1\ge1$. If no excess equals $1$ and at most $m-1$ items have intermediate values, then
\[
\|\bm\xi\|_1\le(m-1)\Big(\frac1m+\frac1{4m^2}\Big)=1-\frac{3m+1}{4m^2}\le1-\frac3{4m}\le1-\nu,
\]
because $\nu=1/(8m^2b_{\max})\le3/(4m)$. If all $m$ items have intermediate values, say $\xi_j=\zeta_{i_j}$, then
\[
\|\bm\xi\|_1=\sum_j\zeta_{i_j}=1+\frac{\sum_jb_{i_j}-T_0+\frac12}{\Lambda},
\]
and the numerator lies in $\Z+\frac12$: if $\sum_jb_{i_j}\ge T_0$, it is at least $\frac12$, and $\|\bm\xi\|_1\ge1+\nu$; otherwise it is at most $-\frac12$, and $\|\bm\xi\|_1\le1-\nu$. This proves both the separation and the equivalence.
\end{proof}

\begin{lemma}\label{lem:FpropsB}
The distribution $F$ has exactly $N+2$ support points, all positive and rational, with positive rational probabilities summing to one; moreover, $0<2^{-m}-\eta^m<d_0<2^{-m}$, $P>2^m$ and $a>0$. The number $p=d_0^{1/m}$ lies in $[1/3,2/3]$, $\|f-f_0\|_1=2\eta$, and $mB_m\|f-f_0\|_1\le1/4$; hence $\chi\ge3/4$ in \Cref{thm:stabB}.
\end{lemma}

\begin{proof}
By \eqref{eq:etaB}, $\Lambda\ge16$, $0<\nu\le1/32$ and $0<\eta<\nu<1/6$; moreover, $0<\vartheta<1$. This gives the bounds on $d_0$, hence $P=1/d_0>2^m$ and $a=(P-1)/m>0$. By \Cref{lem:incrB}, the intermediate values $a+\zeta_\ell$ are distinct and lie strictly between $a$ and $a+1$, and the probabilities are positive rationals with sum $(\frac12-\eta)+\frac12+\eta=1$.

Since $p^m=d_0<2^{-m}$, we have $p<\frac12$. The binomial expansion gives $2^{-m}=\big((\frac12-\eta)+\eta\big)^m\ge(\frac12-\eta)^m+\eta^m$, so $p^m=d_0\ge2^{-m}-\eta^m\ge(\frac12-\eta)^m$ and $p\ge\frac12-\eta>\frac13$. Relative to $F_0$, the probabilities at $a$ and $a+1$ decrease by $p-(\frac12-\eta)\ge0$ and $(1-p)-\frac12\ge0$, which sum to $\eta$, and the intermediate values receive total probability $\eta$; hence $\|f-f_0\|_1=2\eta$.

Finally, $\eta<\frac14$ gives $\eta^m\vartheta<4^{-m}$, so $d_0\ge2^{-m}-4^{-m}\ge2^{-(m+1)}$ and $P\le2^{m+1}$. Then $V+1=m(a+1)+1=P+m\le2^{m+2}$, so
\[
B_m\le\frac{(m+1)\,2^{m+2}\cdot8m^2\,3^m}{\nu}=\frac{32\,m^2(m+1)\,6^m}{\nu}
\]
and
\[
mB_m\|f-f_0\|_1=2mB_m\eta\le2m\cdot\frac{32\,m^2(m+1)\,6^m}{\nu}\cdot\frac{\nu}{256\,m^3(m+1)\,6^m}=\frac14 .\qedhere
\]
\end{proof}

We now derive \eqref{eq:encrevB}. By \Cref{lem:incrB}, $\|\bm\xi\|_1\ge1$ holds exactly when some item has value $a+1$, or all items have intermediate values $a+\zeta_{i_1},\dots,a+\zeta_{i_m}$ with $\sum_jb_{i_j}\ge T_0$. These two events are disjoint. The first has probability $1-2^{-m}$. The second has probability $\eta^m\tau$, because all items have intermediate values with probability $\eta^m$, and conditionally on this event the indices $i_1,\dots,i_m$ are independent and uniform on $[N]$. Hence
\[
\pi_<:=\Pr_F[\|\bm\xi\|_1<1]=2^{-m}-\eta^m\tau .
\]
By \eqref{eq:MzB},
\[
\Rev_F(M_z)=\pi_<\,maz+(1-\pi_<)(P-z)=P(1-\pi_<)+z\,(P\pi_<-1).
\]
Here $P\pi_<-1=P(\pi_<-d_0)=P\eta^m(\vartheta-\tau)$, and $P(1-\pi_<)=ma-(P\pi_<-1)$ because $P-1=ma$. Substituting gives \eqref{eq:encrevB}.

\subsection{Proofs of Lemma~\ref*{lem:PhiB} and Proposition~\ref*{prop:lossB}}\label{app:B-loss}

The flow of \Cref{sec:B-stability} is the flow of \Cref{app:A-flow} for the prior $F_0$ and a single buyer. Hence it satisfies the conservation condition \eqref{eq:conserv}, which reads $k\lambda_k-(m-k)\lambda_{k+1}=f_0(A)$ for $|A|=k\ge1$, with $\lambda_{m+1}=0$. Moreover, by the computation of the virtual values in \Cref{app:A-flow}, $\lambda_{k+1}/f_0(A)=\Gamma_{m-k,k}$ for $|A|=k<m$, where $\Gamma$ is given by \eqref{eq:GammaA} with $\rho=(1-p)/p$; hence $\Phi_j(A)=a-\Gamma_{m-|A|,|A|}$ for $j\notin A$.

\begin{proof}[Proof of \Cref{lem:PhiB}]
For $A=\emptyset$, $\Gamma_{m,0}=\frac{1-p^m}{m\,p^m}=\frac{p^{-m}-1}{m}=a$ by \eqref{eq:tieB}, so $\Phi_j(\emptyset)=0$. For $j\in A$, $\Phi_j(A)=a+1>1$. Let $j\notin A$ with $|A|=k\ge1$. By the second part of \Cref{lem:monoA}, $\Gamma_{m-k,k}\le\Gamma_{m-1,1}$, so $\Phi_j(A)\ge\Gamma_{m,0}-\Gamma_{m-1,1}$. Compare the expansions \eqref{eq:GammaA} of these two numbers: the coefficient of $\rho$ is $\alpha_{m,0,1}=1$ in $\Gamma_{m,0}$ and $\alpha_{m-1,1,1}=\frac12$ in $\Gamma_{m-1,1}$; for $2\le h\le m-1$, the coefficient of $\rho^h$ is smaller in $\Gamma_{m-1,1}$ than in $\Gamma_{m,0}$ (proof of \Cref{lem:monoA}); and $\Gamma_{m,0}$ has the additional positive term $\alpha_{m,0,m}\rho^m$. Hence $\Gamma_{m,0}-\Gamma_{m-1,1}\ge\rho/2\ge\frac14$, using $p\le\frac23$.
\end{proof}

\begin{proof}[Proof of \Cref{prop:lossB}]
The restriction of $M$ to the vertices is an IC, IR mechanism for the prior $F_0$. Apply \Cref{lem:duality} to it and the flow above. The constraints of positive weight are the IC constraints from $\bm t^A$ to $\bm t^{A\setminus\{j\}}$, whose slacks are the numbers $\sigma(A,j)$, and the IR constraint of $\bm t^\emptyset$, which is tight because $u(\emptyset)=0$. Hence the identity \eqref{eq:slackid} reads
\begin{equation}\label{eq:exactidB}
\Rev_{F_0}(M)=\sum_{A\ne\emptyset}f_0(A)\sum_j\Phi_j(A)\,x_j(A)-\sum_{A\ne\emptyset}\sum_{j\in A}\lambda_{|A|}\,\sigma(A,j),
\end{equation}
in which the terms for $A=\emptyset$ are absent because $\Phi_j(\emptyset)=0$. For $M_z$, we have $x_j(A)=1$ for $A\ne\emptyset$, and every slack vanishes: $u(A)=|A|-1+z$ for $A\ne\emptyset$ and $u(\emptyset)=0$ give $\sigma(\{j\},j)=z-0-z=0$ and $\sigma(A,j)=1-1=0$ for $|A|\ge2$. Thus $ma=\Rev_{F_0}(M_z)=\sum_{A\ne\emptyset}f_0(A)\sum_j\Phi_j(A)$, and subtracting \eqref{eq:exactidB} for $M$ gives \eqref{eq:lossB}. Every summand is nonnegative, because $\Phi_j(A)\ge\frac14$ by \Cref{lem:PhiB}, $x_j\le1$, $\lambda_k>0$, and $\sigma\ge0$ by IC.

By \Cref{lem:norm}, every IC, IR mechanism for $F_0$ earns at most the revenue of one with $u(\emptyset)=0$, which is at most $ma$ by \eqref{eq:lossB} applied with $F=F_0$; since $M_z$ attains $ma$, $\OPT(F_0;1,m)=ma$. Finally, since $\Phi_j(A)>0$ for $A\ne\emptyset$ and $\lambda_k>0$, the loss vanishes if and only if $x_j(A)=1$ for all $A\ne\emptyset$ and all $j$, and all slacks vanish.
\end{proof}

\subsection{Proof of Lemma~\ref*{lem:pointwiseB}}\label{app:B-pointwise}

Since $p,1-p\ge\frac13$, we have $f_0(A)\ge3^{-m}$, and conservation gives $k\lambda_k\ge f_0(A)$ for $|A|=k$, so $\lambda_k\ge1/(m3^m)$. Keeping one summand of \eqref{eq:lossB} at a time and using \Cref{lem:PhiB},
\begin{equation}\label{eq:singleB}
1-x_j(A)\le4\cdot3^m\varepsilon\quad(A\ne\emptyset),\qquad 0\le\sigma(A,j)\le m3^m\varepsilon\quad(j\in A).
\end{equation}
We first show that, for all $j\in[m]$ and $A\ne\emptyset$,
\begin{equation}\label{eq:pw1B}
|x_j(\emptyset)-z|\le C_m\varepsilon,\qquad 1-x_j(A)\le C_m\varepsilon,\qquad \big|u(A)-(|A|-1+z)\big|\le C_m\varepsilon .
\end{equation}
For distinct $i$ and $j$, expanding $u(\{i,j\})$ along the paths $\emptyset,\{i\},\{i,j\}$ and $\emptyset,\{j\},\{i,j\}$ with the definition of $\sigma$ gives
\[
u(\{i,j\})=x_i(\emptyset)+\sigma(\{i\},i)+x_j(\{i\})+\sigma(\{i,j\},j)=x_j(\emptyset)+\sigma(\{j\},j)+x_i(\{j\})+\sigma(\{i,j\},i).
\]
By \eqref{eq:singleB}, $x_j(\{i\})$ and $x_i(\{j\})$ lie in $[1-4\cdot3^m\varepsilon,1]$ and each slack lies in $[0,m3^m\varepsilon]$, so $|x_i(\emptyset)-x_j(\emptyset)|\le(4+2m)3^m\varepsilon\le C_m\varepsilon$; with $i=1$, this is the first bound. The second bound follows from \eqref{eq:singleB}. For the third, let $A=\{j_1,\dots,j_k\}$ and $A_h=\{j_1,\dots,j_h\}$. Expanding along the path $\emptyset,A_1,\dots,A_k$ gives
\[
u(A)=x_{j_1}(\emptyset)+\sum_{h=2}^kx_{j_h}(A_{h-1})+\sum_{h=1}^k\sigma(A_h,j_h),
\]
and comparing with $k-1+z$ term by term gives $|u(A)-(k-1+z)|\le[(4+2m)+4(k-1)+km]\,3^m\varepsilon\le(m^2+6m)\,3^m\varepsilon\le C_m\varepsilon$.

Next, we bound the payments at the vertices. For $A\ne\emptyset$, $s_z(\bm t^A)=\bm t^A\cdot\one-(|A|-1+z)$, so
\[
s(\bm t^A)-s_z(\bm t^A)=-\sum_jt^A_j\big(1-x_j(A)\big)-\big(u(A)-(|A|-1+z)\big);
\]
the coordinates of $\bm t^A$ are nonnegative with sum at most $V$, so $|s(\bm t^A)-s_z(\bm t^A)|\le(V+1)C_m\varepsilon$ by \eqref{eq:pw1B}. At $\emptyset$, $s(\bm t^\emptyset)-s_z(\bm t^\emptyset)=a\sum_j(x_j(\emptyset)-z)$, which is at most $maC_m\varepsilon\le(V+1)C_m\varepsilon$ in absolute value. Together with \eqref{eq:pw1B}, this proves~(a), and the case $\varepsilon=0$ follows.

For~(b), fix $\bm t=a\one+\bm\xi\in\cT$ and let $\bar x=\max_jx_j(\bm t)$. Since $\bm t^\emptyset$ does not gain by reporting $\bm t$, $s(\bm t)\ge a\|\bm x(\bm t)\|_1$. Since $\bm t^{\{j\}}$ does not gain by reporting $\bm t$, \eqref{eq:pw1B} gives
\[
s(\bm t)\ge a\|\bm x(\bm t)\|_1+x_j(\bm t)-u(\{j\})\ge a\|\bm x(\bm t)\|_1+x_j(\bm t)-z-C_m\varepsilon .
\]
Choosing $j$ with $x_j(\bm t)=\bar x$ and combining the two bounds gives $s(\bm t)\ge a\|\bm x(\bm t)\|_1+\pos{\bar x-z-C_m\varepsilon}$. Since $u(\bm t)=a\|\bm x(\bm t)\|_1+\bm\xi\cdot\bm x(\bm t)-s(\bm t)$, $\bm\xi\ge\bm 0$ and $x_j(\bm t)\le\bar x$, it follows that
\begin{equation}\label{eq:uupperB}
u(\bm t)\le\bar x\,\|\bm\xi\|_1-\pos{\bar x-z-C_m\varepsilon}.
\end{equation}
In the other direction, $\bm t$ does not gain by reporting $\bm t^\emptyset$, so by \eqref{eq:pw1B} and $\xi_j\le1$,
\begin{equation}\label{eq:ulowerB}
u(\bm t)\ge\bm\xi\cdot\bm x(\emptyset)\ge z\|\bm\xi\|_1-mC_m\varepsilon .
\end{equation}
Moreover, $\bm t$ does not gain by reporting $\bm t^{\{1\}}$. Using $u(\{1\})=x_1(\emptyset)+\sigma(\{1\},1)\ge z$, $x_1(\{1\})\le1$ (its coefficient $\xi_1-1$ is nonpositive) and $x_j(\{1\})\ge1-C_m\varepsilon$ for $j\ge2$, we get
\[
u(\bm t)\ge u(\{1\})+(\bm\xi-\one_{\{1\}})\cdot\bm x(\{1\})\ge z+(\xi_1-1)+\sum_{j\ge2}\xi_j-(m-1)C_m\varepsilon\ge\|\bm\xi\|_1-1+z-mC_m\varepsilon .
\]

If $\|\bm\xi\|_1\ge1$, then $s_z(\bm t)=P-z$, and the last bound gives
\[
s(\bm t)=\bm t\cdot\bm x(\bm t)-u(\bm t)\le(ma+\|\bm\xi\|_1)-(\|\bm\xi\|_1-1+z-mC_m\varepsilon)=P-z+mC_m\varepsilon=s_z(\bm t)+mC_m\varepsilon .
\]
If $\|\bm\xi\|_1<1$, then $s_z(\bm t)=maz$. We first show that $\bar x-z\le(m+1)C_m\varepsilon/\nu$. This is clear if $\bar x\le z+C_m\varepsilon$, since $\nu\le1$. Otherwise, \eqref{eq:uupperB} and \eqref{eq:ulowerB} give $z\|\bm\xi\|_1-mC_m\varepsilon\le\bar x\|\bm\xi\|_1-\bar x+z+C_m\varepsilon$, that is, $(1-\|\bm\xi\|_1)(\bar x-z)\le(m+1)C_m\varepsilon$, and $1-\|\bm\xi\|_1\ge\nu$ by \eqref{eq:sepB}. Hence, using $\bm t\cdot\bm x(\bm t)\le(ma+\|\bm\xi\|_1)\,\bar x$, \eqref{eq:ulowerB} and $ma+\|\bm\xi\|_1\le V$,
\begin{align*}
s(\bm t)&\le(ma+\|\bm\xi\|_1)\,\bar x-(z\|\bm\xi\|_1-mC_m\varepsilon)=maz+(ma+\|\bm\xi\|_1)(\bar x-z)+mC_m\varepsilon\\
&\le maz+\Big(\frac{(m+1)V}{\nu}+m\Big)C_m\varepsilon\le maz+B_m\varepsilon,
\end{align*}
where the last step uses $m\le(m+1)/\nu$. Since $B_m\ge mC_m$, we have $s(\bm t)\le s_z(\bm t)+B_m\varepsilon$ in both cases.

\subsection{Proof of Theorem~\ref*{thm:stabB}}\label{app:B-stab}

Write $f(\bm t)=\prod_jf(t_j)$ and $f_0(\bm t)=\prod_jf_0(t_j)$, and let $\cT_0=\{a,a+1\}^m$ be the set of vertices; $f_0$ is supported on $\cT_0\subseteq\cT$. We first show
\begin{equation}\label{eq:telescopeB}
\sum_{\bm t\in\cT}|f(\bm t)-f_0(\bm t)|\le m\,\|f-f_0\|_1 .
\end{equation}
For $\bm t\in\cT$,
\[
\prod_jf(t_j)-\prod_jf_0(t_j)=\sum_{j=1}^m\Big[\prod_{i<j}f_0(t_i)\Big]\big(f(t_j)-f_0(t_j)\big)\Big[\prod_{i>j}f(t_i)\Big].
\]
Take absolute values and sum over $\bm t\in\cT$: the $j$-th term contributes the product of $\big(\sum_{a'}f_0(a')\big)^{j-1}$, $\|f-f_0\|_1$ and $\big(\sum_{a'}f(a')\big)^{m-j}$, where both sums run over $a'\in\supp(F)$ and equal $1$.

Now let $z=z(M)$ and $\varepsilon=\varepsilon(M)$. Since $\Rev_{F_0}(M)=ma-\varepsilon$, $\Rev_{F_0}(M_z)=ma$ by \eqref{eq:tiedrevB}, and $f_0$ vanishes outside $\cT_0$,
\begin{align*}
\Rev_F(M)-\Rev_F(M_z)
&=-\varepsilon+\sum_{\bm t\in\cT_0}\big(f(\bm t)-f_0(\bm t)\big)\big(s(\bm t)-s_z(\bm t)\big)+\sum_{\bm t\in\cT\setminus\cT_0}f(\bm t)\big(s(\bm t)-s_z(\bm t)\big)\\
&\le-\varepsilon+(V+1)C_m\varepsilon\sum_{\bm t\in\cT_0}|f(\bm t)-f_0(\bm t)|+B_m\varepsilon\sum_{\bm t\in\cT\setminus\cT_0}f(\bm t)\\
&\le-\varepsilon+B_m\varepsilon\sum_{\bm t\in\cT}|f(\bm t)-f_0(\bm t)|\ \le\ -\chi\,\varepsilon .
\end{align*}
The first inequality uses \Cref{lem:pointwiseB}(a) on $\cT_0$, in its two-sided form because $f-f_0$ may take either sign there, and \Cref{lem:pointwiseB}(b) on $\cT\setminus\cT_0$. The last two use $(V+1)C_m\le B_m$, $f_0=0$ on $\cT\setminus\cT_0$, \eqref{eq:telescopeB} and the definition of $\chi$. This proves \eqref{eq:stabB}.

Every $M_z$ is IC and IR on $\cT$. By \Cref{lem:norm}, every IC, IR mechanism on $\cT$ earns at most the revenue of one with $u(\emptyset)=0$, which in turn earns at most the revenue of some $M_z$ by \eqref{eq:stabB}. Hence $\OPT(F;1,m)=\max_{0\le z\le1}\Rev_F(M_z)$, where the maximum exists because $\Rev_F(M_z)$ is linear in $z$. If $M$ is optimal, then $u(\emptyset)=0$, and \eqref{eq:stabB} gives $\Rev_F(M_{z(M)})\ge\Rev_F(M)+\chi\,\varepsilon(M)$; since $\chi>0$, it follows that $\varepsilon(M)=0$ and that $M_{z(M)}$ is optimal, and by \Cref{lem:pointwiseB}(a), $M$ agrees with $M_{z(M)}$ at every vertex.

\subsection{Bit lengths}\label{app:B-bits}

The numbers $\zeta_\ell$ and $\nu$ have numerators and denominators of $O(\log m+\log b_{\max})$ bits. We have $\eta=1/E_\eta$ with $E_\eta=2048\,m^5(m+1)\,b_{\max}\,6^m$, which has $O(m+\log b_{\max})$ bits, and $\vartheta$ has $O(m\log N)$ bits. The number $d_0=2^{-m}-\eta^m\vartheta$ has the common denominator $2^mN^mE_\eta^m$, so $d_0$, $P$ and $a$ have $O(m\log N+m\log E_\eta)$ bits. Hence $F$ and its rescaling by $1/(a+1)$ can be computed in polynomial time.

\section{Proofs for Section~\ref*{sec:two}}\label{app:two}

We keep the notation of \Cref{sec:two}. \Cref{app:realization} proves the realization lemma, \Cref{app:capacity} the statements of \Cref{sec:two-tables} on capacity tables, \Cref{app:high} those of \Cref{sec:two-high} on the very high value, and \Cref{app:scalar} those of \Cref{sec:two-scalar} on the scalar recursion. \Cref{app:C} contains the proofs for \Cref{sec:C}. Results that \Cref{sec:two} describes in words are stated here formally, next to their proofs.

\subsection{The realization lemma}\label{app:realization}

This subsection proves \Cref{lem:real,lem:stab}. \Cref{app:realparam} fixes the parameters and the support. \Cref{app:realgeom} shows that the canonical menu~\eqref{eq:canon} is IC and IR, that it gives every type the outcome of its class whatever the capacities, and that its revenue is $\gamma_p\|\bm y\|_1$; computing the class probabilities proves \Cref{lem:real}(b) and the bounds on $\gamma_p$. \Cref{app:realflow} constructs the flow, and \Cref{app:realeq} derives the bound $\gamma_p\|\bm y\|_1$, \Cref{lem:stab}, which contains the equality case of \Cref{lem:real}(a), and~(c). Throughout, $p\in[d/2,3d/2]$ is fixed, and we write $F=F_p$ and $f=f_p$.

\subsubsection{Parameters and support}\label{app:realparam}

Let
\begin{equation}\label{eq:params}
\begin{gathered}
d=\frac{g\underline\beta^2}{10^{18}(J+1)^6},\qquad h=\frac{d}{10^6(J+1)^2},\qquad \theta=\frac{d^2\underline\beta^2}{10^{18}(J+1)^6},\\
q_0=\frac1{10^5(J+1)},\qquad \mu=\frac{1-p-q_0}{J+1}.
\end{gathered}
\end{equation}
The numbers $d,h,\theta,q_0$ depend only on $(J,\underline\beta,g)$, and $p$ enters only through $\mu$. Put $d_\kappa=d/\kappa$. The support values are
\begin{equation}\label{eq:values}
\ell_0=1-d,\quad c_0=1+d_\kappa,\qquad c_j=c_{j-1}+h,\quad \ell_j=\ell_{j-1}-\beta_jh\quad(1\le j\le J),\qquad\text{and }1,
\end{equation}
with masses
\begin{equation}\label{eq:masses}
f(\ell_j)=p(1-\theta)\theta^j\ \ (0\le j<J),\quad f(\ell_J)=p\theta^J,\quad f(1)=q_0,\quad f(c_j)=\mu\ \ (0\le j\le J).
\end{equation}
The masses sum to $p+q_0+(J+1)\mu=1$. The parameters play the following roles: $d$ is the scale on which the values deviate from~$1$; $h$ is the spacing of consecutive high values, and the spacing $\beta_jh$ of consecutive low values makes the segment between consecutive corners $(c_{j-1},\ell_{j-1})$ and $(c_j,\ell_j)$ have slope $-\beta_j$; the mass $q_0$ of the value $1$ is small, which keeps the inflow of the corner $(c_0,\ell_0)$ small, but $q_0^2$ is large compared with $d$, so that the virtual values of $(1,1)$ are positive (\Cref{lem:flow}) and the prices $c_j$ earn less than the price $1$ (\Cref{lem:revcurve}); and the masses of the low values decrease rapidly in powers of $\theta$, so that the flow that each layer of single-item types passes on to the previous one is small.

\begin{lemma}[Numerical facts]\label{lem:numerics}
Under the hypotheses of \Cref{lem:real} (so that $g,\underline\beta\le1$, since $\underline\beta\le\beta_J\le\beta_1\le\kappa-g<1$), with $p\in[d/2,3d/2]$:
\begin{enumerate}[label=(N\arabic*),leftmargin=3em]
\item $d\le10^{-19}$, $d\le10^{-18}\underline\beta$, $h\le d/(4\cdot10^6)$, $(J+1)h\le d/(2\cdot10^6)$, $\theta\le d^2\cdot10^{-19}$, $\theta\le h\,d\,\underline\beta/10^{12}$, and $\theta/(\underline\beta h)\le10^{-12}$.
\item $p\le2d$, $q_0\le5\cdot10^{-6}$, $1-p-q_0\ge0.9999$, $0.99/(J+1)\le\mu\le1/(J+1)\le1/2$, $\mu/q_0\le10^5$, $d\le10^{-8}q_0^2$, and $d_\kappa\le2d$.
\item All values in \eqref{eq:values} lie in $(1-2d,1+3d)$ and are distinct; the $\ell_j$ decrease and the $c_j$ increase with~$j$.
\item $d\le g/10^{19}$; hence $\ell-\beta c\ge1-\beta-5d\ge g/2$ for every low value $\ell\in\{\ell_0,\dots,\ell_J\}$, every high value $c\in\{c_0,\dots,c_J\}$ and every slope $\beta\le\kappa-g$.
\end{enumerate}
\end{lemma}

\begin{proof}
(N1) $d\le1/(10^{18}\cdot64)<10^{-19}$ and $d\le\underline\beta/10^{18}$, since $g,\underline\beta\le1$; $h=d/(10^6(J+1)^2)\le d/(4\cdot10^6)$ and $(J+1)h=d/(10^6(J+1))\le d/(2\cdot10^6)$; $\theta=d^2\underline\beta^2/(10^{18}(J+1)^6)\le d^2\cdot10^{-19}$; $hd\underline\beta/10^{12}=d^2\underline\beta/(10^{18}(J+1)^2)\ge d^2\underline\beta^2/(10^{18}(J+1)^6)=\theta$; and $\theta/(\underline\beta h)=d\underline\beta/(10^{12}(J+1)^4)\le10^{-12}$. (N2) $p\le3d/2$ and $d\le10^{-18}/(J+1)^6\le10^{-8}q_0^2$; the remaining bounds are direct computations, with $d_\kappa=d/\kappa\le2d$ as $\kappa\ge1/2$. (N3) $\ell_J\ge1-d-Jh>1-2d$ and $c_J\le1+2d+Jh<1+3d$; monotonicity and distinctness follow from $h>0$, $\beta_jh>0$ and $d_\kappa>0$. (N4) $d\le g/(10^{18}\cdot64)$, and then $\ell-\beta c\ge(1-2d)-\beta(1+3d)\ge1-\beta-5d\ge(1-\kappa)+g-5d\ge g-5d\ge g/2$.
\end{proof}

\subsubsection{The canonical menu and the partition}\label{app:realgeom}

With the convention $c_{-1}=\ell_{-1}=1$, \eqref{eq:canon} gives $A_0=1+\kappa$. In the plane of types consider the lines $t_2=A_j-\beta_jt_1$ ($0\le j\le J$) and their upper envelope $\cB(\xi)=\max_{0\le j\le J}(A_j-\beta_j\xi)$, $\xi\in\R$.

\begin{lemma}[Envelope]\label{lem:envelope}
For $j\ge1$, line~$j$ passes through $(c_{j-1},\ell_{j-1})$ and $(c_j,\ell_j)$; line~$0$ passes through $(1,1)$ and $(c_0,\ell_0)$. Consequently, $\cB(\xi)=A_j-\beta_j\xi$ on $[c_{j-1},c_j]$ for $1\le j\le J$, $\cB(\xi)=1+\kappa-\kappa\xi$ for $\xi\le c_0$, and $\cB(\xi)=A_J-\beta_J\xi$ for $\xi\ge c_J$; moreover, $\cB$ is strictly decreasing and $\cB(1)=1$. Further, $A_j-(1+\beta_j)<0$ for $j\ge1$, and for every $j\ge0$ and $(y_1,y_2)\in\cC_\kappa$,
\begin{equation}\label{eq:vertexIR}
(c_j-1)y_1+(\ell_j-1)y_2=d_\kappa(y_1-\kappa y_2)+h\sum_{i=1}^j(y_1-\beta_iy_2)\ \ge\ 0 .
\end{equation}
\end{lemma}

\begin{proof}
By definition $A_j-\beta_jc_{j-1}=\ell_{j-1}$, and $A_j-\beta_jc_j=\ell_{j-1}-\beta_jh=\ell_j$. For $j=0$: $1+\kappa-\kappa=1$ and $1+\kappa-\kappa c_0=1-d=\ell_0$. Lines $j-1$ and $j$ meet at $(c_{j-1},\ell_{j-1})$, and
\[
(A_j-\beta_j\xi)-(A_{j-1}-\beta_{j-1}\xi)=(\beta_{j-1}-\beta_j)(\xi-c_{j-1})
\]
is positive exactly for $\xi>c_{j-1}$, since $\beta_{j-1}>\beta_j$. Hence, as $\xi$ increases, the maximum passes from line $j-1$ to line $j$ at $c_{j-1}$; since $c_0<c_1<\dots<c_J$, the envelope is as stated, and it is strictly decreasing because all slopes $-\beta_j$ are negative. Next,
\begin{align*}
A_j-1-\beta_j&=(\ell_{j-1}-1)+\beta_j(c_{j-1}-1)=-d-h\sum_{1\le i<j}\beta_i+\beta_j\big(d_\kappa+(j-1)h\big)\\
&=-d\Big(1-\frac{\beta_j}\kappa\Big)+h\sum_{1\le i<j}(\beta_j-\beta_i)<0,
\end{align*}
because $\beta_j<\kappa$ and $\beta_i>\beta_j$ for $1\le i<j$. Finally, $c_j-1=d_\kappa+jh$ and $\ell_j-1=-d-h\sum_{1\le i\le j}\beta_i$ give \eqref{eq:vertexIR}, whose terms are nonnegative on $\cC_\kappa$ since $\beta_i\le\kappa$ and $y_1-\kappa y_2\ge0$.
\end{proof}

Fix $(y_1,y_2)\in\cC_\kappa$. We view \eqref{eq:canon} as the menu consisting of the following options (allocation, payment), for $0\le j\le J$:
\[
\begin{gathered}
Z:\ \big((0,0),\,0\big),\qquad B:\ \big((y_1,y_2),\,y_1+y_2\big),\\
X_j:\ \big((y_1-\beta_jy_2,\,0),\,y_1+(1-A_j)y_2\big),\qquad Y_j:\ \big((0,\,y_2-\beta_jy_1),\,y_2+(1-A_j)y_1\big);
\end{gathered}
\]
each type chooses an option of maximal utility. A type $\bm t=(t_1,t_2)$ obtains from the option $X_j$ the utility $(t_1-1)y_1+(A_j-\beta_jt_1-1)y_2$, so the best $X$-option gives
\begin{equation}\label{eq:bestX}
(t_1-1)y_1+\big(\cB(t_1)-1\big)y_2 ;
\end{equation}
symmetrically, the best $Y$-option gives $(t_2-1)y_2+(\cB(t_2)-1)y_1$, the option $B$ gives $(t_1-1)y_1+(t_2-1)y_2$, and the option $Z$ gives $0$. All allocations lie in $[0,1]$, since $\beta_j\le\kappa$ and $(y_1,y_2)\in\cC_\kappa$.

\begin{lemma}[Partition]\label{lem:partition}
Let $(y_1,y_2)\in\cC_\kappa$, and break ties in favor of the larger payment and, among options with identical outcomes, by the priority $B$, then $X_J,\dots,X_0$, then $Y_J,\dots,Y_0$, then $Z$. Then every type $\bm t\in\supp(F)^2$ receives the outcome of its class in the partition of \Cref{sec:two-real}: it chooses the option of its class or, when $(y_1,y_2)$ lies on the boundary of $\cC_\kappa$, an option with the same outcome. In particular, the canonical menu, which gives every type the outcome of its class, is an IC, IR menu on $\supp(F)^2$ with capacities exactly $(y_1,y_2)$ (attained on $B\ni(1,1)$), and its revenue is $\gamma_p(y_1+y_2)$, with $\gamma_p$ as in \Cref{lem:real}(a).
\end{lemma}

\begin{proof}
Let $\bm t=(t_1,t_2)$. If $y_2=0$, then $y_1=0$ (as $y_1\le y_2/\kappa$), every option has allocation $(0,0)$ and payment $0$, and the claim is trivial; so let $y_1,y_2>0$ (a tie between options with identical outcomes is resolved by the priority and does not affect the outcome). If $t_1<1$, then $\cB(t_1)=1+\kappa(1-t_1)$ and \eqref{eq:bestX} equals $(t_1-1)(y_1-\kappa y_2)\le0$; symmetrically, the best $Y$-option gives $(t_2-1)(y_2-\kappa y_1)$ when $t_2<1$; and the option $B$ gives $(t_1-1)y_1+(t_2-1)y_2$. We distinguish the following cases.
\begin{itemize}[leftmargin=2em]
\item $t_1<1$ and $t_2<1$. Every option has utility at most $0$, the utility of $Z$; an option with utility exactly $0$ is $X_0$ or $Y_0$ on the boundary of $\cC_\kappa$ ($y_1=\kappa y_2$ or $y_2=\kappa y_1$), where its allocation and payment are both $0$, as for $Z$. Class $Z$.
\item $t_1=1$ and $t_2<1$. Since $\cB(1)=1$, every $X$-option gives at most $(\cB(1)-1)y_2=0$: $X_0$ gives exactly $0$, and $X_j$ ($j\ge1$) gives $(A_j-\beta_j-1)y_2\le0$. The $Y$-options give at most $(t_2-1)(y_2-\kappa y_1)\le0$, and $B$ gives $(t_2-1)y_2\le0$. The maximum $0$ is attained by $Z$ and $X_0$ (by $X_j$, $j\ge1$, only if $y_2=0$, and by $Y_0$ only if $y_2=\kappa y_1$, when $Y_0$ has the outcome of $Z$); $X_0$ pays $y_1-\kappa y_2\ge0$, at least the payment of $Z$, so $X_0$ is chosen (when $y_1=\kappa y_2$ the two outcomes coincide). Class $X_0$.
\item $t_1=c_{j-1}$ with $1\le j\le J$. By \Cref{lem:envelope}, the best $X$-utility is $(c_{j-1}-1)y_1+(\ell_{j-1}-1)y_2$, attained by $X_{j-1}$ and $X_j$ (both lines pass through this corner), and $X_j$ pays more because $A_j-A_{j-1}=(\beta_j-\beta_{j-1})c_{j-1}<0$. By \eqref{eq:vertexIR}, this utility is nonnegative, hence at least the utility of $Z$ and, when $t_2<1$, of every $Y$-option. The option $B$ gives $(c_{j-1}-1)y_1+(t_2-1)y_2$, which is at least the best $X$-utility if and only if $t_2\ge\ell_{j-1}$; at $t_2=\ell_{j-1}$ they tie, and $B$ pays $y_1+y_2>y_1+(1-A_j)y_2$ (as $A_j>0$). If $t_2\ge1$, the $Y$-options give $(t_2-1)y_2+(\cB(t_2)-1)y_1\le(t_2-1)y_2$, at most the utility of $B$, since $\cB(t_2)\le1$. Hence class $X_j$ if $t_2<\ell_{j-1}$, and class $B$ otherwise.
\item The remaining types with $t_1\ge1$ (namely $t_1=1$ with $t_2\ge1$, and $t_1=c_J$). The option $B$ gives $(t_1-1)y_1+(t_2-1)y_2$. The best $X$-option gives $(t_1-1)y_1+(\cB(t_1)-1)y_2$, which is at most this whenever $t_2\ge\cB(t_1)$; for $t_1=c_j$ this means $t_2\ge\ell_j$ (as $\cB(c_j)=\ell_j$), and for $t_1=1$ it means $t_2\ge1$. The best $Y$-option is not better than $B$ either: if $t_2<1$ (so that $t_1=c_j$ and $t_2\ge\ell_j$), it gives $(t_2-1)(y_2-\kappa y_1)\le0\le(c_j-1)y_1+(\ell_j-1)y_2\le(t_1-1)y_1+(t_2-1)y_2$ by \eqref{eq:vertexIR}; if $t_2\ge1$, it gives $(t_2-1)y_2+(\cB(t_2)-1)y_1\le(t_2-1)y_2+(t_1-1)y_1$, since $\cB(t_2)\le1\le t_1$. At ties $B$ pays most. Class $B$.
\item The types with $t_1<1\le t_2$ are the reflections of the cases above: the canonical menu and $\cC_\kappa$ are symmetric under $(t_1,t_2)\mapsto(t_2,t_1)$, $(y_1,y_2)\mapsto(y_2,y_1)$ and $X_j\leftrightarrow Y_j$.
\end{itemize}
Every type chooses an option of maximal utility, so the menu is IC; the option $Z$ is always available, so all utilities are nonnegative and the menu is IR. The capacities are attained at $(1,1)\in B$. With the class probabilities computed below, the revenue is
\[
\Pr(B)(y_1+y_2)+pq_0\big[(y_1-\kappa y_2)+(y_2-\kappa y_1)\big]+\sum_{j\ge1}\mu p\theta^j\big[y_1+(1-A_j)y_2+y_2+(1-A_j)y_1\big]=\gamma_p(y_1+y_2).
\]
\end{proof}

We now compute the class probabilities. The total mass below $1$ is $p$, so $\Pr(Z)=p^2$, and $\Pr(X_0)=f(1)\cdot p=q_0p$. For $j\ge1$, $\Pr(X_j)=f(c_{j-1})\cdot\Pr[t_2<\ell_{j-1}]$, and the mass strictly below $\ell_{j-1}$ is $p(1-\theta)\sum_{i=j}^{J-1}\theta^i+p\theta^J=p\theta^j$; so $\Pr(X_j)=\mu p\theta^j$. Conditionally on the type not being in $Z$, the single-item classes have total probability
\[
\alpha=\frac{2pq_0+2\sum_{j\ge1}\mu p\theta^j}{1-p^2}\le\frac{3dq_0+3d\theta}{1/2}\le d,
\]
using $p\le\frac32d$, $\mu\le\frac12$ and $\sum_{j\ge1}\theta^j\le2\theta$; in particular, $\alpha\le1/100$. This proves \Cref{lem:real}(b). Since $\Pr(B)=1-p^2-2pq_0-2\sum_{j\ge1}\mu p\theta^j$,
\[
\gamma_p=1-p^2-(1+\kappa)pq_0-\sum_{j\ge1}A_j\,\mu p\,\theta^j .
\]
The numbers $A_j=\ell_{j-1}+\beta_jc_{j-1}$ lie in $(0,3)$, so $\gamma_p<1$, and the subtracted terms add up to at most $\frac94d^2+3dq_0+9d\theta\le\frac d4$; hence $1-\frac d4\le\gamma_p<1$.

We also record the revenues of single prices, which \Cref{app:high,app:C} use. For a price $z$, let $\rev(z)=z\Pr_{a\sim F}[a\ge z]$, as in \Cref{sec:two-high}.

\begin{lemma}[Revenue curve]\label{lem:revcurve}
$\rev(1)=1-p$ and $\rev(\ell_0)=\ell_0(1-p\theta)$, so that $\rev(\ell_0)-\rev(1)=p\big(1-\theta(1-d)\big)-d$. Moreover, $\rev(\ell_j)<\rev(\ell_0)$ for $1\le j\le J$, $\rev(c_j)\le\rev(1)-q_0/2$ for $0\le j\le J$, and $\rev^*\le1-d/2$.
\end{lemma}

\begin{proof}
The mass below $1$ is $p$ and the mass below $\ell_0$ is $p\theta$, which gives $\rev(1)$ and $\rev(\ell_0)$, and $(1-d)(1-p\theta)-(1-p)=p(1-\theta+d\theta)-d$. For $j\ge1$, $\rev(\ell_j)\le\ell_j\le\ell_0-\underline\beta h$, whereas $\rev(\ell_0)\ge\ell_0-p\theta\ge\ell_0-2d\theta$, and $2d\theta<\underline\beta h$ by (N1). Next, $\rev(c_j)\le(1+3d)(1-p-q_0)\le1-p-q_0+3d\le\rev(1)-q_0/2$ by (N2). Finally, $\rev(\ell_j)\le\ell_0=1-d$, $\rev(1)\le1-d/2$ and $\rev(c_j)\le1-d/2$.
\end{proof}

\subsubsection{The flow}\label{app:realflow}

We write $\bm t'\to\bm t$ for the constraint that $\bm t'$ does not gain by reporting $\bm t$, and say that $\bm t'$ \emph{sends} its weight to $\bm t$. Let $\bm P_j=(c_j,\ell_j)$ for $-1\le j\le J$, so that $\bm P_{-1}=(1,1)$; these \emph{corners} belong to $B$, and by~\eqref{eq:values},
\begin{equation}\label{eq:corners}
\bm P_j-\bm P_{j-1}=(c_j-c_{j-1})(1,-\beta_j)\qquad(0\le j\le J).
\end{equation}
For $0\le j\le l\le J$ let $\bm t_{j,l}=(c_{j-1},\ell_l)$, so that $X_j=\{\bm t_{j,l}:j\le l\le J\}$. Write $p_j=\sum_{l\ge j}f(\ell_l)=p\theta^j$, the mass of the values below $\ell_{j-1}$, so that $f(\ell_l)=(1-\theta)p_l$ for $l<J$, and put $p_{J+1}=0$. Let $\mu_0=q_0$ and $\mu_j=\mu$ for $j\ge1$, so that $\bm t_{j,l}$ has mass $\mu_jf(\ell_l)$. Finally, let $R=1-p^2-2pq_0(1+d_\kappa^{-1})$,
\[
v_{0,l}=\frac{R\,f(\ell_l)}{2p},\qquad v_{j,l}=\frac{2\mu f(\ell_l)\,\ell_{j-1}}{\ell_{j-1}-\ell_l}\quad(j\ge1),\qquad T_j=\frac{\mu p_jc_j}h+\sum_{l=j}^Jv_{j,l}\quad(1\le j\le J),
\]
and $T_{J+1}=0$. The flow consists of the following weights and of the reflections of the weights in (F1)--(F3) between two types, obtained by exchanging the two coordinates of all types involved; the IR weights in (F2) and the weights (F4) are symmetric as stated.
\begin{enumerate}[label=(F\arabic*),leftmargin=3em]
\item For $0\le j\le l\le J$, the corner $\bm P_{j-1}$ sends $v_{j,l}$ to $\bm t_{j,l}$, and the corner $\bm P_j$ sends $\mu_jf(\ell_l)c_{j-1}/(c_j-c_{j-1})-\ind{l=j}\,T_{j+1}$ to $\bm t_{j,l}$. For $1\le j\le l\le J$, the type $\bm t_{j,l}$ sends $\mu f(\ell_l)c_j/h+v_{j,l}$ to $\bm t_{j-1,j-1}$.
\item For $0\le s,l\le J$, the type $\bm t_{0,l}$ sends $f(\ell_s)f(\ell_l)(1-p^2)/(2p^2)$ to $(\ell_s,\ell_l)\in Z$, and the IR constraint of $(\ell_s,\ell_l)$ has weight $f(\ell_s)f(\ell_l)/p^2$.
\item For $0\le j\le J$, let $I_j$ be the total weight that $\bm P_j$ sends in (F1), minus its mass $\mu f(\ell_j)$. For $1\le j\le J$, the type $(c_j,\ell_{j-1})$ sends $I_j$ to $\bm P_j$; for $0\le i\le J$, the type $(c_0,c_i)$ sends $I_0/(J+1)$ to $\bm P_0$.
\item Every type of $B$ other than the corners and their reflections sends to $\bm P_{-1}$ its mass minus the weights that it sends in (F3) and in the reflection of (F3).
\end{enumerate}
Thus every type of $X_j$ receives flow from the two corners next to it: from $\bm P_{j-1}$, which has the same first coordinate, and from $\bm P_j$. For $j\ge1$ it passes the flow on to $\bm t_{j-1,j-1}$, and for $j=0$ to $Z$, whose IR constraints absorb it. The corners obtain what they pass on from other types of $B$, and the rest of $B$ flows through $\bm P_{-1}$ into $X_0$ and $Y_0$.

\begin{lemma}[The flow]\label{lem:flow}
The weights (F1)--(F4) are positive and satisfy~\eqref{eq:conserv} for the probabilities $f(\bm t)=f(t_1)f(t_2)$ of the types. For every $\bm y\in\cC_\kappa$, every constraint of positive weight is tight at the canonical menu with capacities $\bm y$. The virtual values satisfy $\Phi_1,\Phi_2\ge g/2$ on $B$; $\Phi_1=0$ and $\Phi_2\le-\frac12$ on every $X_j$, and symmetrically on every $Y_j$; and $\Phi_1,\Phi_2\le-\frac12$ on $Z$. Finally, with $f_*=(d\theta^J/2)^2$, every type has mass at least $f_*$, and so have the weights of the constraints $\bm P_{j-1}\to\bm t$ and $\bm P_j\to\bm t$ for $\bm t\in X_j$, of their reflections, of $\bm t_{0,0}\to(\ell_0,\ell_0)$ and of the IR constraint of $(\ell_0,\ell_0)$.
\end{lemma}

\begin{proof}
We use (N1)--(N4) without further mention; in particular, $\theta\le10^{-12}\underline\beta h$, $q_0=1/(10^5(J+1))$, $\mu\ge0.99/(J+1)$, $p\le\frac32d$ and $d\le10^{-8}q_0^2$.

\emph{Positivity.} Since $\ell_{j-1}-\ell_l\ge\beta_jh\ge\underline\beta h$ for $l\ge j\ge1$, $\ell_{j-1}<1$ and $c_j<2$,
\begin{equation}\label{eq:VTbound}
T_j\le\frac{4\mu p_j}{\underline\beta h}\qquad(1\le j\le J).
\end{equation}
Hence the weight that $\bm P_j$ sends to $\bm t_{j,j}$ is at least half of $\mu_jf(\ell_j)c_{j-1}/(c_j-c_{j-1})$: for $j\ge1$ because $T_{j+1}\le4\mu\theta p_j/(\underline\beta h)\le\mu(1-\theta)p_j/(4h)$, and for $j=0$ because $T_1\le4\mu\theta p/(\underline\beta h)\le p\le q_0f(\ell_0)/(2d_\kappa)$. So all weights in (F1) are positive, and those sent by $\bm P_j$ to $X_j$ are at least $\mu_jf(\ell_l)c_{j-1}/(2(c_j-c_{j-1}))$. Also $R\ge1-3d^2-4q_0\ge\frac9{10}$. The corner $\bm P_j$ sends $\mu_jp_jc_{j-1}/(c_j-c_{j-1})-T_{j+1}$ to $X_j$ and $T_{j+1}-\mu p_{j+1}c_{j+1}/h$ to $X_{j+1}$, so
\begin{equation}\label{eq:Psends}
I_j+\mu f(\ell_j)=\frac{\mu_jp_jc_{j-1}}{c_j-c_{j-1}}-\frac{\mu p_{j+1}c_{j+1}}h\qquad(0\le j\le J),
\end{equation}
where the last term is absent for $j=J$. As $p_{j+1}=\theta p_j$ for $j<J$, this gives $I_j=\mu p_j\big[(c_{j-1}-\theta c_{j+1})/h-(1-\theta)\big]$ for $1\le j<J$, $I_J=\mu p_J(c_{J-1}/h-1)$ and $I_0=q_0p\kappa/d-\mu\theta pc_1/h-\mu(1-\theta)p$. Hence $0<I_j\le2\mu p_j/h$ for $j\ge1$, and $0<I_0\le\frac32q_0$. For $1\le j\le J$, the type $(c_j,\ell_{j-1})\in B$ has mass $\mu f(\ell_{j-1})=\mu(1-\theta)p_{j-1}$ and sends $I_j\le2\mu\theta p_{j-1}/h\le\mu f(\ell_{j-1})/2$. The types $(c_0,c_i)\in B$ have mass $\mu^2$ and send $I_0/(J+1)\le3q_0/(2(J+1))$ each, except that $(c_0,c_0)$ sends this amount also to the reflection of $\bm P_0$; and $3q_0/(J+1)\le\mu^2/2$. These types are distinct from each other and from the reflections of $(c_j,\ell_{j-1})$, so every type of $B$ other than the corners and their reflections sends at least half its mass to $\bm P_{-1}$ in (F4). Hence all weights are positive.

\emph{Conservation.} For $0\le j<J$, the types $\bm t_{j+1,l}$ ($l>j$) send to $\bm t_{j,j}$ in total $\sum_{l>j}(\mu f(\ell_l)c_{j+1}/h+v_{j+1,l})=T_{j+1}$, which offsets the term $-T_{j+1}$ in (F1). So every $\bm t_{j,l}$ receives $\mu_jf(\ell_l)c_{j-1}/(c_j-c_{j-1})$ from types with first coordinate $c_j$, and $v_{j,l}$ from $\bm P_{j-1}$. For $j\ge1$, $\bm t_{j,l}$ sends $\mu f(\ell_l)c_j/h+v_{j,l}$, which is its inflow plus its mass $\mu f(\ell_l)$, as $c_j=c_{j-1}+h$. The type $\bm t_{0,l}$ receives $q_0f(\ell_l)/d_\kappa+v_{0,l}$ and sends $f(\ell_l)(1-p^2)/(2p)$, which by the choice of $R$ is its inflow plus its mass $q_0f(\ell_l)$. The type $(\ell_s,\ell_l)\in Z$ receives $f(\ell_s)f(\ell_l)(1-p^2)/(2p^2)$ from $\bm t_{0,l}$ and as much from the reflection of $\bm t_{0,s}$; together with its mass $f(\ell_s)f(\ell_l)$, this is the weight of its IR constraint. The corner $\bm P_j$, $0\le j\le J$, receives $I_j$ and sends $I_j+\mu f(\ell_j)$, and the types of $B$ other than the corners and their reflections receive nothing and send their mass. It remains to check $\bm P_{-1}$, which sends $2\sum_lv_{0,l}=R$. No flow enters $B$ from outside, and $B$ sends twice $\sum_{j=0}^J(I_j+\mu f(\ell_j))$ through the corners $\bm P_0,\dots,\bm P_J$ and their reflections; so $\bm P_{-1}$ is conserved if $R=\Pr(B)-2\sum_{j=0}^J(I_j+\mu f(\ell_j))$. By~\eqref{eq:Psends}, the sum telescopes,
\[
\sum_{j=0}^J(I_j+\mu f(\ell_j))=\frac{q_0p}{d_\kappa}+\sum_{j=1}^J\frac{\mu p_j(c_{j-1}-c_j)}h=\frac{q_0p}{d_\kappa}-\mu\sum_{j=1}^Jp_j ,
\]
and $\Pr(B)=1-p^2-2pq_0-2\mu\sum_{j\ge1}p_j$ (\Cref{app:realgeom}); this gives the claim.

\emph{Tightness.} Let $\bm y\in\cC_\kappa$. By \Cref{lem:partition}, the canonical menu gives every type $\bm t\in B$ the utility $(\bm t-\one)\cdot\bm y$, every type of $Z$ the utility $0$, and $\bm t_{j,l}\in X_j$ the utility $c_{j-1}(y_1-\beta_jy_2)-y_1-(1-A_j)y_2=(c_{j-1}-1)y_1+(\ell_{j-1}-1)y_2$, which is the utility of $\bm P_{j-1}$, as $A_j-\beta_jc_{j-1}=\ell_{j-1}$. Constraints between types of $B$ are tight, since all of them receive $\bm y$. The constraint $\bm P_{j-1}\to \bm t_{j,l}$ is tight, because the two types have the same utility, $\bm P_{j-1}-\bm t_{j,l}=(0,\ell_{j-1}-\ell_l)$ and $x_2(\bm t_{j,l})=0$. The constraint $\bm P_j\to \bm t_{j,l}$ is tight, because by~\eqref{eq:corners} the utility of $\bm P_j$ exceeds that of $\bm P_{j-1}$, and so that of $\bm t_{j,l}$, by $(c_j-c_{j-1})(y_1-\beta_jy_2)=(c_j-c_{j-1})\,x_1(\bm t_{j,l})$, while $x_2(\bm t_{j,l})=0$; the same computation with $j-1$ in place of $j$ shows that $\bm t_{j,l}\to \bm t_{j-1,j-1}$ is tight. The constraints in (F2) are tight, because the types involved have utility $0$ and $Z$ receives nothing. The reflected constraints are tight by symmetry.

\emph{Virtual values.} A type of $B$ other than the corners and their reflections receives nothing, so its virtual values are its coordinates, which exceed $1-2d$. The corner $\bm P_{-1}$ receives in total at most $\Pr(B)\le1$, from types $\bm t'$ with $t'_k-1<3d$; so $q_0^2\Phi_k(\bm P_{-1})\ge q_0^2-3d\ge q_0^2/2$. For $1\le j\le J$, the corner $\bm P_j$ receives only from $(c_j,\ell_{j-1})=\bm P_j+(0,\beta_jh)$; so $\Phi_1(\bm P_j)=c_j$ and $\mu f(\ell_j)\Phi_2(\bm P_j)=\mu f(\ell_j)\ell_j-\beta_jhI_j=\mu f(\ell_j)\ell_{j-1}-\beta_jh(I_j+\mu f(\ell_j))$. By~\eqref{eq:Psends}, this is $\mu p_j\big[(1-\theta)(\ell_{j-1}-\beta_jc_{j-1})+\theta\beta_j(c_{j+1}-c_{j-1})\big]$ for $j<J$ and $\mu p_J(\ell_{J-1}-\beta_Jc_{J-1})$ for $j=J$; so $\Phi_2(\bm P_j)\ge\ell_{j-1}-\beta_jc_{j-1}\ge g/2$ by (N4). The corner $\bm P_0$ receives $I_0\le\frac32q_0$ from types $(c_0,c_i)$ with $0<c_i-\ell_0<4d$; so $\Phi_1(\bm P_0)=c_0$ and $\Phi_2(\bm P_0)\ge\ell_0-6dq_0/(\mu f(\ell_0))\ge\frac12$, as $f(\ell_0)\ge d/4$. The reflections of the corners are symmetric. Since $g\le1$, the virtual values on $B$ are at least $g/2$.

Let $\bm t=\bm t_{j,l}\in X_j$. Its inflow from $\bm P_j$ and from the types $\bm t_{j+1,l'}$ comes from the first coordinate $c_j$, and its inflow from $\bm P_{j-1}$ from the first coordinate $c_{j-1}$; so $\mu_jf(\ell_l)\Phi_1(\bm t)=\mu_jf(\ell_l)c_{j-1}-(c_j-c_{j-1})\,\mu_jf(\ell_l)c_{j-1}/(c_j-c_{j-1})=0$. In the second coordinate, the inflow from $\bm P_{j-1}$ and $\bm P_j$ comes from values at least $\ell_l$, and for $l=j$ the inflow from the types $\bm t_{j+1,l'}$ raises $\mu_jf(\ell_j)\Phi_2(\bm t)$ by at most $(\ell_j-\ell_J)T_{j+1}\le(J+1)hT_{j+1}\le4(J+1)\mu\theta p_j/\underline\beta$, by~\eqref{eq:VTbound}. For $j\ge1$, the inflow from $\bm P_{j-1}$ lowers $\mu f(\ell_l)\Phi_2(\bm t)$ by $v_{j,l}(\ell_{j-1}-\ell_l)=2\mu f(\ell_l)\ell_{j-1}$, and $4(J+1)\mu\theta p_j/\underline\beta\le\mu f(\ell_j)/8$; hence $\Phi_2(\bm t)\le\ell_l-2\ell_{j-1}+\frac18<-\frac12$. For $j=0$, the inflow from $\bm P_{-1}$ lowers $q_0f(\ell_l)\Phi_2(\bm t)$ by at least $v_{0,l}d=dRf(\ell_l)/(2p)$, and $4(J+1)\mu\theta p/\underline\beta\le q_0f(\ell_0)$; hence $\Phi_2(\bm t)\le2-dR/(2pq_0)<-\frac12$. Finally, $(\ell_s,\ell_l)\in Z$ receives $f(\ell_s)f(\ell_l)(1-p^2)/(2p^2)$ from $(1,\ell_l)$ and from $(\ell_s,1)$, so that $\Phi_1(\ell_s,\ell_l)=\ell_s-(1-\ell_s)(1-p^2)/(2p^2)\le1-d/(4p^2)<-\frac12$, and symmetrically $\Phi_2(\ell_s,\ell_l)<-\frac12$.

\emph{Masses and weights.} The smallest mass of a value is $f(\ell_J)=p\theta^J\ge d\theta^J/2$, so every type has mass at least $f_*$. For $j\ge1$, $v_{j,l}\ge2\mu f(\ell_l)\ge f_*$; $v_{0,l}\ge f(\ell_l)/(4d)\ge f_*$; the weights sent by $\bm P_j$ to $X_j$ are at least $\mu_jf(\ell_l)c_{j-1}/(2(c_j-c_{j-1}))\ge\mu_jf(\ell_l)\ge f_*$; and the weights of $\bm t_{0,0}\to(\ell_0,\ell_0)$ and of the IR constraint of $(\ell_0,\ell_0)$ are at least its mass $f(\ell_0)^2\ge f_*$.
\end{proof}

\subsubsection{The bound and stability}\label{app:realeq}

By \Cref{lem:flow}, the flow satisfies~\eqref{eq:conserv} for the probabilities $f(\bm t)=f(t_1)f(t_2)$ of the types. We first show that every IC, IR menu earns at most $\gamma_p$ times the sum of its capacities. Let $(\bm x,s)$ be an IC, IR menu with capacities $(y_1,y_2)$. By \eqref{eq:certcap}, its revenue is at most $y_1\sum_{\bm t}f(\bm t)\pos{\Phi_1(\bm t)}+y_2\sum_{\bm t}f(\bm t)\pos{\Phi_2(\bm t)}$; by symmetry, both sums equal a common value $\gamma'$; and the positive virtual values are exactly those at $B$ types. The canonical menu with capacities $(1,1)$ attains equality in~\eqref{eq:certcap}: it is IC and IR (\Cref{lem:partition}); every constraint of positive weight is tight (\Cref{lem:flow}); it allocates $(1,1)$ on $B$, where $\Phi_k>0$, nothing of item~$2$ on $X_j$, nothing of item~$1$ on $Y_j$ and nothing on $Z$; and the remaining terms vanish because $\Phi_1=0$ on $X_j$ and $\Phi_2=0$ on $Y_j$. Its revenue is $2\gamma_p$ (\Cref{lem:partition}), so $\gamma'=\gamma_p$. Every IC, IR menu therefore earns at most $\gamma_p(y_1+y_2)\le2\gamma_p$, and the canonical menu with capacities $(1,1)$ attains $2\gamma_p$; so $\OPT(F;1,2)=2\gamma_p$. We write $\gamma=\gamma_p$.

The identity~\eqref{eq:slackid} turns this bound into an identity. For an IC, IR menu $Q=(\bm x,s)$ with capacities $\bm y$, let $\xi(\bm t',\bm t)=u(\bm t')-u(\bm t)-(\bm t'-\bm t)\cdot\bm x(\bm t)\ge0$ be the slack of the constraint $\bm t'\to\bm t$. Since the revenue of $Q$ is its expected virtual surplus minus the weighted slacks, by~\eqref{eq:slackid}, and $\gamma\|\bm y\|_1=\sum_ky_k\sum_{\bm t\in B}f(\bm t)\Phi_k(\bm t)$,
\begin{equation}\label{eq:lossC}
\begin{aligned}
\loss(Q)={}&\sum_{\bm t\in B}\sum_kf(\bm t)\Phi_k(\bm t)\big(y_k-x_k(\bm t)\big)+\sum_{\bm t,k:\,\Phi_k(\bm t)<0}f(\bm t)\,|\Phi_k(\bm t)|\,x_k(\bm t)\\
&+\sum_{\bm t',\bm t}\lambda(\bm t',\bm t)\,\xi(\bm t',\bm t)+\sum_{\bm t}\lambda(\bm t,\bot)\,u(\bm t),
\end{aligned}
\end{equation}
and every term is nonnegative.

\begin{proof}[Proof of \Cref{lem:stab}]
(a) Let $e=2\loss(Q)/(gf_*)$, and write $\bar{\bm x}(\bm t)$, $\bar s(\bm t)$ and $\bar u(\bm t)$ for the allocation, payment and utility of $\bm t$ in the canonical outcome~\eqref{eq:canon} with capacities $\bm y$; these are linear in $\bm y$, and defined for every $\bm y\in\R^2$. By \Cref{lem:flow}, every term of~\eqref{eq:lossC} carries a factor $f(\bm t)\Phi_k(\bm t)\ge gf_*/2$, a factor $f(\bm t)|\Phi_k(\bm t)|\ge f_*/2$ or a weight $\lambda$, which is at least $f_*$ for the constraints listed in \Cref{lem:flow}. Hence the following quantities are at most $e$: $y_k-x_k(\bm t)$ for $\bm t\in B$; $x_k(\bm t)$ whenever $\Phi_k(\bm t)<0$, that is, both coordinates on $Z$, the second on $X_j$ and the first on $Y_j$; the slacks of the constraints $\bm P_{j-1}\to\bm t$ and $\bm P_j\to\bm t$ with $\bm t\in X_j$, of their reflections and of $\bm t_{0,0}\to(\ell_0,\ell_0)$; and $u(\ell_0,\ell_0)$. Recall that all values lie in $(1-2d,1+3d)$.

\emph{Types of $B$.} The constraints $\bm t_{0,0}\to(\ell_0,\ell_0)$ and $\bm P_{-1}\to \bm t_{0,0}$ give $u(\bm t_{0,0})\le u(\ell_0,\ell_0)+d\,x_1(\ell_0,\ell_0)+e\le(2+d)e$ and $u(\one)\le u(\bm t_{0,0})+d\,x_2(\bm t_{0,0})+e\le(3+2d)e$. For $\bm t\in B$, the IC constraints between $\bm t$ and $\one$ give $(\bm t-\one)\cdot\bm x(\one)\le u(\bm t)-u(\one)\le(\bm t-\one)\cdot\bm x(\bm t)$, where $\bm x(\one)$ and $\bm x(\bm t)$ differ from $\bm y$ by at most $e$ in each coordinate and $\|\bm t-\one\|_1<6d$. As $u(\one)\ge0$ and $\bar u(\bm t)=(\bm t-\one)\cdot\bm y$,
\begin{equation}\label{eq:Bclose}
|u(\bm t)-\bar u(\bm t)|\le(3+8d)\,e\le4e\qquad(\bm t\in B).
\end{equation}

\emph{Types of $X_j$ and $Y_j$.} Let $\bm t=\bm t_{j,l}$, and let $\xi$ and $\xi'$ be the slacks of $\bm P_{j-1}\to\bm t$ and $\bm P_j\to\bm t$. Since $\bm P_{j-1}-\bm t=(0,\ell_{j-1}-\ell_l)$, and by~\eqref{eq:corners},
\[
u(\bm P_{j-1})=u(\bm t)+(\ell_{j-1}-\ell_l)\,x_2(\bm t)+\xi,\qquad u(\bm P_j)-u(\bm P_{j-1})=(c_j-c_{j-1})\big(x_1(\bm t)-\beta_jx_2(\bm t)\big)+\xi'-\xi .
\]
The canonical outcome satisfies the same equations with zero slacks, $\bar x_2(\bm t)=0$ and $\bar x_1(\bm t)=y_1-\beta_jy_2$. By~\eqref{eq:Bclose}, $x_2(\bm t)\le e$ and $c_j-c_{j-1}\ge h$, we get $|x_1(\bm t)-\bar x_1(\bm t)|\le9e/h+e\le10e/h$ and $|u(\bm t)-\bar u(\bm t)|\le4e+2de+e\le6e$; hence $|s(\bm t)-\bar s(\bm t)|=|\bm t\cdot(\bm x(\bm t)-\bar{\bm x}(\bm t))-(u(\bm t)-\bar u(\bm t))|\le28e/h$. The types of $Y_j$ are symmetric.

\emph{Types of $Z$.} Here $x_k(\bm t)\le e$, and IR and the constraint $(\ell_0,\ell_0)\to\bm t$ give $0\le u(\bm t)\le u(\ell_0,\ell_0)+4de\le2e$; so $|s(\bm t)|\le4e$, while $\bar{\bm x}(\bm t)=\bm 0$ and $\bar s(\bm t)=0$. Together with $|x_k(\bm t)-y_k|\le e$ and $|s(\bm t)-\bar s(\bm t)|\le8e$ for $\bm t\in B$, by~\eqref{eq:Bclose}, every allocation probability and payment of $Q$ differs from its canonical value by at most $30e/h$.

\emph{The cone.} As $\bm P_0\in B$, IR and~\eqref{eq:Bclose} give $d_\kappa(y_1-\kappa y_2)=\bar u(\bm P_0)\ge-4e$, and symmetrically $d_\kappa(y_2-\kappa y_1)\ge-4e$. Hence $\bm 0\le\bm y-\hat{\bm y}\le(4e/d)\one$; for instance, $y_1-\hat y_1=\pos{\kappa y_1-y_2}/\kappa\le4e/(d_\kappa\kappa)=4e/d$. The canonical outcome is linear in the capacities, and the absolute values of its coefficients sum to at most $2$, as $0<\beta_j<1$ and $0<A_j\le1+\beta_j<2$ (\Cref{lem:envelope}); so the canonical outcomes with capacities $\bm y$ and $\hat{\bm y}$ differ by at most $8e/d\le8e/h$. Altogether, every allocation probability, payment and capacity of $Q$ differs from that of the canonical menu with capacities $\hat{\bm y}$ by at most $38e/h$. This proves~(a) with $K_0=80/(ghf_*)$.

(b) Suppose that all payments of $Q$ are nonnegative. We regard a menu together with its capacities as the vector of its allocation probabilities, payments and capacities. By~\eqref{eq:canon}, the vector $\bm q(\bm y)$ of the canonical menu with capacities $\bm y$ is linear in $\bm y$, and the vector $\bm q$ of $Q$ is nonnegative. The vectors $\bm q^{(1)}=\bm q(1,\kappa)$ and $\bm q^{(2)}=\bm q(\kappa,1)$ are nonnegative, their positive coordinates are at least $m_*=\min\{1-\kappa^2,g\}$, and the coordinates of $\bm q^{(1)}+\bm q^{(2)}=\bm q(1+\kappa,1+\kappa)$ are at most $4$. Indeed, the coordinates of $\bm q^{(1)}$ are $1$, $\kappa$ and $1+\kappa$ on $B$ and for the capacities; $1-\kappa\beta_j\ge1-\kappa^2$ and $1+(1-A_j)\kappa\ge1-\kappa\beta_j$ on $X_j$; $\kappa-\beta_j$ and $\kappa+1-A_j\ge\kappa-\beta_j$ on $Y_j$, which vanish for $j=0$ and are at least $g$ for $j\ge1$; and $0$ on $Z$; here we used $A_0=1+\kappa$ and $A_j\le1+\beta_j$. The vector $\bm q^{(2)}$ is symmetric.

Since $\hat{\bm y}\in\cC_\kappa$, we have $\hat{\bm y}=\alpha_1(1,\kappa)+\alpha_2(\kappa,1)$ with $\alpha_1=(\hat y_1-\kappa\hat y_2)/(1-\kappa^2)\ge0$ and $\alpha_2=(\hat y_2-\kappa\hat y_1)/(1-\kappa^2)\ge0$, so that $\bm q(\hat{\bm y})=\alpha_1\bm q^{(1)}+\alpha_2\bm q^{(2)}$. Let $E=K_0\loss(Q)$, lower both coefficients to $\alpha'_i=\pos{\alpha_i-E/m_*}$, and let $\bm y'=\alpha'_1(1,\kappa)+\alpha'_2(\kappa,1)\in\cC_\kappa$. Consider a coordinate of $\bm q(\bm y')=\alpha'_1\bm q^{(1)}+\alpha'_2\bm q^{(2)}$. If it is positive, then for some $i$ both $\alpha'_i$ and the coordinate of $\bm q^{(i)}$ are positive; so lowering the coefficients lowers this coordinate of $\bm q(\hat{\bm y})$ by at least $m_*\cdot E/m_*=E$, and by~(a) the result is at most the coordinate of $\bm q$. If it is $0$, it is at most the coordinate of $\bm q$, as $\bm q\ge\bm 0$. Hence $\bm q(\bm y')\le \bm q$, and every coordinate of $\bm q-\bm q(\bm y')$ is at most $E+4E/m_*\le5E/m_*=K_1\loss(Q)$.
\end{proof}

In particular, if an IC, IR menu earns $\gamma\|\bm y\|_1$, then (a) with $\loss(Q)=0$ shows that it is the canonical menu with capacities $\hat{\bm y}=\bm y\in\cC_\kappa$; conversely, the canonical menu with capacities $\bm y\in\cC_\kappa$ earns $\gamma\|\bm y\|_1$ by \Cref{lem:partition}. This proves \Cref{lem:real}(a).

Finally, we prove (c). The values \eqref{eq:values}, the masses \eqref{eq:masses}, the numbers $A_j$ and $\gamma_p$ are rationals obtained from the inputs and $p$ by polynomially many additions, multiplications and divisions, in which the only powers are $\theta^j$ with $j\le J$ (and $J$ is at most the input length). Hence they have polynomial bit length and can be computed in polynomial time, and so can the class partition, which is determined by the values, and the constant $K_0=80/(ghf_*)$ of \Cref{lem:stab}, which depends only on $(J,\underline\beta,g)$. The flow is used only in the proofs and need not be computed. This completes the proofs of \Cref{lem:real,lem:stab}.

\subsection{Capacity tables}\label{app:capacity}

This subsection proves the statements of \Cref{sec:two-tables} on averaging and subproblems and \Cref{lem:narrow}, and states and proves the quantitative form of the narrow cone (\Cref{lem:narrowquant}). Two facts are used repeatedly. For $\bm y\in\cC_\kappa$, the coordinates of $C_\tau(\bm y)$ are nonnegative, because $\beta_jy_2\le\kappa y_2\le y_1$ and $\beta_jy_1\le\kappa y_1\le y_2$. And for every class $\tau$ and every $\bm y\in\R^2$,
\begin{equation}\label{eq:Cnorm}
\|C_\tau(\bm y)\|_1\le\|\bm y\|_1 ,
\end{equation}
because $|y_1-\beta_jy_2|\le|y_1|+|y_2|$.

We also use two properties of the map $G_\kappa$ of~\eqref{eq:G}. First, for $0\le U_1,U_2\le1$, the point $\bm y^*=G_\kappa(U_1,U_2)$ is the coordinatewise largest point of $\cC_\kappa$ below $(U_1,U_2)$. Indeed, if $U_1\le\kappa U_2$, then $\bm y^*=(U_1,U_1/\kappa)$, so $y^*_1=\kappa y^*_2$ and $y^*_2\ge\kappa y^*_1$; if $\kappa U_2<U_1\le U_2/\kappa$, then $\bm y^*=(U_1,U_2)$, and both cone inequalities hold; if $U_1>U_2/\kappa$, then $\bm y^*=(U_2/\kappa,U_2)$, symmetrically to the first case. In all three cases $\bm 0\le\bm y^*\le(U_1,U_2)\le(1,1)$. Conversely, if $\bm y\in\cC_\kappa$ satisfies $y_1\le U_1$ and $y_2\le U_2$, then $y_1\le y_2/\kappa\le U_2/\kappa$, so $y_1\le\min\{U_1,U_2/\kappa\}=y^*_1$; symmetrically, $y_2\le y^*_2$. Second, $G_\kappa$ is $(1+\kappa^{-1})$-Lipschitz for $\|\cdot\|_1$: $|\min\{\xi_1,\xi_2\}-\min\{\xi'_1,\xi'_2\}|\le\max\{|\xi_1-\xi'_1|,|\xi_2-\xi'_2|\}$ gives $|\min\{U_1,U_2/\kappa\}-\min\{U'_1,U'_2/\kappa\}|\le|U_1-U'_1|+\kappa^{-1}|U_2-U'_2|$, the second coordinate satisfies the symmetric bound, and we add the two.

\subsubsection{Averaging and subproblems}\label{app:averaging}

For reports $\bm v_{-i}$ of the opponents of buyer~$i$, let $\mathsf m(\bm v_{-i})$ be the multiset of the classes of these $n-1$ reports. For a feasible table $\bm y$, define the anonymous table
\[
\bar{\bm y}(\mathsf m)=\frac1n\sum_{i=1}^n\E\big[\bm y_i(\bm v_{-i})\ \big|\ \mathsf m(\bm v_{-i})=\mathsf m\big]
\]
for every multiset $\mathsf m$ of $n-1$ classes; it results from the two averaging operations of \Cref{sec:two-tables}. Note that $\bar{\bm y}(\mathsf m)=\E_{l,S}\,\bm y_i(\bm v_{-i})$ in the notation of \Cref{sec:two-high}, where $n-l$ is the number of copies of $Z$ in $\mathsf m$ and $S$ is the multiset of single-item classes in $\mathsf m$.

\begin{lemma}[Averaging]\label{lem:average}
For every feasible table $\bm y$, the table $\bar{\bm y}$ is feasible and takes values in $\cC_\kappa$, $\cY(\bar{\bm y})=\cY(\bm y)$, and $\sum_i\E\,\psi\big(\bar{\bm y}(\mathsf m(\bm v_{-i}))\big)\ge\sum_i\E\,\psi\big(\bm y_i(\bm v_{-i})\big)$ for every concave function $\psi$ on $\cC_\kappa$.
\end{lemma}

\begin{proof}
The set $\cC_\kappa$ is convex, so $\bar{\bm y}$ takes values in it. The classes of the $n$ buyers are independent and identically distributed, and conditionally on the classes, the reports of different buyers are independent. Hence, conditionally on the event that the classes of all $n$ buyers form a multiset $\mathsf M$, the vector of classes is uniformly distributed over the arrangements of $\mathsf M$; buyer~$i$ has a class $\tau$ with probability $\mathrm{mult}_{\mathsf M}(\tau)/n$, where $\mathrm{mult}_{\mathsf M}(\tau)$ is the multiplicity of $\tau$ in $\mathsf M$; and conditionally on this, $\bm v_{-i}$ has the distribution of $\bm v_{-i}$ conditioned on $\mathsf m(\bm v_{-i})=\mathsf M-\tau$. Since $\bm x_i(\bm v)=C_{\mathrm{class}(\bm v_i)}(\bm y_i(\bm v_{-i}))$ and the maps $C_\tau$ are linear,
\[
\begin{aligned}
\E\Big[\sum_i\bm x_i(\bm v)\ \Big|\ \text{the classes form }\mathsf M\Big]&=\sum_\tau\frac{\mathrm{mult}_{\mathsf M}(\tau)}n\sum_iC_\tau\Big(\E\big[\bm y_i(\bm v_{-i})\,\big|\,\mathsf m(\bm v_{-i})=\mathsf M-\tau\big]\Big)\\
&=\sum_{\tau\in\mathsf M}C_\tau\big(\bar{\bm y}(\mathsf M-\tau)\big),
\end{aligned}
\]
where the first sum runs over the distinct classes of $\mathsf M$ and the last one counts multiplicity. The left side is at most $(1,1)$ by~\eqref{eq:capsupply}, and the right side is the left side of~\eqref{eq:capsupply} for the anonymous table $\bar{\bm y}$ at every report profile whose classes form $\mathsf M$; so $\bar{\bm y}$ is feasible. Next, $\mathsf m(\bm v_{-i})$ has the same distribution for every buyer; write $\Pr[\mathsf m]$ for the probability that it equals $\mathsf m$. For every function $\psi$ on $\cC_\kappa$,
\[
\sum_i\E\,\psi\big(\bm y_i(\bm v_{-i})\big)=n\sum_{\mathsf m}\Pr[\mathsf m]\cdot\frac1n\sum_i\E\big[\psi(\bm y_i(\bm v_{-i}))\,\big|\,\mathsf m(\bm v_{-i})=\mathsf m\big],
\]
and the inner average is the average of $\psi$ over a distribution with mean $\bar{\bm y}(\mathsf m)$. For $\psi=\|\cdot\|_1$, which is linear on $\cC_\kappa$, this gives $\cY(\bm y)=n\sum_{\mathsf m}\Pr[\mathsf m]\,\|\bar{\bm y}(\mathsf m)\|_1=\cY(\bar{\bm y})$, and for concave $\psi$, Jensen's inequality gives the last claim.
\end{proof}

An anonymous table $\bm w$ is determined by the pairs $\bm w(\mathsf m)$. If $\mathsf m$ consists of $n-l$ copies of $Z$, $l-1-|S|$ copies of $B$ and the single-item classes $S$, we write $\bm w(l,S)=\bm w(\mathsf m)$. At a report profile whose classes consist of $n-l$ copies of $Z$, $l-|S'|$ copies of $B$ and the single-item classes $S'$, the buyers in $Z$ receive nothing, so~\eqref{eq:capsupply} reads
\[
(l-|S'|)\,\bm w(l,S')+\sum_{\tau\in S'}C_\tau\big(\bm w(l,S'-\tau)\big)\le(1,1).
\]
This is the supply constraint of the anonymous table $\bm w(l,\cdot)$ of a subproblem with $l$ buyers at the class profile consisting of $S'$ and $l-|S'|$ buyers of class $B$. Hence $\bm w$ is feasible if and only if, for every $l\le n$, the table $\bm w(l,\cdot)$ is a feasible anonymous table of a subproblem with $l$ buyers. Moreover, the opponents of a buyer have exactly $l-1$ classes other than $Z$ with probability $\binom{n-1}{l-1}(1-p^2)^{l-1}p^{2(n-l)}$, and conditionally on this, these classes are independent with the class probabilities conditional on not being in $Z$. So for every function $\psi$ on $\cC_\kappa$,
\begin{equation}\label{eq:blocks}
\sum_i\E\,\psi\big(\bm w(\mathsf m(\bm v_{-i}))\big)=\sum_{l=1}^n\binom nl(1-p^2)^{l-1}p^{2(n-l)}\cdot l\,\E\,\psi\big(\bm w(l,S)\big),
\end{equation}
where on the right $S$ is the multiset of single-item classes among $l-1$ independent buyers with the conditional class probabilities.

\subsubsection{The narrow cone}\label{app:narrow}

\begin{lemma}[Bounds and contraction]\label{lem:narrowquant}
Let $\chi=\max_{0\le j\le J}(1-\kappa\beta_j)$, suppose that $n\chi\le1/4$, and let $l\le n$.
\begin{enumerate}[label=(\alph*),nosep]
\item For every $S$, we have $\frac3{4n_B(S)}\le U_1(S),U_2(S)\le\frac1{n_B(S)}$ and $\bm w^*_l(S)\le\big(\frac1{n_B(S)},\frac1{n_B(S)}\big)$. The table $\bm w^*_l$ is feasible, and at every class profile with a buyer of class $B$, at least one of the two supply constraints holds with equality.
\item Let $\alpha\le1/100$ be the total probability of the single-item classes conditionally on not being in $Z$. Every feasible anonymous table $\bm w$ of a subproblem with $l$ buyers satisfies
\begin{equation}\label{eq:contraction}
l\,\E\big\|\bm w^*_l(S)\big\|_1-l\,\E\big\|\bm w(S)\big\|_1\ \ge\ l\,c_\alpha\,\E\big\|\bm w(S)-\bm w^*_l(S)\big\|_1,\qquad c_\alpha=1-\frac{2(1+\kappa^{-1})\,\alpha}{1-\alpha}\ge\frac9{10},
\end{equation}
where $S$ is the multiset of single-item classes among $l-1$ independent buyers with the conditional class probabilities.
\end{enumerate}
\end{lemma}

\begin{proof}
For $\bm y\in\cC_\kappa$, the inequalities $y_2\ge\kappa y_1$ and $\beta_jy_2\le\kappa y_2\le y_1$ give
\begin{equation}\label{eq:chi}
0\le y_1-\beta_jy_2\le(1-\kappa\beta_j)\,y_1\le\chi\,y_1 ,
\end{equation}
and the classes $Y_j$ satisfy the symmetric bound.

(a) We show by induction on $|S|$ that $(1-l\chi)/n_B\le U_1(S),U_2(S)\le1/n_B$ and $\bm w^*_l(S)\le(1/n_B,1/n_B)$, where $n_B=n_B(S)$; since $l\chi\le n\chi\le\frac14$, this gives the bounds of~(a). For $S=\emptyset$ this is clear. For $S\ne\emptyset$ and $\tau\in S$, we have $n_B(S-\tau)=n_B+1$, so by induction $\bm w^*_l(S-\tau)\in\cC_\kappa$ has coordinates at most $1/(n_B+1)$; hence every term $C_\tau(\bm w^*_l(S-\tau))$ in~\eqref{eq:greedy} is nonnegative, and by~\eqref{eq:chi} its only possibly nonzero coordinate is at most $\chi/(n_B+1)$. The sum therefore has each coordinate at most $|S|\chi/(n_B+1)\le l\chi$, which gives the bounds on $U_1(S)$ and $U_2(S)$. In particular $0\le U_1(S),U_2(S)\le1$, so $\bm w^*_l(S)\in\cC_\kappa$ and $\bm w^*_l(S)\le(U_1(S),U_2(S))\le(1/n_B,1/n_B)$.

At a class profile with $n_B\ge1$ buyers of class $B$ and single-item classes $S$, so that every buyer of class $B$ faces $S$, the total allocation is
\[
n_B\,\bm w^*_l(S)+\sum_{\tau\in S}C_\tau\big(\bm w^*_l(S-\tau)\big)\le n_B\big(U_1(S),U_2(S)\big)+\sum_{\tau\in S}C_\tau\big(\bm w^*_l(S-\tau)\big)=(1,1),
\]
with equality in the first coordinate if $U_1(S)\le U_2(S)/\kappa$, because then $\bm w^*_l(S)_1=U_1(S)$, and in the second coordinate otherwise, because then $U_1(S)>U_2(S)/\kappa\ge\kappa U_2(S)$ and $\bm w^*_l(S)_2=U_2(S)$. At a class profile without buyers of class $B$, every buyer has $l-1$ opponents of single-item classes, her capacities are at most $(1,1)$, and by~\eqref{eq:chi} she receives at most $\chi$ of a single item; so the total allocation of each item is at most $l\chi<1$. Hence $\bm w^*_l$ is feasible.

(b) For an anonymous table $\bm w$ of the subproblem and a context $S$ with $|S|=k$, let
\[
\big(U^{\bm w}_1(S),U^{\bm w}_2(S)\big)=\frac{(1,1)-\sum_{\tau\in S}C_\tau\big(\bm w(S-\tau)\big)}{l-k};
\]
for $\bm w=\bm w^*_l$, this is the pair $(U_1(S),U_2(S))$ of~\eqref{eq:greedy}. Let $\bm w$ be feasible. At the class profile consisting of $S$ and $l-k$ buyers of class $B$, each buyer of class $B$ receives $\bm w(S)$, and a buyer of class $\tau\in S$ receives $C_\tau(\bm w(S-\tau))$; so feasibility gives $\bm w(S)\le(U^{\bm w}_1(S),U^{\bm w}_2(S))$. In particular, $U^{\bm w}_1(S),U^{\bm w}_2(S)\ge0$, and they are at most $1$ because the terms $C_\tau(\bm w(S-\tau))$ are nonnegative. Since $\bm w(S)\in\cC_\kappa$, the first property of $G_\kappa$ gives $\bm w(S)\le G_\kappa(U^{\bm w}_1(S),U^{\bm w}_2(S))$. With the identity $\|\bm\xi-\bm\xi'\|_1=\|\bm\xi'\|_1-\|\bm\xi\|_1+2\pos{\xi_1-\xi'_1}+2\pos{\xi_2-\xi'_2}$ for $\bm\xi,\bm\xi'\in\R^2_{\ge0}$ and the second property of $G_\kappa$, we get
\begin{align*}
\|\bm w(S)-\bm w^*_l(S)\|_1&\le\|\bm w^*_l(S)\|_1-\|\bm w(S)\|_1+2\big\|G_\kappa(U^{\bm w}_1(S),U^{\bm w}_2(S))-G_\kappa(U_1(S),U_2(S))\big\|_1\\
&\le\|\bm w^*_l(S)\|_1-\|\bm w(S)\|_1+2(1+\kappa^{-1})\big\|(U^{\bm w}_1(S),U^{\bm w}_2(S))-(U_1(S),U_2(S))\big\|_1 ,
\end{align*}
and, by~\eqref{eq:Cnorm}, $\|(U^{\bm w}_1(S),U^{\bm w}_2(S))-(U_1(S),U_2(S))\|_1\le\frac1{l-k}\sum_{\tau\in S}\|\bm w(S-\tau)-\bm w^*_l(S-\tau)\|_1$.

From now on, $S$ is random, as in~\eqref{eq:contraction}. Let $\varpi_k=\binom{l-1}k\alpha^k(1-\alpha)^{l-1-k}$ be the probability that $|S|=k$, and let
\begin{align*}
D_k&=\E\big[\|\bm w(S)-\bm w^*_l(S)\|_1\ \big|\ |S|=k\big],\\
\Theta_k&=\E\big[\|\bm w^*_l(S)\|_1-\|\bm w(S)\|_1\ \big|\ |S|=k\big].
\end{align*}
Conditionally on $|S|=k$, the context $S$ consists of $k$ independent single-item classes, and removing one of them leaves $k-1$ independent single-item classes; so the conditional expectation of $\sum_{\tau\in S}\|\bm w(S-\tau)-\bm w^*_l(S-\tau)\|_1$ is $kD_{k-1}$. Taking conditional expectations in the inequality above gives
\[
D_k\le\Theta_k+2(1+\kappa^{-1})\frac k{l-k}D_{k-1}\quad(k\ge1),\qquad D_0=\Theta_0 ,
\]
the latter because feasibility at the profile in which all buyers are of class $B$ gives $\bm w(\emptyset)\le(1/l,1/l)=\bm w^*_l(\emptyset)$. Since $\varpi_k\frac k{l-k}=\varpi_{k-1}\frac\alpha{1-\alpha}$, multiplying by $\varpi_k$ and summing gives
\[
\sum_{k=0}^{l-1}\varpi_kD_k\le\sum_{k=0}^{l-1}\varpi_k\Theta_k+(1-c_\alpha)\sum_{k=0}^{l-2}\varpi_kD_k ,
\]
so $c_\alpha\sum_k\varpi_kD_k\le\sum_k\varpi_k\Theta_k$, which is~\eqref{eq:contraction}. Finally, $2(1+\kappa^{-1})\le6$ and $\alpha\le\frac1{100}$ give $c_\alpha\ge1-\frac6{99}\ge\frac9{10}$.
\end{proof}

\begin{proof}[Proof of \Cref{lem:narrow}]
By \Cref{lem:narrowquant}(a), $\bm w^*_l$ is feasible. Every context $S$ for $l$ buyers has positive probability, so by \Cref{lem:narrowquant}(b), every feasible anonymous table $\bm w$ of the subproblem other than $\bm w^*_l$ has $\E\|\bm w(S)\|_1<\E\|\bm w^*_l(S)\|_1$. Let $\bm y$ be any feasible table. By \Cref{lem:average}, $\bar{\bm y}$ is a feasible anonymous table with $\cY(\bar{\bm y})=\cY(\bm y)$, so~\eqref{eq:blocks} with $\psi=\|\cdot\|_1$ gives $\cY(\bm y)\le\cY(\bm w^{\mathrm{ex}})$, where $\bm w^{\mathrm{ex}}$ is feasible because every table $\bm w^*_l$ is. Hence $\Capac^*_n=\cY(\bm w^{\mathrm{ex}})$, and~\eqref{eq:blocks} gives $\cY(\bm w^{\mathrm{ex}})=n\,\E\|\bm w^*_l(S)\|_1$ with $l$ and $S$ as in the lemma.
\end{proof}

\subsection{A very high value}\label{app:high}

In this subsection, $F=F_p$ and $f=f_p$ for some $p\in[d/2,3d/2]$, $\gamma=\gamma_p$, and $\sigma\in(0,\gamma)$, $\underline\gamma$, $\pi$ and $F_H$ are as in~\eqref{eq:FH}, with $H>c_J$. The types in $\supp(F_H)^2$ are the base types $\bm t\in\cT$, of probability $(1-\pi)^2f(\bm t)$; the types $(H,a)$ and $(a,H)$ with $a\in\supp F$, of probability $\pi(1-\pi)f(a)$ each; and $(H,H)$, of probability $\pi^2$. All support points of $F$ lie in $(1-2d,1+3d)$ by (N3), so $1<c_J<2$ and $|a-1|\le3d$ for $a\in\supp F$.

For an IC, IR menu $Q$ on $\cT$, the IC, IR menus on $\supp(F_H)^2$ that agree with $Q$ on $\cT$ form a polyhedron. It is nonempty: let every type with a value $H$ choose an option of $Q$ of maximal utility; since allocations are nonnegative, the type $(H,a)$ obtains at least the utility $u_Q(c_J,a)\ge0$ of the base type $(c_J,a)$, and similarly for $(a,H)$ and $(H,H)$. It is bounded: by IR, payments are at most $2H$, and by IC for the lowest type $\ell_J\one$, at least $-u_Q(\ell_J\one)$. Hence it is a polytope, and the maximum $V_H(Q)$ is attained.

Two facts are used repeatedly. The first is \emph{filling}: given an IC, IR menu $(\bm x',s')$ on $\supp(F_H)^2$, raise $x'_k(\bm t)$ to $1$ for every type $\bm t$ and item $k$ with $t_k=H$, and add $H(1-x'_k(\bm t))$ to the payment of $\bm t$ for each such $k$. The utility of every type from its own option is unchanged, and the utility of a type $\bm t'$ from the option of $\bm t$ changes by $-\sum_{k:t_k=H}(H-t'_k)(1-x'_k(\bm t))\le0$. So the new menu is IC and IR, it agrees with the old one on $\cT$, and its revenue is larger by $H\,\E_{\bm t}\sum_{k:t_k=H}(1-x'_k(\bm t))$, where $\bm t\sim F_H\otimes F_H$. The second fact concerns menus for a single value.

\begin{lemma}[Posted prices]\label{lem:posted}
Let $G$ be a distribution on a finite set $\cZ$ of positive numbers, and let $(x(z),s(z))\in\R^2$, $z\in\cZ$, satisfy $z\,x(z)-s(z)\ge z\,x(z')-s(z')$ for all $z,z'\in\cZ$ and $z\,x(z)-s(z)\ge0$ for $z=\min\cZ$. Put $\lambda_z=x(z)-x(z^-)$, where $z^-$ is the predecessor of $z$ in $\cZ$, and $x(z^-)=0$ for $z=\min\cZ$. Then $\lambda_z\ge0$ for $z>\min\cZ$, and $s(z)\le\sum_{z'\le z}\lambda_{z'}z'$ for every $z$, the sums running over $\cZ$. Consequently, with $R_G(z)=z\,G(\{z'\ge z\})$,
\[
\sum_zG(\{z\})\,s(z)\le\sum_z\lambda_z\,R_G(z),\qquad \hat z\,x(\max\cZ)-s(\max\cZ)\ge\sum_z\lambda_z\,(\hat z-z)\quad(\hat z\in\R).
\]
\end{lemma}

Thus a menu for a single value earns at most the revenue of the mixture $\lambda$ of posted prices with the same allocations, and its top option gives every value $\hat z$ at least the utility $\sum_z\lambda_z(\hat z-z)$ of buying the whole mixture.

\begin{proof}
Adding the IC inequalities between $z$ and $z'$ gives $(z-z')\big(x(z)-x(z')\big)\ge0$, so $x$ is nondecreasing. The utility $u(z)=z\,x(z)-s(z)$ satisfies $u(\min\cZ)\ge0$ and $u(z)\ge z\,x(z^-)-s(z^-)=u(z^-)+(z-z^-)\,x(z^-)$ for $z>\min\cZ$. Hence $u(z)\ge\sum_{z'<z}(z'^+-z')\,x(z')$, where $z'^+$ is the successor of $z'$ in $\cZ$, and summation by parts gives
\[
s(z)=z\,x(z)-u(z)\le z\,x(z)-\sum_{z'<z}(z'^+-z')\,x(z')=\sum_{z'\le z}z'\,\lambda_{z'} .
\]
Multiplying by $G(\{z\})$ and summing over $z$ gives the first consequence; the second follows from $x(\max\cZ)=\sum_z\lambda_z$.
\end{proof}

\subsubsection{Exact separation}\label{app:separation}

\begin{lemma}[Exact separation]\label{lem:separation}
Suppose that $(1-\pi)\rev^*+\pi c_J<\underline\gamma$. For every DSIC, ex-post IR mechanism $M$ for $F$, the largest revenue of a DSIC, ex-post IR mechanism for $F_H$ with base restriction $M$ is
\[
(1-\pi)^{2n-2}\sum_{i=1}^n\E\,V_H\big(M_i(\bm v_{-i})\big)+2H\Big[1-(1-\pi)^n-n\pi(1-\pi)^{2n-2}\Big],
\]
where $M_i(\bm v_{-i})$ is the menu of buyer~$i$ at the base reports $\bm v_{-i}\sim F^{\otimes(2n-2)}$.
\end{lemma}

\begin{proof}
Let $M'$ be a DSIC, ex-post IR mechanism for $F_H$ with base restriction $M$. For a profile $\bm v$, let $K(\bm v)$ be the set of items $k$ for which some buyer reports the value $H$. We bound the expected payment of buyer~$i$ given the reports $\bm v_{-i}$ of the others, in two cases.

If all reports in $\bm v_{-i}$ are base reports, buyer~$i$ faces a menu on $\supp(F_H)^2$ that agrees with $M_i(\bm v_{-i})$ on $\cT$. Filling it gives a menu that still agrees with $M_i(\bm v_{-i})$, so
\begin{equation}\label{eq:sepbase}
\E_{\bm v_i}\big[s_i(\bm v)\big]\le V_H\big(M_i(\bm v_{-i})\big)-H\,\E_{\bm v_i}\sum_{k:\,v_{ik}=H}\big(1-x_{ik}(\bm v)\big).
\end{equation}

Otherwise, choose an item $k$ for which some other buyer reports $H$, and fix also the value $v_{ik}$ of buyer~$i$ for item~$k$. As the value $a=v_{i,3-k}\sim F_H$ of the other item varies, the allocations $x_{i,3-k}$ and the numbers $s_i-v_{ik}x_{ik}$ satisfy the hypotheses of \Cref{lem:posted} with $G=F_H$, because $M'$ is DSIC and ex-post IR. The allocations lie in $[0,1]$, so all $\lambda_z$ are nonnegative, and their sum is the allocation at $a=H$. A price $z$ earns $R_{F_H}(H)=\pi H=\underline\gamma$ if $z=H$, and $R_{F_H}(z)=\pi z+(1-\pi)\rev(z)\le\pi c_J+(1-\pi)\rev^*<\underline\gamma$ if $z\in\supp F$. Hence the expectation of $s_i-v_{ik}x_{ik}$ over $a$ is at most $\underline\gamma$ times the allocation at $a=H$, that is, at most $\E_a\big[H\ind{a=H}x_{i,3-k}\big]$. Averaging over $v_{ik}$ and using $v_{ik}\le H$, $k\in K(\bm v)$, and $3-k\in K(\bm v)$ whenever $v_{i,3-k}=H$, we get
\begin{equation}\label{eq:sephigh}
\E_{\bm v_i}\big[s_i(\bm v)\big]\le\E_{\bm v_i}\big[v_{ik}x_{ik}(\bm v)+H\ind{v_{i,3-k}=H}\,x_{i,3-k}(\bm v)\big]\le H\,\E_{\bm v_i}\sum_{k'\in K(\bm v)}x_{ik'}(\bm v).
\end{equation}

The reports are independent, so we may average~\eqref{eq:sepbase} and~\eqref{eq:sephigh} over $\bm v_{-i}$ and sum over $i$. The terms $V_H(M_i(\bm v_{-i}))$ contribute $(1-\pi)^{2n-2}\sum_i\E\,V_H(M_i(\bm v_{-i}))$, and the remaining terms contribute $H\,\E[\Psi(\bm v)]$, where $\Psi(\bm v)$ is the sum of $\sum_{k\in K(\bm v)}x_{ik}(\bm v)$ over the buyers $i$ who face a report with a value $H$, minus the sum of $\sum_{k:v_{ik}=H}(1-x_{ik}(\bm v))$ over the other buyers. If no buyer reports a value $H$, then $\Psi(\bm v)=0$. If exactly one buyer $j$ does, then $K(\bm v)=\{k:v_{jk}=H\}$, all other buyers face a report with a value $H$, and $\Psi(\bm v)=\sum_{k\in K(\bm v)}\big(\sum_ix_{ik}(\bm v)-1\big)\le0$ by supply. If at least two buyers do, then every buyer faces a report with a value $H$, and $\Psi(\bm v)\le|K(\bm v)|$ by supply. Item~$k$ belongs to $K(\bm v)$ with probability $1-(1-\pi)^n$, and it belongs to $K(\bm v)$ while exactly one buyer reports a value $H$ with probability $n\pi(1-\pi)^{2n-2}$. Hence
\[
\Rev(M')\le(1-\pi)^{2n-2}\sum_i\E\,V_H\big(M_i(\bm v_{-i})\big)+2H\big[1-(1-\pi)^n-n\pi(1-\pi)^{2n-2}\big].
\]

For the reverse inequality, let $Q^*_i(\bm v_{-i})$ attain $V_H(M_i(\bm v_{-i}))$, for every buyer~$i$ and base reports $\bm v_{-i}$. Define $M'$ as follows. If all reports are base reports, run $M$. If exactly one buyer $j$ reports a value $H$, give $j$ the outcome of $Q^*_j(\bm v_{-j})$ at her report, and give nothing to the others. If at least two buyers report a value $H$, sell each item $k\in K(\bm v)$ at the price $H$ to the buyer of smallest index who reports $H$ for it. Supply holds in all three cases. If all opponents of buyer~$i$ report base types, her menu is $Q^*_i(\bm v_{-i})$, which is IC and IR. Otherwise she receives nothing when she reports a base type, and when she reports a type with a value $H$, she receives some of the items that she reports at $H$, at the price $H$ each; every option gives a type utility at most $0$, and truthful reporting gives exactly $0$. So $M'$ is DSIC and ex-post IR. At a profile at which at most one buyer reports a value $H$, every buyer whose opponents all report base types receives the outcome of $Q^*_i(\bm v_{-i})$ at her report, and the other buyers receive nothing. So the revenue of $M'$ from these profiles is $(1-\pi)^{2n-2}\sum_i\E\,V_H(M_i(\bm v_{-i}))$, and its revenue from the other profiles is $H\,\E\big[|K(\bm v)|\ind{\text{at least two buyers report a value }H}\big]=2H[1-(1-\pi)^n-n\pi(1-\pi)^{2n-2}]$.
\end{proof}

\subsubsection{The frontier}\label{app:frontier}

Let $\rho=\max_{z\in\supp F}\rev(z)/(H-z)$, the largest revenue of a price per unit of utility that it leaves to the value $H$; let $a^*$ be the largest point of $\supp F$ with $\rev(a^*)=\rev^*$, and put $h^*=H-a^*$.

\begin{lemma}[Increment menus]\label{lem:increment}
Let $r,T\ge0$, and let $(x(a),s(a))_{a\in\supp F}$ satisfy, for all $a,a'\in\supp F$,
\[
a\,x(a)-s(a)\ge a\,x(a')-s(a'),\qquad a\,x(a)-s(a)\ge0,\qquad x(a)\le r,\qquad H\,x(a)-s(a)\le T .
\]
Then $\sum_af(a)\,s(a)\le \cF_r(T)$.
\end{lemma}

Here $(x(a),s(a))$ is a menu for a single value $a\sim F$ whose allocations may be negative; the lemma bounds its revenue by the frontier.

\begin{proof}
Consider the menu consisting of the options $(x(a),s(a))$ with $s(a)>0$ and the option $(0,0)$, and let every $a\in\supp F$ choose an option of maximal utility in it, breaking ties in favor of the larger payment. If $s(a)>0$, the option $(x(a),s(a))$ is available and optimal for $a$, so $a$ pays at least $s(a)$; otherwise $a$ pays at least $0\ge s(a)$. Every retained option has $x(a)\ge s(a)/a>0$. So the new menu has allocations in $[0,r]$, it leaves utility at most $T$ to the value $H$, and its expected payment is at least $\sum_af(a)s(a)$. By \Cref{lem:posted} with $G=F$, this expected payment is at most $\sum_z\lambda_z\rev(z)$ for some $\lambda\ge0$ with $\sum_z\lambda_z\le r$ and $\sum_z\lambda_z(H-z)\le T$, hence at most $\cF_r(T)$.
\end{proof}

\begin{lemma}[The frontier]\label{lem:frontierprops}
\begin{enumerate}[label=(\alph*),nosep]
\item $\Gamma_H$ is concave and nondecreasing in each argument on $\R^2_{\ge0}$, $0\le\Gamma_H(\bm r)\le\rev^*\|\bm r\|_1$, and $|\Gamma_H(\bm r)-\Gamma_H(\bm r')|\le\rev^*\|\bm r-\bm r'\|_1$.
\item For $r\ge0$, $\cF_r$ is concave and nondecreasing, $\cF_r(T)\le\rho T$, and $\cF_r(T')-\cF_r(T)\le\rho(T'-T)$ for $T'\ge T$. Moreover, $\cF_r(T)=\rev^*r$ for $T\ge h^*r$, and every $\lambda$ attaining $\cF_r(T)$ with $T\le h^*r$ satisfies $\sum_z\lambda_z(H-z)=T$.
\item If $\pi/(1-\pi)>\rho$, then for all $r_1,r_2\ge0$ the maximum in~\eqref{eq:Gamma} is attained at some $T^*\le h^*\min\{r_1,r_2\}$.
\end{enumerate}
\end{lemma}

\begin{proof}
(a) Write $\hat\pi=\pi/(1-\pi)$. By~\eqref{eq:frontier}, $\Gamma_H(\bm r)$ is the largest value of $\sum_z(\lambda^{(1)}_z+\lambda^{(2)}_z)\rev(z)-\hat\pi T$ over $\lambda^{(1)},\lambda^{(2)}\ge0$ and $T\ge0$ with $\sum_z\lambda^{(k)}_z\le r_k$ and $\sum_z\lambda^{(k)}_z(H-z)\le T$ for $k=1,2$. We may restrict to $T\le(H-\ell_J)\max\{r_1,r_2\}$, a compact set, so the maximum exists. A convex combination of feasible points for $\bm r$ and $\bm r'$ is feasible for the same convex combination of $\bm r$ and $\bm r'$, which gives concavity; the feasible set grows with $\bm r$, which gives monotonicity; $\lambda=0$, $T=0$ give $\Gamma_H\ge0$; and $\sum_z\lambda^{(k)}_z\rev(z)\le\rev^*r_k$ gives the upper bound. If $\bm r'\ge\bm r$, take an optimal point for $\bm r'$ and decrease the $\lambda^{(k)}_z$ until $\sum_z\lambda^{(k)}_z\le r_k$; this removes a quantity of at most $r'_k-r_k$ from each row, keeps the point feasible for $\bm r$ with the same $T$, and lowers the value by at most $\rev^*\|\bm r'-\bm r\|_1$. For general $\bm r$ and $\bm r'$, compare both with their coordinatewise minimum.

(b) Concavity and monotonicity follow as in~(a), and $\cF_r(0)=0$. For feasible $\lambda$, $\sum_z\lambda_z\rev(z)\le\rho\sum_z\lambda_z(H-z)\le\rho T$, so $\cF_r(T)\le\rho T$. For $T'>T>0$, concavity bounds the slope of the chord: $\big(\cF_r(T')-\cF_r(T)\big)/(T'-T)\le \cF_r(T)/T\le\rho$; for $T=0$ the claim is $\cF_r(T')\le\rho T'$. Next, $\cF_r\le\rev^*r$, with equality for $T\ge h^*r$, attained by $\lambda=r\,\delta_{a^*}$. Let $T\le h^*r$, and suppose that some $\lambda$ attaining $\cF_r(T)$ has $T_\lambda:=\sum_z\lambda_z(H-z)<T$. Then $\cF_r(T_\lambda)=\cF_r(T)$, so $\cF_r$, being concave and nondecreasing, is constant on $[T_\lambda,\infty)$, and $\lambda$ attains $\cF_r(T_\lambda)=\rev^*r$. This forces $\sum_z\lambda_z=r$ and $\lambda_z=0$ unless $\rev(z)=\rev^*$, hence unless $z\le a^*$; but then $T_\lambda\ge h^*r\ge T$, a contradiction.

(c) By symmetry let $r_1\le r_2$, and put $T_0=h^*r_1$. For $T>T_0$, (b) gives $\cF_{r_1}(T)=\cF_{r_1}(T_0)$ and $\cF_{r_2}(T)-\cF_{r_2}(T_0)\le\rho(T-T_0)<\hat\pi(T-T_0)$, so the function maximized in~\eqref{eq:Gamma} is smaller at $T$ than at $T_0$. On $[0,T_0]$ this function is continuous, because $\cF_r$ is concave on $[0,\infty)$ and $0\le \cF_r(T)\le\rho T$; so it attains its maximum there.
\end{proof}

\subsubsection{Menus with a high value}\label{app:canonext}

For a capacity pair $\bm y$, put $\bar s_1=H(1-y_1)+\|\bm y\|_1$ and $\bar s_2=H(1-y_2)+\|\bm y\|_1$. The option with allocation $(1,y_2)$ and payment $\bar s_1$ is the option $B$ of the canonical menu $Q_{\bm y}$, filled for a value $H$ of item~$1$; it gives the type $(H,a)$ the utility $Hy_1+a\,y_2-\|\bm y\|_1$, and the type $(H,H)$ the utility $(H-1)\|\bm y\|_1$. The next lemma bounds the revenue from the types with a value $H$ for any menu on $\cT$, given how much utility its options leave to these types.

\begin{lemma}[High types]\label{lem:highbound}
Let $Q=(\bm x,s)$ be an IC, IR menu on $\cT$ with capacities $\bm y$, and let $\omega_1,\omega_2\ge0$. Suppose that for every $a\in\supp F$ there are base types $\bm t,\bm t'$ with
\[
Hx_1(\bm t)+a\,x_2(\bm t)-s(\bm t)\ge Hy_1+a\,y_2-\|\bm y\|_1-\omega_1,\qquad a\,x_1(\bm t')+Hx_2(\bm t')-s(\bm t')\ge a\,y_1+Hy_2-\|\bm y\|_1-\omega_1,
\]
and that some base type $\bm t''$ satisfies $H\|\bm x(\bm t'')\|_1-s(\bm t'')\ge(H-1)\|\bm y\|_1-\omega_2$. Then
\[
V_H(Q)\le(1-\pi)^2\Rev_F(Q)+2\underline\gamma+(2\pi-\pi^2-\underline\gamma)\|\bm y\|_1+\pi(1-\pi)\,\Gamma_H(\one-\bm y)+2\pi\omega_1+\pi^2\omega_2 .
\]
\end{lemma}

\begin{proof}
Let $(\bm x',s')$ be an IC, IR menu on $\supp(F_H)^2$ that agrees with $Q$ on $\cT$. By filling, we may assume that $x'_1(H,a)=1$, $x'_2(a,H)=1$ and $\bm x'(H,H)=(1,1)$ for all $a$. Put $\bm r=\one-\bm y$.

For $a\in\supp F$, let $x^{(2)}(a)=x'_2(H,a)-y_2$ and $s^{(2)}(a)=s'(H,a)-\bar s_1-\omega_1$. For $a,a'\in\supp F$, the number $a\,x^{(2)}(a')-s^{(2)}(a')$ is the utility of $(H,a)$ from the option of $(H,a')$ minus $(Hy_1+a\,y_2-\|\bm y\|_1-\omega_1)$, and $H\,x^{(2)}(a')-s^{(2)}(a')$ is the utility of $(H,H)$ from the option of $(H,a')$ minus $((H-1)\|\bm y\|_1-\omega_1)$. Hence IC and the hypothesis give the first two conditions of \Cref{lem:increment}, and $x^{(2)}(a)\le1-y_2=r_2$. Define $x^{(1)}$ and $s^{(1)}$ symmetrically on the types $(a,H)$, and let $T\ge0$ be the largest of $0$ and the numbers $H\,x^{(k)}(a)-s^{(k)}(a)$. By \Cref{lem:increment},
\[
\sum_af(a)\,s'(H,a)\le \bar s_1+\omega_1+\cF_{r_2}(T),\qquad \sum_af(a)\,s'(a,H)\le \bar s_2+\omega_1+\cF_{r_1}(T).
\]
The type $(H,H)$ may choose the option of $\bm t''$ and those of the types $(H,a)$ and $(a,H)$, so its utility is at least $(H-1)\|\bm y\|_1+T-\omega_1-\omega_2$, and it pays at most $2H-(H-1)\|\bm y\|_1-T+\omega_1+\omega_2$. We weight the types with their probabilities and use $\pi(1-\pi)(\bar s_1+\bar s_2)+\pi^2\big(2H-(H-1)\|\bm y\|_1\big)=2\pi H-(\pi H-2\pi+\pi^2)\|\bm y\|_1$ and $\pi H=\underline\gamma$. The revenue of $(\bm x',s')$ is therefore at most
\[
(1-\pi)^2\Rev_F(Q)+2\underline\gamma+(2\pi-\pi^2-\underline\gamma)\|\bm y\|_1+\pi(1-\pi)\Big[\cF_{r_1}(T)+\cF_{r_2}(T)-\frac\pi{1-\pi}T\Big]+\pi(2-\pi)\omega_1+\pi^2\omega_2 ,
\]
and the bracket is at most $\Gamma_H(\bm r)$, while $\pi(2-\pi)\le2\pi$.
\end{proof}

\begin{lemma}[Extension of a canonical menu]\label{lem:canonext}
Suppose that $(1-\pi)\rev^*+\pi c_J<\underline\gamma$ and $H\ge3c_J$. For every $\bm y\in\cC_\kappa$, the canonical menu $Q_{\bm y}$ with capacities $\bm y$ satisfies
\[
V_H(Q_{\bm y})=2\underline\gamma+\sigma_H\|\bm y\|_1+\pi(1-\pi)\,\Gamma_H(\one-\bm y),\qquad \sigma_H=\sigma+(1-\gamma)(2\pi-\pi^2).
\]
\end{lemma}

\begin{proof}
Let $\bm y\in\cC_\kappa$ and $\bm r=\one-\bm y$. The canonical menu earns $\Rev_F(Q_{\bm y})=\gamma\|\bm y\|_1$, and $(1-\pi)^2\gamma+2\pi-\pi^2-\underline\gamma=\sigma_H$. Its option $B$ satisfies the hypotheses of \Cref{lem:highbound} with $\omega_1=\omega_2=0$, so $V_H(Q_{\bm y})\le2\underline\gamma+\sigma_H\|\bm y\|_1+\pi(1-\pi)\Gamma_H(\bm r)$.

For the reverse inequality, we extend $Q_{\bm y}$. The hypothesis gives $(1-\pi)\rev^*<\pi(H-c_J)$, hence $\pi/(1-\pi)>\rho$. By \Cref{lem:frontierprops}(c) and~(b), the maximum in~\eqref{eq:Gamma} is attained at some $T^*\le h^*\min\{r_1,r_2\}$, and there are $\lambda^{(1)},\lambda^{(2)}\ge0$ with $\sum_z\lambda^{(k)}_z\le r_k$, $\sum_z\lambda^{(k)}_z(H-z)=T^*$ and $\sum_z\lambda^{(k)}_z\rev(z)=\cF_{r_k}(T^*)$. For $a\in\supp F$, let $x^{(k)}(a)=\sum_{z\le a}\lambda^{(k)}_z$ and $s^{(k)}(a)=\sum_{z\le a}\lambda^{(k)}_zz$ be the allocation and the payment of the value $a$ under the mixture $\lambda^{(k)}$ of posted prices. Keep the outcomes of $Q_{\bm y}$ on $\cT$, and give the other types the following outcomes:
\[
\begin{array}{lll}
\text{type} & \text{allocation} & \text{payment}\\ \hline
(H,a) & (1,\ y_2+x^{(2)}(a)) & \bar s_1+s^{(2)}(a)\\
(a,H) & (y_1+x^{(1)}(a),\ 1) & \bar s_2+s^{(1)}(a)\\
(H,H) & (1,\ 1) & 2H-(H-1)\|\bm y\|_1-T^*
\end{array}
\]
The allocations lie in $[0,1]$, and the revenue is $(1-\pi)^2\gamma\|\bm y\|_1+\pi(1-\pi)\big[\bar s_1+\bar s_2+\cF_{r_1}(T^*)+\cF_{r_2}(T^*)\big]+\pi^2\big[2H-(H-1)\|\bm y\|_1-T^*\big]=2\underline\gamma+\sigma_H\|\bm y\|_1+\pi(1-\pi)\Gamma_H(\bm r)$.

It remains to check IC and IR. Posted prices give the value $a$ the utility $\phi_k(a)=a\,x^{(k)}(a)-s^{(k)}(a)=\sum_{z\le a}\lambda^{(k)}_z(a-z)\ge0$, and from the option of any $a'$ the utility $\sum_{z\le a'}\lambda^{(k)}_z(a-z)\le\phi_k(a)$; the value $H$ obtains at most $\sum_z\lambda^{(k)}_z(H-z)=T^*$ from any option; and $\phi_k(a)\ge\sum_z\lambda^{(k)}_z(a-z)\ge T^*-(H-a)r_k$. Moreover, $\sum_z\lambda^{(k)}_z\le T^*/(H-c_J)\le2\min\{r_1,r_2\}$, because $T^*\le H\min\{r_1,r_2\}$ and $H\ge2c_J$. A type $\bm t$ obtains from the option of $(H,a')$ the utility
\[
\bm t\cdot\bm y-\|\bm y\|_1-(H-t_1)r_1+\big(t_2\,x^{(2)}(a')-s^{(2)}(a')\big),
\]
from the option of $(H,H)$ the utility $\bm t\cdot\bm y-\|\bm y\|_1-(H-t_1)r_1-(H-t_2)r_2+T^*$, and from the option $B$ the utility $\bm t\cdot\bm y-\|\bm y\|_1$; the options of the types $(a',H)$ are symmetric. We check the deviations of each type.
\begin{itemize}[leftmargin=2em]
\item A base type $\bm t$ obtains at least $\bm t\cdot\bm y-\|\bm y\|_1$ in $Q_{\bm y}$. From the option of $(H,a')$ it obtains at most $\bm t\cdot\bm y-\|\bm y\|_1-(H-c_J)r_1+2c_J\min\{r_1,r_2\}\le \bm t\cdot\bm y-\|\bm y\|_1$, since $t_2\,x^{(2)}(a')-s^{(2)}(a')\le c_J\sum_z\lambda^{(2)}_z$ and $H\ge3c_J$. From the option of $(H,H)$ it obtains at most $\bm t\cdot\bm y-\|\bm y\|_1-(H-c_J)(r_1+r_2)+T^*\le \bm t\cdot\bm y-\|\bm y\|_1$, since $T^*\le H\min\{r_1,r_2\}\le(H-c_J)(r_1+r_2)$.
\item The type $(H,a)$ obtains $Hy_1+a\,y_2-\|\bm y\|_1+\phi_2(a)$. Every option $(\hat{\bm x},\hat s)$ of $Q_{\bm y}$ gives it at most $Hy_1+a\,y_2-\|\bm y\|_1$, the utility of the option $B$: the base type $(c_J,a)$ belongs to $B$, so it prefers $B$ to $(\hat{\bm x},\hat s)$, and $\hat x_1\le y_1$; hence $H\hat x_1+a\hat x_2-\hat s\le(c_J\hat x_1+a\hat x_2-\hat s)+(H-c_J)y_1\le Hy_1+a\,y_2-\|\bm y\|_1$. This utility is nonnegative, as it is at least that of $(c_J,a)$. The options of the types $(H,a')$ give it at most its own utility, by the IC property of posted prices. The options of $(a',H)$ and $(H,H)$ give it $Hy_1+a\,y_2-\|\bm y\|_1-(H-a)r_2+H\,x^{(1)}(a')-s^{(1)}(a')$ and $Hy_1+a\,y_2-\|\bm y\|_1-(H-a)r_2+T^*$; both are at most its own utility, because $H\,x^{(1)}(a')-s^{(1)}(a')\le T^*\le\phi_2(a)+(H-a)r_2$.
\item The type $(a,H)$ is symmetric.
\item The type $(H,H)$ obtains $(H-1)\|\bm y\|_1+T^*$. Every option of $Q_{\bm y}$ gives it at most $(H-1)\|\bm y\|_1$, by the same argument with the base type $(c_J,c_J)\in B$. The options of $(H,a')$ and $(a',H)$ give it $(H-1)\|\bm y\|_1+H\,x^{(k)}(a')-s^{(k)}(a')\le(H-1)\|\bm y\|_1+T^*$.
\end{itemize}
All utilities are nonnegative, so the extended menu is IC and IR.
\end{proof}

\subsubsection{Lowered menus}\label{app:lowered}

\begin{lemma}[Lowered menus]\label{lem:lowered}
Let $Q=(\bm x,s)$ be an IC, IR menu on $\cT$ with capacities $\bm y$, let $b\ge0$, and let $(\bm x',s')$ be the canonical menu with capacities $\bm y'\in\cC_\kappa$, where $\bm 0\le\bm x(\bm t)-\bm x'(\bm t)\le(b,b)$ and $0\le s(\bm t)-s'(\bm t)\le b$ for all $\bm t\in\cT$, and $\bm 0\le\bm y-\bm y'\le(b,b)$. Then
\[
V_H(Q)\le(1-\pi)^2\Rev_F(Q)+2\underline\gamma+(2\pi-\pi^2-\underline\gamma)\|\bm y\|_1+\pi(1-\pi)\,\Gamma_H(\one-\bm y)+6\pi b .
\]
\end{lemma}

For a canonical menu, $b=0$, and the bound is the value of \Cref{lem:canonext}. The error term is proportional to $\pi b$, which is what makes lowering profitable when $\pi$ is small.

\begin{proof}
We verify the hypotheses of \Cref{lem:highbound} with $\omega_1=2b$ and $\omega_2=(2H-1)b$; then $2\pi\omega_1+\pi^2\omega_2\le4\pi b+2\pi(\pi H)b\le6\pi b$, because $\pi H=\underline\gamma<1$.

First, for every base type $\bm t$ and every $a\in\supp F$,
\begin{equation}\label{eq:unused}
a\,x'_2(\bm t)-s'(\bm t)\ge a\,y'_2-\|\bm y'\|_1,\qquad a\,x'_1(\bm t)-s'(\bm t)\ge a\,y'_1-\|\bm y'\|_1 .
\end{equation}
Indeed, $\bm t$ does not gain by reporting $(1,1)\in B$ in the canonical menu, whose outcome there is $(\bm y',\|\bm y'\|_1)$; so $s'(\bm t)\le\|\bm y'\|_1-\bm t\cdot\bm\Delta$, where $\bm\Delta=\bm y'-\bm x'(\bm t)$ is the capacity left unused at $\bm t$. By~\eqref{eq:canon}, $\bm\Delta$ is $(0,0)$ on $B$, $y'_2(\beta_j,1)$ on $X_j$, $y'_1(1,\beta_j)$ on $Y_j$ and $\bm y'$ on $Z$; in all cases $\bm\Delta\ge\bm 0$ and $\Delta_1\ge\underline\beta\Delta_2$, as $\bm y'\in\cC_\kappa$. Hence, with $\ell_J\le t_1$ and $t_2-a\ge\ell_J-c_J$,
\[
a\,x'_2(\bm t)-s'(\bm t)\ge a\,y'_2-\|\bm y'\|_1+t_1\Delta_1+(t_2-a)\Delta_2\ge a\,y'_2-\|\bm y'\|_1+\big(\underline\beta\ell_J-(c_J-\ell_J)\big)\Delta_2\ge a\,y'_2-\|\bm y'\|_1,
\]
because $c_J-\ell_J<5d<\underline\beta/2\le\underline\beta\ell_J$ by (N1) and (N3). The second inequality is symmetric.

For the type $(H,a)$, let $\bm t$ be a base type with $x_1(\bm t)=y_1$. By the hypotheses and~\eqref{eq:unused},
\[
Hx_1(\bm t)+a\,x_2(\bm t)-s(\bm t)\ge Hy_1+a\,x'_2(\bm t)-s'(\bm t)-b\ge Hy_1+a\,y'_2-\|\bm y'\|_1-b\ge Hy_1+a\,y_2-\|\bm y\|_1-2b ,
\]
where the last step uses $a\,y'_2-\|\bm y'\|_1=a\,y_2-\|\bm y\|_1+(y_1-y'_1)+(1-a)(y_2-y'_2)$ and $a-1<3d$. The type $(a,H)$ is symmetric. For $(H,H)$, the type $\bm t''=(1,1)\in B$ satisfies $\bm x(\bm t'')\ge\bm x'(\bm t'')=\bm y'\ge\bm y-(b,b)$ and $s(\bm t'')\le\|\bm y'\|_1+b$, so $H\|\bm x(\bm t'')\|_1-s(\bm t'')\ge(H-1)\|\bm y'\|_1-b\ge(H-1)\|\bm y\|_1-(2H-1)b$.
\end{proof}

\subsubsection{Proof of Theorem~\ref*{thm:structure}}\label{app:structure}

Recall that $H_0=2/\sigma$, a rational of polynomial bit length. Let $H\ge H_0$ and $p\in[d/2,3d/2]$, and let $F=F_p$, $f=f_p$ and $\gamma=\gamma_p$, with $\underline\gamma$, $\pi$ and $F_H$ as in~\eqref{eq:FH}. Since $\sigma=\min\{1/(100K_1),d/100\}$, we have $\pi\le1/H\le\sigma/2$. We first record three consequences.
\begin{enumerate}[label=(H\arabic*),leftmargin=3em]
\item $(1-\pi)\rev^*+\pi c_J<\underline\gamma$. Indeed, $\rev^*\le1-d/2$ by \Cref{lem:revcurve}, $\underline\gamma\ge1-d/4-\sigma\ge1-0.26\,d$ by \Cref{lem:real}(a), and $\pi c_J\le2/H\le d/100$.
\item $(1-\pi)^2>K_1(2\sigma_H+6\pi)$. Indeed, $\sigma_H\le\sigma+2\pi(1-\gamma)\le\sigma+\pi d/2$, so the right side is at most $K_1(2\sigma+7\pi)\le\frac{11}2K_1\sigma\le0.06$, while $(1-\pi)^2\ge1-2\pi\ge0.99$.
\item $\pi(1-\pi)\rev^*<c_\alpha\sigma_H$, because $\pi(1-\pi)\rev^*\le\pi\le\sigma/2$, $c_\alpha\ge\frac9{10}$ and $\sigma_H\ge\sigma$.
\end{enumerate}
Moreover, $H\ge200/d\ge3c_J$. By (H1), \Cref{lem:separation,lem:canonext} apply. Let $M$ be an optimal mechanism for $(F_H;n,2)$.

\emph{The menus of the base restriction are canonical.} The lowest type $\ell_J\one$ of $\supp(F_H)^2$ has utility $0$ in every menu of $M$: otherwise, adding its utility to all payments of that menu would preserve DSIC and IR (\Cref{lem:norm}) and strictly increase the revenue, since every profile of reports has positive probability. Hence, by~\eqref{eq:bounded}, all payments of $M$ are nonnegative. For every buyer~$i$ and base reports $\bm v_{-i}$, let $\loss_i(\bm v_{-i})$ be the loss of the menu $M_i(\bm v_{-i})$ of the base restriction, and let $\bm y'_i(\bm v_{-i})$ be the capacities of the canonical menu to which \Cref{lem:stab}(b) lowers it. These canonical menus form a DSIC, ex-post IR mechanism $\tilde M$ for $F$, which satisfies supply because its allocations are at most those of the base restriction of $M$. By \Cref{lem:separation}, the revenue of $M$ is at most $(1-\pi)^{2n-2}\sum_i\E\,V_H(M_i(\bm v_{-i}))$ plus a term that depends only on $n$, $H$ and $\pi$, and some mechanism for $F_H$ with base restriction $\tilde M$ earns $(1-\pi)^{2n-2}\sum_i\E\,V_H(Q_{\bm y'_i(\bm v_{-i})})$ plus the same term. Since $M$ is optimal,
\[
\sum_i\E\Big[V_H\big(M_i(\bm v_{-i})\big)-V_H\big(Q_{\bm y'_i(\bm v_{-i})}\big)\Big]\ge0 .
\]
Fix $(i,\bm v_{-i})$, and write $Q=M_i(\bm v_{-i})$, $\bm y=\bm y_i(\bm v_{-i})$, $\bm y'=\bm y'_i(\bm v_{-i})$, $\loss=\loss_i(\bm v_{-i})$ and $b=K_1\loss$, so that $\Rev_F(Q)=\gamma\|\bm y\|_1-\loss$. By \Cref{lem:lowered,lem:canonext}, using $(1-\pi)^2\gamma+2\pi-\pi^2-\underline\gamma=\sigma_H$,
\begin{align*}
V_H(Q)-V_H(Q_{\bm y'})&\le\sigma_H(\|\bm y\|_1-\|\bm y'\|_1)-(1-\pi)^2\loss+\pi(1-\pi)\big[\Gamma_H(\one-\bm y)-\Gamma_H(\one-\bm y')\big]+6\pi b\\
&\le(2\sigma_H+6\pi)\,b-(1-\pi)^2\loss=-\big[(1-\pi)^2-K_1(2\sigma_H+6\pi)\big]\loss ,
\end{align*}
because $\|\bm y\|_1-\|\bm y'\|_1\le2b$, $\one-\bm y\le\one-\bm y'$ and $\Gamma_H$ is nondecreasing. By (H2), the coefficient of $\loss$ is negative. Since every $\bm v_{-i}$ has positive probability, every $\loss_i(\bm v_{-i})$ vanishes, so every menu of the base restriction is canonical, with capacities in $\cC_\kappa$, by \Cref{lem:real}(a). These capacities form a feasible table, because the allocations of $M$ satisfy supply. This proves~(a).

\emph{The capacity problem.} For a table $\bm w$ with values in $\cC_\kappa$, let $\cV(\bm w)=\sum_i\E\,V_H\big(Q_{\bm w_i(\bm v_{-i})}\big)$; by \Cref{lem:canonext},
\[
\cV(\bm w)=2n\underline\gamma+\sum_i\E\Big[\sigma_H\big\|\bm w_i(\bm v_{-i})\big\|_1+\pi(1-\pi)\,\Gamma_H\big(\one-\bm w_i(\bm v_{-i})\big)\Big].
\]
For a feasible table $\bm w$, \Cref{lem:separation} shows that the tight mechanism with the table $\bm w$ extends to a mechanism for $F_H$ with revenue $(1-\pi)^{2n-2}\cV(\bm w)+2H[1-(1-\pi)^n-n\pi(1-\pi)^{2n-2}]$; and by~(a), the same lemma shows that the revenue of $M$ equals this expression for the table $\bm y$ of $M$. Hence $\bm y$ maximizes $\cV$ over the feasible tables.

\emph{Averaging.} By \Cref{lem:canonext}, $V_H(Q_{\bm y})$ is a concave function of $\bm y\in\cC_\kappa$, and by \Cref{lem:narrowquant}(b) and \Cref{lem:frontierprops}(a), every feasible anonymous table $\bm w$ of a subproblem with $l$ buyers satisfies
\[
l\,\E\,V_H\big(Q_{\bm w^*_l(S)}\big)-l\,\E\,V_H\big(Q_{\bm w(S)}\big)\ge l\,\big(c_\alpha\sigma_H-\pi(1-\pi)\rev^*\big)\,\E\big\|\bm w(S)-\bm w^*_l(S)\big\|_1 ,
\]
whose right side is positive unless $\bm w=\bm w^*_l$, by (H3). By~\eqref{eq:blocks}, $\bm w^{\mathrm{ex}}$ is therefore the unique maximizer of $\cV$ over the feasible anonymous tables. By \Cref{lem:average}, $\bar{\bm y}$ is a feasible anonymous table with $\cV(\bar{\bm y})\ge\cV(\bm y)$. Hence $\cV(\bm w^{\mathrm{ex}})\ge\cV(\bar{\bm y})\ge\cV(\bm y)\ge\cV(\bm w^{\mathrm{ex}})$, so $\bar{\bm y}=\bm w^{\mathrm{ex}}$, which is~(b), and $\cV(\bar{\bm y})=\cV(\bm y)$. Since $\cY(\bar{\bm y})=\cY(\bm y)$, the latter means that $\sum_i\E\,\Gamma_H(\one-\bm y_i(\bm v_{-i}))=\sum_i\E\,\Gamma_H(\one-\bar{\bm y}(\mathsf m(\bm v_{-i})))$. As in the proof of \Cref{lem:average}, both sides are combinations, with the positive weights $n\Pr[\mathsf m]$, of the two sides of Jensen's inequality $\E_{l,S}\,\Gamma_H(\one-\bm y_i(\bm v_{-i}))\le\Gamma_H\big(\one-\E_{l,S}\,\bm y_i(\bm v_{-i})\big)$; so each of these inequalities holds with equality, which is~(c). Finally, the revenue of $M$ is $(1-\pi)^{2n-2}\cV(\bm w^{\mathrm{ex}})+2H[1-(1-\pi)^n-n\pi(1-\pi)^{2n-2}]$, which is~(d), because $(1-\pi)^{2n-2}\cdot2n\underline\gamma=2Hn\pi(1-\pi)^{2n-2}$.

\subsection{The scalar recursion}\label{app:scalar}

We first record two identities of the greedy recursion~\eqref{eq:greedy} in the narrow cone; for $\kappa=1$ they give~\eqref{eq:Dscalar}, and for $\kappa<1$ they are used in \Cref{app:D-imbalance}. For classes $X(a)$ and $Y(b)$ and $\bm y\in\R^2$, we have $C_{X(a)}(\bm y)_1=y_1-(1-a)y_2=a\,y_2+(y_1-y_2)$ and $C_{Y(b)}(\bm y)_2=b\,y_1-(y_1-y_2)$.

\begin{lemma}[Greedy identities]\label{lem:greedyid}
Let $l\ge1$ and $\kappa\in[\frac12,1]$, let all slopes be at most $\kappa$, and suppose that $l\chi\le\frac14$, where $\chi$ is the largest number $1-\kappa\beta$ over the slopes $\beta$ of the classes. Write $\bm w=\bm w^*_l$. For a context $S$ for $l$ buyers with $n_B=n_B(S)$, let $\omega_{\bm w}(S)=\bm w(S)_1-\bm w(S)_2$,
\[
L^X_{\bm w}(S)=\sum_{X(a)\in S}C_{X(a)}\big(\bm w(S-X(a))\big)_1,\qquad L^Y_{\bm w}(S)=\sum_{Y(b)\in S}C_{Y(b)}\big(\bm w(S-Y(b))\big)_2,
\]
$D_{\bm w}(S)=L^X_{\bm w}(S)-L^Y_{\bm w}(S)$ and $\hat D_{\bm w}(S)=n_B\,\omega_{\bm w}(S)+D_{\bm w}(S)$; thus $\hat D_{\bm w}(S)$ is the difference between the total allocations of the two items at the class profile consisting of $S$ and $n_B$ buyers of class $B$. Then:
\begin{enumerate}[label=(\alph*),nosep]
\item $n_B\,\bm w(S)_1=1-L^X_{\bm w}(S)-\pos{-\hat D_{\bm w}(S)}$ and $n_B\,\bm w(S)_2=1-L^Y_{\bm w}(S)-\pos{\hat D_{\bm w}(S)}$;
\item if $|S|\ge1$, then
\begin{equation}\label{eq:Dsigned}
\begin{aligned}
(n_B+1)D_{\bm w}(S)=\ &\mathsf d(S)+(n_B+1)\sum_{\tau\in S}\omega_{\bm w}(S-\tau)\\
&+\sum_{X(a),Y(b)\in S}(a+b-ab)\,\omega_{\bm w}\big(S-X(a)-Y(b)\big)\\
&-\sum_{X(a)\in S}a\pos{\hat D_{\bm w}\big(S-X(a)\big)}+\sum_{Y(b)\in S}b\pos{-\hat D_{\bm w}\big(S-Y(b)\big)} .
\end{aligned}
\end{equation}
\end{enumerate}
For $\kappa=1$, both coordinates of $\bm w$ equal the scalar table, $\omega_{\bm w}\equiv0$ and $\hat D_{\bm w}=D_{\bm w}=D$, and~\eqref{eq:Dsigned} is~\eqref{eq:Dscalar}.
\end{lemma}

\begin{proof}
(a) By~\eqref{eq:greedy}, $\bm w(S)=G_\kappa(U_1,U_2)$ with $U_1=(1-L^X_{\bm w}(S))/n_B$ and $U_2=(1-L^Y_{\bm w}(S))/n_B$, and $U_1,U_2\ge0$ by the proof of \Cref{lem:narrowquant}(a). If $U_1\le U_2/\kappa$, then $\bm w(S)_1=U_1$ and $\bm w(S)_2\le U_2$; otherwise $U_1>U_2/\kappa\ge U_2$, so $\bm w(S)_2=U_2$ and $\bm w(S)_1=U_2/\kappa<U_1$. So the total allocations $n_B\,\bm w(S)_1+L^X_{\bm w}(S)$ and $n_B\,\bm w(S)_2+L^Y_{\bm w}(S)$ of the two items are at most $1$, the larger one equals $1$, and their difference is $\hat D_{\bm w}(S)$. Hence $1-n_B\,\bm w(S)_1-L^X_{\bm w}(S)=\pos{-\hat D_{\bm w}(S)}$ and $1-n_B\,\bm w(S)_2-L^Y_{\bm w}(S)=\pos{\hat D_{\bm w}(S)}$, which is~(a).

(b) Every child $S-\tau$ has $n_B+1$ buyers of class $B$. By the formulas for $C_{X(a)}$ and $C_{Y(b)}$ above, and by~(a) at the children,
\begin{align*}
(n_B+1)D_{\bm w}(S)&=\sum_{X(a)\in S}a\,(n_B+1)\,\bm w\big(S-X(a)\big)_2-\sum_{Y(b)\in S}b\,(n_B+1)\,\bm w\big(S-Y(b)\big)_1\\
&\quad+(n_B+1)\sum_{\tau\in S}\omega_{\bm w}(S-\tau)\\
&=\sum_{X(a)\in S}a\Big[1-L^Y_{\bm w}\big(S-X(a)\big)-\pos{\hat D_{\bm w}\big(S-X(a)\big)}\Big]\\
&\quad-\sum_{Y(b)\in S}b\Big[1-L^X_{\bm w}\big(S-Y(b)\big)-\pos{-\hat D_{\bm w}\big(S-Y(b)\big)}\Big]+(n_B+1)\sum_{\tau\in S}\omega_{\bm w}(S-\tau).
\end{align*}
The constant terms give $\mathsf d(S)$. The loads of the children expand over the grandchildren $S'=S-X(a)-Y(b)$:
\begin{align*}
&-\sum_{X(a)\in S}a\,L^Y_{\bm w}\big(S-X(a)\big)+\sum_{Y(b)\in S}b\,L^X_{\bm w}\big(S-Y(b)\big)\\
&\qquad=\sum_{X(a),Y(b)\in S}\Big[-a\big(b\,\bm w(S')_1-\omega_{\bm w}(S')\big)+b\big(a\,\bm w(S')_2+\omega_{\bm w}(S')\big)\Big],
\end{align*}
and the bracket equals $(a+b-ab)\,\omega_{\bm w}(S')$. This is~\eqref{eq:Dsigned}. For $\kappa=1$, the two coordinates of $G_1$ are equal, so both coordinates of $\bm w$ equal the scalar table, $\omega_{\bm w}\equiv0$ and $\hat D_{\bm w}=D_{\bm w}=D$.
\end{proof}

\begin{proof}[Proof of \Cref{lem:Dsign}]
Write $W=W^{(l)}$. We first show by induction on $|S|$ that $\frac3{4n_B(S)}\le W(S)\le\frac1{n_B(S)}$. For $S=\emptyset$, $W(\emptyset)=1/l$. For $|S|\ge1$, every child $S-\tau$ satisfies $0<W(S-\tau)\le1$ by induction, so $0\le L^X(S),L^Y(S)\le(l-1)\bar a\le\frac14$, and~\eqref{eq:scalar} gives the bounds. For $\kappa=1$, the number $\chi$ of \Cref{lem:greedyid} is the largest coefficient, so $l\chi\le l\bar a\le\frac14$, and the lemma gives~\eqref{eq:Dscalar}.

We now prove the claims about $D$ by induction on $|S|$; for $S=\emptyset$, both $D$ and $\mathsf d$ vanish. Let $|S|\ge1$ and $n_B=n_B(S)$, and suppose that $\mathsf d(S)>0$. For a $Y$-occurrence, $\mathsf d(S-Y(b))=\mathsf d(S)+b>0$, so $D(S-Y(b))>0$ by induction and $\pos{-D(S-Y(b))}=0$. For an $X$-occurrence, $\mathsf d(S-X(a))=\mathsf d(S)-a\le\mathsf d(S)$, so by induction $\pos{D(S-X(a))}\le\pos{\mathsf d(S-X(a))}/(n_B+2)\le\mathsf d(S)/(n_B+2)$. Hence~\eqref{eq:Dscalar} gives
\[
\mathsf d(S)\ \ge\ (n_B+1)D(S)\ \ge\ \mathsf d(S)\Big(1-\frac{\sum_{X(a)\in S}a}{n_B+2}\Big)\ \ge\ \tfrac34\,\mathsf d(S),
\]
since $\sum_{X(a)\in S}a\le(l-1)\bar a\le1/4$. The case $\mathsf d(S)<0$ is symmetric (exchange the items). If $\mathsf d(S)=0$, then $\mathsf d(S-X(a))=-a\le0$ and $\mathsf d(S-Y(b))=b\ge0$, so $D(S-X(a))\le0\le D(S-Y(b))$ by induction; both sums in~\eqref{eq:Dscalar} vanish, and $D(S)=0$.
\end{proof}

\begin{lemma}[Comparison with the scalar table]\label{lem:Wclose}
Let $l\ge1$, $0\le\varepsilon\le\frac12$ and $\bar a\ge0$ with $l(\bar a+\varepsilon)\le\frac14$, and let $\kappa=1-\varepsilon$. Fix a context $S_{\max}$ for $l$ buyers, and give each occurrence $\tau$ in it two coefficients $a_\tau\in[0,\bar a]$ and $\tilde a_\tau\in[\varepsilon,\bar a]$; for a sub-multiset $S$, write $\|\tilde{\bm a}-\bm a\|_S=\sum_{\tau\in S}|\tilde a_\tau-a_\tau|$, where $\bm a=(a_\tau)$ and $\tilde{\bm a}=(\tilde a_\tau)$. On the sub-multisets of $S_{\max}$, let $W$ be the scalar table with $l$ buyers for the coefficients $a_\tau$, with load difference $D$, and let $\bm w$ be the greedy table~\eqref{eq:greedy} with $l$ buyers and $\kappa$ for the coefficients $\tilde a_\tau$, with loads $L^X_{\bm w},L^Y_{\bm w}$ and load difference $D_{\bm w}$ as in \Cref{lem:greedyid}. Then, for every $S\subseteq S_{\max}$ and $k=1,2$,
\[
|\bm w(S)_k-W(S)|\le\frac{(|S|+1)\,\varepsilon+\|\tilde{\bm a}-\bm a\|_S}{n_B(S)},\qquad |D_{\bm w}(S)-D(S)|\le\frac{5\big(|S|\varepsilon+\|\tilde{\bm a}-\bm a\|_S\big)}{4\,(n_B(S)+1)} .
\]
\end{lemma}

\begin{proof}
Every slope $1-\tilde a_\tau$ is at most $\kappa$, and $1-\kappa(1-\tilde a_\tau)\le\bar a+\varepsilon$; so the proof of \Cref{lem:narrowquant}(a), with $l$ in place of $n$, shows that the residuals of $\bm w$ satisfy $0\le U_1(S),U_2(S)\le1/n_B(S)$ and that $\bm w(S)\le(1/n_B(S),1/n_B(S))$. Also $0<W(S)\le1/n_B(S)$ by \Cref{lem:Dsign}. Let $m(S)=\min\{U_1(S),U_2(S)\}$. By~\eqref{eq:G}, $\bm w(S)=m(S)\one+\bm\varphi(S)$, where one coordinate of $\bm\varphi(S)$ is $0$ and the other is at most $(1-\kappa)\max_k\bm w(S)_k\le\varepsilon/n_B(S)$, because $\bm w(S)\in\cC_\kappa$.

Let $\tau\in S$ be an occurrence of an $X$-class, and write $\bm w(S-\tau)=m\one+\bm\varphi$. Then $C_\tau(\bm w(S-\tau))_1=\tilde a_\tau m+\varphi_1-(1-\tilde a_\tau)\varphi_2$ differs from $\tilde a_\tau m$ by at most $\varepsilon/(n_B(S)+1)$, while the corresponding term of the scalar load $L^X(S)$ is $a_\tau W(S-\tau)$, with $W(S-\tau)\le1/(n_B(S)+1)$; the same holds for $Y$-classes. Summing over $S$, with $A(S)=\sum_{\tau\in S}\tilde a_\tau\le\frac14$, we obtain
\begin{equation}\label{eq:loaddiff}
|L^X_{\bm w}(S)-L^X(S)|+|L^Y_{\bm w}(S)-L^Y(S)|\le\sum_{\tau\in S}\tilde a_\tau\,|m(S-\tau)-W(S-\tau)|+\frac{|S|\varepsilon+\|\tilde{\bm a}-\bm a\|_S}{n_B(S)+1} .
\end{equation}
We show that $|m(S)-W(S)|\le(|S|\varepsilon+\|\tilde{\bm a}-\bm a\|_S)/n_B(S)$ by induction on $|S|$. For $S=\emptyset$, both equal $1/l$. For $|S|\ge1$, put $e(S)=|S|\varepsilon+\|\tilde{\bm a}-\bm a\|_S$; the children satisfy $n_B(S-\tau)=n_B(S)+1$ and $e(S-\tau)\le e(S)$, so by induction the right side of~\eqref{eq:loaddiff} is at most $(A(S)+1)\,e(S)/(n_B(S)+1)\le5e(S)/(4(n_B(S)+1))$. Since $n_B(S)\,m(S)=1-\max\{L^X_{\bm w}(S),L^Y_{\bm w}(S)\}$ and $n_B(S)\,W(S)=1-\max\{L^X(S),L^Y(S)\}$, we get $n_B(S)\,|m(S)-W(S)|\le5e(S)/(4(n_B(S)+1))\le e(S)$. This completes the induction. The two bounds of the lemma follow: $|\bm w(S)_k-W(S)|\le|m(S)-W(S)|+\bm\varphi(S)_k$, and $|D_{\bm w}(S)-D(S)|$ is at most the left side of~\eqref{eq:loaddiff}.
\end{proof}

\subsection{Proof of Theorem C}\label{app:C}

\Cref{app:C-blocks,app:C-read,app:C-count} contain the proofs for \Cref{sec:C-blocks,sec:C-revenue}, and \Cref{app:D-compare,app:D-imbalance,app:D-reveal} those for \Cref{sec:C-compare,sec:C-imbalance,sec:C-reveal}.

\subsubsection{Groups of four numbers}\label{app:C-blocks}

We use the numbers of~\eqref{eq:blocks4}.

\begin{lemma}[Groups of four numbers]\label{lem:blocks}
The sets $I\subseteq[2k]$ with $|I|=k$ and $\sum_{j\in I}z_j=\hat T$ correspond bijectively to the solutions, and each of them has index sum $\sum_{j\in I}j=4r^2+r$. Every ordered $k$-tuple of the numbers $z_1,\dots,z_{2k}$ with sum $\hat T$, repetitions allowed, has distinct entries, whose indices form one of these sets; so each of these sets arises from exactly $k!$ such tuples. All these numbers have bit length polynomial in the input length.
\end{lemma}

\begin{proof}[Proof of \Cref{lem:blocks}]
Within group~$i$, consecutive numbers differ by at least $\bar wL^{2i-1}>0$, and $z_{i+1,1}-z_{i,4}\ge \bar w\big(L^{2i}-4L^{2i-1}-L^{2i-2}\big)-w_i>0$, because $L\ge9$ and $\bar w>w_i$; so the numbers increase in the lexicographic order of $(i,j)$. Consider an ordered $k$-tuple of them, let $d_{i,j}$ be the number of its entries equal to $z_{i,j}$, and put $D_i=\sum_jd_{i,j}$ and $E_i=\sum_jj\,d_{i,j}$. Its sum is
\[
\bar w\sum_{i=1}^r\big(D_iL^{2i-2}+E_iL^{2i-1}\big)+q,\qquad q=\sum_{i=1}^rd_{i,4}\,w_i\le k\max_iw_i<\bar w.
\]
If the sum is $\hat T$, then $q\equiv T$ modulo $\bar w$, and $0\le q,T<\bar w$ give $q=T$; hence $\sum_i(D_iL^{2i-2}+E_iL^{2i-1})=\sum_i(2L^{2i-2}+5L^{2i-1})$, and since $0\le D_i\le k<L$ and $0\le E_i\le4k<L$, uniqueness of base-$L$ expansions gives $D_i=2$ and $E_i=5$ for every $i$. The only two indices in $[4]$, repetitions allowed, with sum $5$ are $\{1,4\}$ and $\{2,3\}$; so every group contributes $z_{i,1}$ and $z_{i,4}$, or $z_{i,2}$ and $z_{i,3}$, and then $q=T$ says that the groups of the first kind form a solution; in particular, the entries of the tuple are distinct, and their indices form a set $I$ as in the lemma. Conversely, every such choice has sum $\hat T$. Hence these sets $I$ are exactly the index sets of the choices above, and $I$ corresponds to the solution formed by its groups of the first kind. Each set $I$ arises from exactly $k!$ orderings, and its index sum is $\sum_i(8i-3)=4r^2+r$, since $(4i-3)+4i=(4i-2)+(4i-1)=8i-3$. The numbers have $O(r\log r+\log \bar w)$ bits.
\end{proof}

We also note that $\hat T<\bar z:=\sum_jz_j$, because $\bar z-\hat T=\bar w\sum_i\big(2L^{2i-2}+5L^{2i-1}\big)+\sum_iw_i-T$ and $\bar w>T$.

\subsubsection{Reading the total capacity}\label{app:C-read}

\begin{lemma}[Approximate reading]\label{lem:Cread}
There is an algorithm that, given the instance and rationals $R\in[\hat T-1,\hat T+1]$ and $e\in(0,1]$, computes a number within $e$ of $\cW(R)$ in time polynomial in the input length and the bit lengths of $R$ and $e$, using comparisons ``$\OPT(F';n,2)\ge\vartheta$?'' for priors $F'$ with $2J+4$ support points, all in $(0,1]$. The same holds if, instead of the comparisons, approximations of the numbers $\OPT(F';n,2)$ are available within an additive error that the algorithm computes, of bit length polynomial in the same quantities.
\end{lemma}

The algorithm forms the prior $F_H$ of~\eqref{eq:FH} from $F=F_{R,\varepsilon}$, for explicit $\varepsilon$ and $H$, and divides all values by $H$. Two errors arise. The capacity problem of $F_{R,\varepsilon}$ satisfies $n\chi\le\frac14$, so $\Capac^*_n=n\,\E\|\bm w^*_l(S)\|_1$ for the random $l$ and $S$ of~\eqref{eq:WR} (\Cref{lem:narrow}); the greedy and the scalar recursions differ only in the projection and in the coefficient $\varepsilon$ of $X_0$ and $Y_0$, so \Cref{lem:Wclose} gives $|\Capac^*_n{}-\cW(R)|\le4n^2\varepsilon$ (\Cref{lem:Cclose}). And by \Cref{thm:structure}(d), the normalized revenue
\begin{equation}\label{eq:PiH}
\Pi_H=\frac{\OPT(F_H;n,2)-2H\big[1-(1-\pi)^n\big]}{\sigma\,(1-\pi)^{2n-2}}
\end{equation}
lies in $[\Capac^*_n,\Capac^*_n{}+3n/(\sigma H)]$, because $\sigma_H-\sigma\le\pi d/2$ and $\pi<1/H$, while $\Capac^*_n$ and the sum of the terms with $\Gamma_H$ are at most $2n$.

\begin{lemma}[Comparison with the scalar limit]\label{lem:Cclose}
For every rational $R\in[\hat T-1,\hat T+1]$ and $0<\varepsilon\le g/2$, the capacity problem of $F_{R,\varepsilon}$ satisfies $n\chi\le\frac14$, and $|\Capac^*_n{}-\cW(R)|\le4n^2\varepsilon$.
\end{lemma}

\begin{proof}
Let $F=F_{R,\varepsilon}$ and $\kappa=1-\varepsilon$. Every slope with $j\ge1$ has $1-\beta_j\le\nu u_{\max}=\frac1{2000n^2}$, and $\varepsilon\le\frac g2\le\frac1{8000n^2}$. Since $1-\kappa\beta\le(1-\beta)+\varepsilon$ for every slope $\beta$, and $1-\kappa^2\le2\varepsilon$, we get $\chi=\max_j(1-\kappa\beta_j)\le\frac1{1000n^2}$, so $n\chi\le\frac14$. By \Cref{lem:narrow},
\[
\Capac^*_n=n\,\E\big\|\bm w^*_l(S)\big\|_1,\qquad\text{while}\qquad \cW(R)=2n\,\E\big[W^{(l)}_R(S)\big],
\]
with the same random $l$ and $S$. For every context $S$ for $l$ buyers, apply \Cref{lem:Wclose} with $S_{\max}=S$, the same $\varepsilon$ and $\bar a=\frac1{2000n^2}$, to the coefficients of $W^{(l)}_R$ and of $F$: they agree except at the occurrences of $X_0$ and $Y_0$, which have coefficients $0$ and $\varepsilon$, and the other coefficients of $F$ are at least $\nu u_1=2g\ge\varepsilon$. Hence both coordinates of $\bm w^*_l(S)$ differ from $W^{(l)}_R(S)$ by at most $(2|S|+1)\varepsilon\le2n\varepsilon$, and $|\Capac^*_n{}-\cW(R)|\le4n^2\varepsilon$.
\end{proof}

\begin{proof}[Proof of \Cref{lem:Cread}]
Let $R$ and $e$ be given, and put $\varepsilon=\min\{g/2,\ e/(16n^2)\}$. The algorithm computes the prior $F=F_{R,\varepsilon}$ (\Cref{lem:real}(c)), the number $K_1$ of \Cref{lem:stab} and the number $\sigma$ of \Cref{sec:two-high} for $F$, the number $H=12n/(\sigma e)$, and the prior $F_H$ of~\eqref{eq:FH} with $\gamma=\gamma_d$ and $\pi=(\gamma-\sigma)/H$. Then $H\ge2/\sigma=H_0$ (\Cref{thm:structure}) and $H\ge4n$, since $\sigma,e\le1$, and all these numbers have bit length polynomial in the input length and the bit lengths of $R$ and $e$.

By \Cref{lem:Cclose}, \Cref{thm:structure} applies to $F_H$ with $p=d$. By part~(d) of that theorem and $\Capac^*_n=\cY(\bm w^{\mathrm{ex}})$ (\Cref{lem:narrow}), the number $\Pi_H$ of~\eqref{eq:PiH} satisfies
\[
\sigma\,\Pi_H=\sigma_H\Capac^*_n+\pi(1-\pi)\sum_{i=1}^n\E\,\Gamma_H\big(\one-\bm w^{\mathrm{ex}}_i(\bm v_{-i})\big).
\]
Here $0\le\sigma_H-\sigma=(1-\gamma)(2\pi-\pi^2)\le\pi d/2$ by \Cref{lem:real}(a), $0\le\Capac^*_n\le2n$, and $0\le\Gamma_H(\bm r)\le\rev^*\|\bm r\|_1\le2$ for $\bm r\in[0,1]^2$ (\Cref{lem:frontierprops}(a)). Since $\pi<1/H$, it follows that $0\le\Pi_H-\Capac^*_n{}\le(nd+2n)\pi/\sigma\le3n/(\sigma H)=e/4$. Together with \Cref{lem:Cclose}, $|\Pi_H-\cW(R)|\le e/4+4n^2\varepsilon\le e/2$.

Dividing all values and payments by $H$ maps the DSIC, ex-post IR mechanisms for $F_H$ bijectively onto those for the prior $\hat F_H$ obtained by dividing all values by $H$, and divides revenues by $H$; so $\OPT(F_H;n,2)=H\,\OPT(\hat F_H;n,2)$, where $\hat F_H$ has $2J+4$ support points, all in $(0,1]$, and $\OPT(\hat F_H;n,2)\in[0,2]$. Binary search with $O(\log(1/e'))$ comparisons finds $\OPT(\hat F_H;n,2)$ within $e'=\sigma e/(8H)$. Since $H\ge4n$, we have $(1-\pi)^{2n-2}\ge1-2n\pi\ge\frac12$, so the resulting value of $\Pi_H$ has error at most $2He'/\sigma=e/4$, and the total error is at most $\frac34e<e$. The same estimate applies to approximations of $\OPT(\hat F_H;n,2)$ within $e'$, whose bit length is polynomial.
\end{proof}

\subsubsection{Counting}\label{app:C-count}

\begin{proof}[Proof of \Cref{lem:Ctarget}]
Write $W=W^{(l)}_R$. The coefficients lie in $[0,\nu u_{\max}]$ and $n\nu u_{\max}\le\frac14$, so \Cref{lem:Dsign} applies for every $R\in[\hat T-1,\hat T+1]$. Give an $X$-occurrence the sign $+1$ and a $Y$-occurrence the sign $-1$. For a context $S$, let $S_{\mathrm d}$ and $S_{\mathrm o}$ be the multisets of its distinguished and of its ordinary occurrences, so that $|S_{\mathrm d}|+|S_{\mathrm o}|\le|S|\le k+1$; let $n_{\mathrm d}$ and $n_{\mathrm o}$ be the sums of their signs, and let $\Sigma_{\mathrm o}$ be the sum of the numbers $z_j$ of the ordinary occurrences, each multiplied by its sign. The occurrences of $X_0,Y_0$ have coefficient $0$. Then $\mathsf d(S)=\nu\,\mathsf g_R(S)$ with
\[
\mathsf g_R(S)=\Lambda(kn_{\mathrm d}+n_{\mathrm o})+n_{\mathrm d}R+\Sigma_{\mathrm o},\qquad |n_{\mathrm d}R+\Sigma_{\mathrm o}|\le|S_{\mathrm d}|(\hat T+1)+|S_{\mathrm o}|z_{2k}\le\Lambda-1 .
\]

(a) Suppose that $S$ is not a critical context. We show that the sign of $\mathsf g_R(S)$ does not depend on $R$. If $kn_{\mathrm d}+n_{\mathrm o}\ne0$, then $|\Lambda(kn_{\mathrm d}+n_{\mathrm o})|\ge\Lambda$, so $\mathsf g_R(S)$ has the sign of $kn_{\mathrm d}+n_{\mathrm o}$. If $n_{\mathrm d}=0$, then $\mathsf g_R(S)$ does not depend on $R$. Otherwise $n_{\mathrm d}\ne0$ and $n_{\mathrm o}=-kn_{\mathrm d}$, so $k|n_{\mathrm d}|=|n_{\mathrm o}|\le|S_{\mathrm o}|\le k+1-|S_{\mathrm d}|\le k+1-|n_{\mathrm d}|$; hence $|n_{\mathrm d}|=|S_{\mathrm d}|=1$ and $|n_{\mathrm o}|=|S_{\mathrm o}|=k$, and $S$ is a critical context. By \Cref{lem:Dsign}, the more heavily loaded item at $S$ is therefore the same for all $R$ (if $\mathsf g_R(S)$ vanishes identically, the two loads agree, and either can be taken): there is a set $\cX(S)$, the $X$-occurrences or the $Y$-occurrences of $S$, such that
\[
n_B(S)\,W(S)=1-\sum_{\tau\in\cX(S)}a_\tau\,W(S-\tau)\qquad(R\in[\hat T-1,\hat T+1]),
\]
where $a_\tau$ is the coefficient of $\tau$: $\nu u_J(R)$, which is affine in $R$, for a distinguished occurrence, and a constant otherwise. The children of $S$ are not critical contexts, which have the largest possible size $k+1$. By induction on $|S|$, starting from $W(\emptyset)=1/l$, $W(S)$ is a polynomial in $R$ of degree at most $|S|\le n-1$.

(b) Let $S=\{X_J,Y_{j_1},\dots,Y_{j_k}\}$, so that $l=n$ and $n_B(S)=1$. The $X$-child $\mathsf Y=S-X_J$ contains no distinguished class, so $W(\mathsf Y)$ does not depend on $R$, and $\mathsf d(\mathsf Y)<0$ gives $D(\mathsf Y)<0$ by \Cref{lem:Dsign}. Each $Y$-child $S-Y_{j_a}$ has $\mathsf d=\nu\big(\Lambda+R-\sum_{b\ne a}z_{j_b}\big)>0$, since $\sum_{b\ne a}z_{j_b}<kz_{2k}<\Lambda$ and $R>0$; so $D(S-Y_{j_a})>0$. Hence both sums in~\eqref{eq:Dscalar} vanish, and $2D(S)=\mathsf d(S)=\nu\big(R-\sum_az_{j_a}\big)$. By~\eqref{eq:scalar}, $W(S)=1-\max\{L^X(S),L^Y(S)\}=1-L^X(S)-\pos{-D(S)}$ with $L^X(S)=\nu u_J(R)\,W(\mathsf Y)$, which is the claim. The mirror image is symmetric.
\end{proof}

\begin{proof}[Proof of \Cref{prop:Ccount}]
For a function $\Xi$ on $[\hat T-1,\hat T+1]$ write $\Delta^n\Xi=\sum_{i=0}^n(-1)^{n-i}\binom ni\Xi(R_i)$; since $R_i=R_0+ih_R$, $\Delta^n$ annihilates every polynomial of degree at most $n-1$. By~\eqref{eq:WR},
\[
\cW(R)=2n\sum_{\bm o}\Pr(\bm o)\,W^{(l(\bm o))}_R\big(S(\bm o)\big),
\]
where $\bm o=(o_1,\dots,o_{n-1})$ ranges over the ordered profiles of classes of the opponents, $\Pr(\bm o)=\prod_{i=1}^{n-1}\Pr(o_i)$ does not depend on $R$, and $l(\bm o)$ and $S(\bm o)$ are the corresponding $l$ and $S$. By \Cref{lem:Ctarget}, $\Delta^n$ annihilates the terms of the profiles $\bm o$ for which $S(\bm o)$ is not a critical context and, if $S(\bm o)=\{X_J,Y_{j_1},\dots,Y_{j_k}\}$, all terms of $W^{(n)}_R(S(\bm o))$ except $-\frac\nu2\pos{\Sigma-R}$ with $\Sigma=\sum_az_{j_a}$. All sample points lie in $(\hat T-1,\hat T+1)$, and $\Sigma$ is an integer. If $\Sigma\ge \hat T+1$, then $\pos{\Sigma-R}=\Sigma-R$ at all sample points, and $\Delta^n$ annihilates it; if $\Sigma\le \hat T-1$, it vanishes at all of them; and if $\Sigma=\hat T$, it equals $h_R$ at $R_0=\hat T-h_R$ and $0$ at $R_1,\dots,R_n$, so its $n$-th difference is $(-1)^nh_R$. Hence such a profile contributes $2n\cdot(-1)^{n+1}\frac{\nu h_R}2\Pr(\bm o)$ if $\Sigma=\hat T$ and nothing otherwise, and the same holds for the mirror images. The profiles with a critical context and $\Sigma=\hat T$ are given by the orientation ($2$ choices), the position of the distinguished opponent ($k+1$ choices) and an ordered $k$-tuple of ordinary indices whose numbers $z_j$ sum to $\hat T$, of which there are $k!\,\cN$ by \Cref{lem:blocks}; each of them has probability $\mu d\theta^J\prod_a\mu d\theta^{j_a}=(\mu d)^{k+1}\theta^{J+4r^2+r}$, because the index sum of the tuple is $4r^2+r$. Summing gives the claim.
\end{proof}

\subsubsection{The hard comparison}\label{app:D-compare}

\begin{lemma}[Contexts with one $X$-class]\label{lem:Dheavy}
For every $I\subseteq[2k]$ and $q\in(0,\bar a]$:
\begin{enumerate}[label=(\alph*),nosep]
\item $(n-|I|)\,\phi_q(I)=\pos{\sum_{i\in I}b_i-q}-\sum_{i\in I}b_i\,\phi_q(I\setminus\{i\})$, with $\phi_q(\emptyset)=0$;
\item $W_q(I)=\dfrac{1-q\,W_\circ(I)-\phi_q(I)}{n-|I|-1}$.
\end{enumerate}
Consequently, as functions of $q\in(0,\bar a]$, $\phi_q(I)$ and $W_q(I)$ are continuous and piecewise affine, and their slopes can change only at the \emph{kinks} $q=\sum_{i\in I'}b_i$, $I'\subseteq I$.
\end{lemma}

\begin{proof}
(a) Let $S=\{X(q)\}\cup\mathsf Y_I$, so that $n_B(S)+1=n-|I|$ and $\mathsf d(S)=q-\sum_{i\in I}b_i$. In~\eqref{eq:Dscalar}, the $X$-child is $\mathsf Y_I$, with $D(\mathsf Y_I)=-\sum_{i\in I}b_iW_\circ(I\setminus\{i\})\le0$, so $\pos{D(\mathsf Y_I)}=0$; the $Y$-children are $\{X(q)\}\cup\mathsf Y_{I\setminus\{i\}}$. Hence $(n-|I|)D(S)=q-\sum_{i\in I}b_i+\sum_{i\in I}b_i\phi_q(I\setminus\{i\})$. If $\sum_{i\in I}b_i>q$, then $D(S)<0$ by \Cref{lem:Dsign} and $\phi_q(I)=-D(S)$, which is~(a). If $\sum_{i\in I}b_i\le q$, then $D(S)\ge0$, so $\phi_q(I)=0$, and for every $i\in I$, $\mathsf d(\{X(q)\}\cup\mathsf Y_{I\setminus\{i\}})=q-\sum_{i'\in I}b_{i'}+b_i\ge0$ gives $\phi_q(I\setminus\{i\})=0$ by \Cref{lem:Dsign}; both sides of~(a) vanish. For $I=\emptyset$, $D(\{X(q)\})=q/n>0$.

(b) At $S$, $L^X(S)=q\,W_\circ(I)$ and $\max\{L^X,L^Y\}=L^X+\pos{-D(S)}=qW_\circ(I)+\phi_q(I)$; insert this into~\eqref{eq:scalar} with $n_B(S)=n-|I|-1$.

Finally, by induction on $|I|$, (a) shows that $\phi_q(I)$ is continuous and piecewise affine in $q$, with slope changes only at the numbers $\sum_{i\in I'}b_i$, $I'\subseteq I$; by~(b), so is $W_q(I)$, since $W_\circ(I)$ does not depend on $q$.
\end{proof}

For the proof of \Cref{prop:Dhard}, let $z_1,\dots,z_{2k}$ and $\hat T$ be the numbers of \Cref{lem:blocks}, with $k=2r\ge2$ and $n=2k+4$, and put $\bar z=\sum_iz_i$ and $z^*=1+\bar z+\hat T$. By \Cref{lem:blocks}, $\cN=\#\mathcal I$ with
\[
\mathcal I=\Big\{I\subseteq[2k]:\ |I|=k,\ z(I)=\hat T\Big\},\qquad z(I)=\sum_{i\in I}z_i,
\]
and $\hat T<\bar z$ (\Cref{app:C-blocks}). For $0<\nu\le b_0/(100z^*)$ define
\begin{equation}\label{eq:Dinstance}
b_i(\nu)=b_0+\nu z_i,\qquad c(\nu)=kb_0+\nu \hat T,\qquad \bar c(\nu)=\|\bm b(\nu)\|_1-c(\nu),
\end{equation}
where $\bm b(\nu)=(b_1(\nu),\dots,b_{2k}(\nu))$, so that $\|\bm b(\nu)\|_1=2kb_0+\nu\bar z$ and $\bar c(\nu)=kb_0+\nu(\bar z-\hat T)$. For every $c'$ with $|c'-c(\nu)|\le\nu/10$, the data $(k,\bm b(\nu),c')$ are valid: $0\le b_i-b_0\le\nu\bar z\le b_0/100$, $|c'-kb_0|\le\nu \hat T+\nu/10\le\nu z^*\le b_0/100$, and $\big|\|\bm b(\nu)\|_1-c'-kb_0\big|\le\nu(\bar z-\hat T)+\nu/10\le b_0/100$. Put
\[
\xi_k=(-1)^{k+1}\frac{2\,k!}{(k+4)!},\qquad C_2=3k\,\xi_k\,b_0^k(b_0-2).
\]

\begin{lemma}[The kink at $c(\nu)$]\label{lem:Dkink}
Fix $\nu\in(0,b_0/(100z^*)]$ and $h_c\in(0,\nu/10]$, let $\varrho=\varrho_{\bm b(\nu)}$ be the function of~\eqref{eq:Jdef} for the coefficients $\bm b(\nu)$ (defined for every $X$-coefficient $q\in(0,\bar a]$), and write $c=c(\nu)$ and $\bar c=\bar c(\nu)$. For an $X$-coefficient $\hat q$, let $[\varrho']_{\hat q}=\varrho'(\hat q^+)-\varrho'(\hat q^-)$, and put $P(\nu)=\bar c\,[\varrho']_c+c\,[\varrho']_{\bar c}$. Then:
\begin{enumerate}[label=(\alph*),nosep]
\item $\varrho$ is $n$-Lipschitz, and affine on each of the intervals $[c-h_c,c]$, $[c,c+h_c]$, $[\bar c-h_c,\bar c]$ and $[\bar c,\bar c+h_c]$;
\item $\cJ_{\bm b(\nu)}(c+h_c)-2\cJ_{\bm b(\nu)}(c)+\cJ_{\bm b(\nu)}(c-h_c)$ differs from $h_cP(\nu)$ by at most $4nh_c^2$;
\item $|P(\nu)-C_2\cN|\le2n\nu\bar z$.
\end{enumerate}
\end{lemma}

\begin{proof}
(a) By \Cref{lem:Dheavy}(a) and induction on $|I|$, $\phi_q(I)$ is $\frac12$-Lipschitz in $q$, since $n-|I|\ge4$ and $\|\bm b(\nu)\|_1\le\frac14$; by \Cref{lem:Dheavy}(b), so is $W_q(I)$, since $W_\circ(I)\le1$ and $n-|I|-1\ge3$. Hence $\varrho$ is $(k+\frac34)$-Lipschitz, and $k+\frac34\le n$. By \Cref{lem:Dheavy}, applied to $I=[2k]$ and to $I=[2k]\setminus\{j\}$, the slope of $\varrho$ can change only at subset sums $b(I'):=\sum_{i\in I'}b_i(\nu)$, $I'\subseteq[2k]$. Now $b(I')-c=(|I'|-k)b_0+\nu\big(z(I')-\hat T\big)$. If $|I'|\ne k$, then $|b(I')-c|\ge b_0-\nu\bar z\ge b_0-b_0/100>h_c$ (as $0\le z(I'),\hat T\le\bar z$). If $|I'|=k$, then $b(I')-c\in\nu\Z$, so either $b(I')=c$, that is, $I'\in\mathcal I$, or $|b(I')-c|\ge\nu>h_c$. Hence $c$ is the only possible kink in $[c-h_c,c+h_c]$. Since $b(I')-\bar c=c-b([2k]\setminus I')$, the same argument applied to the complement shows that $\bar c$ is the only possible kink in $[\bar c-h_c,\bar c+h_c]$, and that $b(I')=\bar c$ if and only if $[2k]\setminus I'\in\mathcal I$.

(b) By~(a), $\varrho(c+h_c)-2\varrho(c)+\varrho(c-h_c)=h_c[\varrho']_c$, and similarly at $\bar c$. Since $\|\bm b(\nu)\|_1-c=\bar c$,
\begin{align*}
(\bar c-h_c)\varrho(c+h_c)-2\bar c\varrho(c)+(\bar c+h_c)\varrho(c-h_c)&=h_c\bar c\,[\varrho']_c-h_c\big(\varrho(c+h_c)-\varrho(c-h_c)\big),\\
(c+h_c)\varrho(\bar c-h_c)-2c\varrho(\bar c)+(c-h_c)\varrho(\bar c+h_c)&=h_cc\,[\varrho']_{\bar c}-h_c\big(\varrho(\bar c+h_c)-\varrho(\bar c-h_c)\big).
\end{align*}
The two left sides add up to $\cJ_{\bm b(\nu)}(c+h_c)-2\cJ_{\bm b(\nu)}(c)+\cJ_{\bm b(\nu)}(c-h_c)$, and by~(a) the two differences on the right are at most $2nh_c$ in absolute value.

(c) Fix $\hat q\in\{c,\bar c\}$, put $\hat z=\hat T$ if $\hat q=c$ and $\hat z=\bar z-\hat T$ if $\hat q=\bar c$, and let $G_I$ be the slope jump of $\phi_q(I)$ at $q=\hat q$. The function $\pos{\sum_{i\in I}b_i(\nu)-q}$ has slope jump $1$ at $q=\sum_{i\in I}b_i(\nu)$ and none elsewhere, so \Cref{lem:Dheavy}(a) and part~(a) above give
\[
(n-|I|)\,G_I=\ind{|I|=k,\ z(I)=\hat z}-\sum_{i\in I}b_i(\nu)\,G_{I\setminus\{i\}},\qquad G_\emptyset=0.
\]
By \Cref{lem:Dheavy}(b), the slope jump of $W_q(I)$ at $\hat q$ is $-G_I/(n-|I|-1)$; since $n-2k-1=3$ and $n-(2k-1)-1=4$, \eqref{eq:Jdef} gives
\[
[\varrho']_{\hat q}=-\frac12\,G_{[2k]}-\frac14\sum_{j=1}^{2k}G_{[2k]\setminus\{j\}} .
\]
Now let $G^0_I$ be defined by the same recursion with every coefficient $b_i(\nu)$ replaced by $b_0$ and the same sources. Let $N_I$ be the number of sets $I'\subseteq I$ with $|I'|=k$ and $z(I')=\hat z$. Each of them lies in $I\setminus\{i\}$ for exactly $|I|-k$ indices $i\in I$, so $\sum_{i\in I}N_{I\setminus\{i\}}=(|I|-k)N_I$, and induction on $|I|$ gives $G^0_I=0$ for $|I|<k$ and
\[
G^0_I=(-1)^{|I|-k}\,\frac{(|I|-k)!\,(n-|I|-1)!}{(n-k)!}\,b_0^{|I|-k}N_I\qquad(|I|\ge k).
\]
Both for $\hat q=c$ and for $\hat q=\bar c$, we have $N_{[2k]}=\cN$ and $\sum_jN_{[2k]\setminus\{j\}}=k\cN$, because the complements of the members of $\mathcal I$ are the sets $I'$ with $|I'|=k$ and $z(I')=\bar z-\hat T$. With $n-k=k+4$, it follows that
\[
-\frac12\,G^0_{[2k]}-\frac14\sum_jG^0_{[2k]\setminus\{j\}}=(-1)^{k+1}\frac{3\,k!}{(k+4)!}\,b_0^{k}\cN+(-1)^{k}\frac{6\,k!}{(k+4)!}\,b_0^{k-1}\cN=\frac32\,\xi_k\,b_0^{k-1}(b_0-2)\,\cN ,
\]
and $kb_0$ times the sum of this value over $\hat q\in\{c,\bar c\}$ is $C_2\cN$.

It remains to bound the effect of the coefficients. Since $n-|I|\ge4$ and $\|\bm b(\nu)\|_1\le1$, induction gives $|G_I|,|G^0_I|\le1$. The differences $E_I=|G_I-G^0_I|$ satisfy $4E_I\le\sum_{i\in I}\big(|b_i(\nu)-b_0|\,|G_{I\setminus\{i\}}|+b_0E_{I\setminus\{i\}}\big)\le\nu\bar z+2kb_0\max_iE_{I\setminus\{i\}}$, so $E_I\le\nu\bar z$ by induction. Hence $|[\varrho']_{\hat q}|\le\frac{k+1}2\le\frac n2$, and $[\varrho']_{\hat q}$ differs from its reference value by at most $\frac n2\nu\bar z$. Since $|c-kb_0|=\nu \hat T\le\nu\bar z$ and $|\bar c-kb_0|=\nu(\bar z-\hat T)\le\nu\bar z$, it follows that $|P(\nu)-C_2\cN|\le2\big(\nu\bar z\cdot\frac n2+kb_0\cdot\frac n2\nu\bar z\big)\le2n\nu\bar z$.
\end{proof}

\begin{proof}[Proof of \Cref{prop:Dhard}]
Let an instance of the counting version of \textsc{SubsetSum} with $r\ge1$ integers be given, and compute $z_1,\dots,z_{2k}$, $\hat T$, $\bar z$, $z^*$, $b_0$ and $C_2$ as above. Put
\[
\nu=\min\Big\{\frac{b_0}{100z^*},\ \frac{|C_2|}{16nz^*}\Big\},\qquad h_c=\min\Big\{\frac\nu{10},\ \frac{|C_2|}{32n}\Big\},\qquad \tau=\frac{|C_2|\,h_c}{48};
\]
since $\log(1/|C_2|)=O(k\log n)$, these numbers have polynomial bit length, and $\tau\le1$. By \Cref{lem:Dkink}, with $c=c(\nu)$ and $\cJ=\cJ_{\bm b(\nu)}$,
\[
\big|\cJ(c+h_c)-2\cJ(c)+\cJ(c-h_c)-h_cC_2\cN\big|\le4nh_c^2+2n\nu\bar z\,h_c\le\frac{|C_2|\,h_c}4 .
\]
For each $c'\in\{c-h_c,c,c+h_c\}$, the data $(k,\bm b(\nu),c')$ are valid, and $|\cJ(c')|<n$. We approximate $\cJ(c')$ by binary search. Start with the interval $[-n,n]$, which contains $\cJ(c')$. Given an interval $[\vartheta^-,\vartheta^+]$ that contains $\cJ(c')$, ask the comparator at $\vartheta=(\vartheta^-+\vartheta^+)/2$. If it answers that $\cJ(c')>\vartheta$, then $\cJ(c')>\vartheta-\tau$, and we replace $\vartheta^-$ by $\max\{\vartheta^-,\vartheta-\tau\}$; otherwise $\cJ(c')<\vartheta+\tau$, and we replace $\vartheta^+$ by $\min\{\vartheta^+,\vartheta+\tau\}$. The new interval contains $\cJ(c')$, and its length is at most half the old length plus $\tau$. After $\lceil\log_2(2n/\tau)\rceil$ steps the length is at most $3\tau$, so its midpoint $\tilde\cJ(c')$ is within $\frac32\tau$ of $\cJ(c')$. Hence
\[
\Big|\tilde\cJ(c+h_c)-2\tilde\cJ(c)+\tilde\cJ(c-h_c)-h_cC_2\cN\Big|\le6\tau+\frac{|C_2|\,h_c}4=\frac38\,|C_2|\,h_c ,
\]
and $\cN$ is the integer nearest to $\big(\tilde\cJ(c+h_c)-2\tilde\cJ(c)+\tilde\cJ(c-h_c)\big)/(h_cC_2)$. All computations are exact rational arithmetic on numbers of polynomial bit length, and the comparator is called $3\lceil\log_2(2n/\tau)\rceil$ times.
\end{proof}

\subsubsection{The capacity imbalance}\label{app:D-imbalance}

Let $\varepsilon\in[0,\bar a]$, $\kappa=1-\varepsilon$ and $\bm w_\varepsilon$ be as in \Cref{sec:C-compare}. Every class has $1-\kappa\beta\le\bar a+\varepsilon\le2\bar a$ and $2n\bar a\le1/4$. The proof of \Cref{lem:narrowquant}(a) uses only the recursion~\eqref{eq:greedy} and the bound~\eqref{eq:chi}, which holds because every slope $1-a$ is at most $\kappa$; so it shows that for every context $S$ with $n_B=n_B(S)$,
\begin{equation}\label{eq:Dbounds}
\frac{3}{4n_B}\le U^\varepsilon_1(S),U^\varepsilon_2(S)\le\frac1{n_B},\qquad \bm w_\varepsilon(S)=G_\kappa\big(U^\varepsilon_1(S),U^\varepsilon_2(S)\big)\le\Big(\frac1{n_B},\frac1{n_B}\Big),
\end{equation}
and \Cref{lem:greedyid} applies to $\bm w_\varepsilon$. We write $L^X_\varepsilon$, $L^Y_\varepsilon$, $D_\varepsilon$, $\omega_\varepsilon$ and $\hat D_\varepsilon$ for the quantities $L^X_{\bm w}$, $L^Y_{\bm w}$, $D_{\bm w}$, $\omega_{\bm w}$ and $\hat D_{\bm w}$ of that lemma for $\bm w=\bm w_\varepsilon$, and $L^X$, $L^Y$ and $D$ for the scalar ones. The branch of the projection determines $\omega_\varepsilon$ and the sign of $\hat D_\varepsilon$: if $U^\varepsilon_1(S)>U^\varepsilon_2(S)/\kappa$, then $\bm w_\varepsilon(S)=(U^\varepsilon_2(S)/\kappa,U^\varepsilon_2(S))$, so $\omega_\varepsilon(S)=\varepsilon\,\bm w_\varepsilon(S)_1$, and by the proof of \Cref{lem:greedyid}(a) the total allocation of item~$2$ is $1$, so $\hat D_\varepsilon(S)\le0$; if $U^\varepsilon_1(S)<\kappa U^\varepsilon_2(S)$, then $\bm w_\varepsilon(S)=(U^\varepsilon_1(S),U^\varepsilon_1(S)/\kappa)$, $\omega_\varepsilon(S)=-\varepsilon\,\bm w_\varepsilon(S)_2$ and $\hat D_\varepsilon(S)\ge0$.

Sub-multisets of $S^*$ and of $S^*_0$ correspond occurrence by occurrence, the occurrence $X(\bar c_\varepsilon)$ corresponding to $X(\bar c)$; this is a correspondence of occurrences, which remains well defined if $c$ happens to equal $\bar c_\varepsilon$ or $\bar c$. For a sub-multiset $S$ of $S^*$, let $S_0$ be the corresponding sub-multiset of $S^*_0$; we write $W(S_0)$ and $D(S_0)$ for the scalar quantities at $S_0$, and $\mathsf d(S)$ and $\mathsf d(S_0)$ for the signed coefficient sums of $S$ and of $S_0$.

\begin{proof}[Proof of \Cref{prop:Dimbalance}]
We first compare the perturbed table with the scalar table. Apply \Cref{lem:Wclose} with $l=n$ and $S_{\max}=S^*$, the coefficients of $S^*_0$ in the role of $\bm a$ and those of $S^*$ in the role of $\tilde{\bm a}$; they differ only at the occurrence $X(\bar c_\varepsilon)$, by $2\varepsilon|\vartheta|\le4n\varepsilon$, and all of them lie in $[b_0,\bar a]$. As $|S|\le n-2$ for every sub-multiset $S$ of $S^*$,
\begin{equation}\label{eq:DDclose}
\big\|\bm w_\varepsilon(S)-W(S_0)\one\big\|_\infty\le\frac{5n\varepsilon}{n_B(S)},\qquad |D_\varepsilon(S)-D(S_0)|\le\frac{7n\varepsilon}{n_B(S)+1}.
\end{equation}

Next we determine the branches of the recursion at the children and grandchildren of $S^*$. The signed coefficient sums of the counterparts of the contexts entering~\eqref{eq:Dsigned} at $S^*$ are: $0$ at $S^*_0$, since $c+\bar c=\|\bm b\|_1$; $-c$ and $-\bar c$ at the $X$-children $S^*_0-X(c)$ and $S^*_0-X(\bar c)$; $b_j$ at the $Y$-children $S^*_0-Y(b_j)$; and $b_j-c$ and $b_j-\bar c$ at the mixed grandchildren $S^*_0-X(c)-Y(b_j)$ and $S^*_0-X(\bar c)-Y(b_j)$. By~\eqref{eq:Dvalid}, $c,\bar c\ge\frac32b_0$ and $b_j\in[b_0,1.01b_0]$, so every one of these signed sums other than $\mathsf d(S^*_0)$ has absolute value at least $0.49\,b_0$; they are negative at the $X$-children and the mixed grandchildren and positive at the $Y$-children. By \Cref{lem:Dsign}, the corresponding $|D(S_0)|$ are at least $\frac34\cdot0.49\,b_0/(n_B(S)+1)\ge0.36\,b_0/(n_B(S)+1)$, with the same signs, and $D(S^*_0)=0$.

Let $S$ be an $X$-child or a mixed grandchild of $S^*$, with $n_B=n_B(S)\in\{3,4\}$. By~\eqref{eq:DDclose} and $\varepsilon\le b_0/(100n)$, $-D_\varepsilon(S)\ge(0.36\,b_0-7n\varepsilon)/(n_B+1)\ge0.29\,b_0/5>0.05\,b_0$. Since $U^\varepsilon_1(S)-U^\varepsilon_2(S)=-D_\varepsilon(S)/n_B$ and $U^\varepsilon_2(S)\le1/n_B$,
\[
U^\varepsilon_1(S)-\frac{U^\varepsilon_2(S)}\kappa=\frac{-D_\varepsilon(S)}{n_B}-U^\varepsilon_2(S)\frac{\varepsilon}{1-\varepsilon}\ge\frac{0.05\,b_0-2\varepsilon}{n_B}>0 .
\]
So $\omega_\varepsilon(S)=\varepsilon\,\bm w_\varepsilon(S)_1$ and $\hat D_\varepsilon(S)\le0$. Let $S$ be a $Y$-child ($n_B=3$). Then $D_\varepsilon(S)>0.05\,b_0$ and
\[
\kappa U^\varepsilon_2(S)-U^\varepsilon_1(S)=\frac{D_\varepsilon(S)}{n_B}-\varepsilon U^\varepsilon_2(S)\ge\frac{0.05\,b_0-\varepsilon}{n_B}>0,
\]
so $\omega_\varepsilon(S)=-\varepsilon\,\bm w_\varepsilon(S)_2$ and $\hat D_\varepsilon(S)\ge0$. Thus $\omega_\varepsilon(S)/\varepsilon=\bm w_\varepsilon(S)_1$ at the $X$-children and the mixed grandchildren, and $\omega_\varepsilon(S)/\varepsilon=-\bm w_\varepsilon(S)_2$ at the $Y$-children.

We now compute the first-order term. Apply~\eqref{eq:Dsigned} at $S^*$ ($n_B=2$). By the previous paragraph the two sums of positive parts vanish, and $\mathsf d(S^*)=c+\bar c_\varepsilon-\|\bm b\|_1=-2\varepsilon\vartheta$. Define the scalar quantity
\begin{equation}\label{eq:D1}
\begin{aligned}
3D^{(1)}:={}&3\big[W(S^*_0-X(c))+W(S^*_0-X(\bar c))\big]-3\sum_{j=1}^{2k}W(S^*_0-Y(b_j))\\
&+\sum_{X(a),Y(b)\in S^*_0}(a+b-ab)\,W\big(S^*_0-X(a)-Y(b)\big).
\end{aligned}
\end{equation}
We compare $3D_\varepsilon(S^*)/\varepsilon$ with $-2\vartheta+3D^{(1)}$ term by term. By the previous paragraph and~\eqref{eq:DDclose}, the $2k+2$ children, with $n_B=3$, contribute at most $3\cdot5n\varepsilon/3$ each; the $4k$ mixed grandchildren, with $n_B=4$ and coefficients $0<\tilde a+b-\tilde ab\le1$, contribute at most $5n\varepsilon/4$ each; and the change of the coefficient of the $2k$ pairs that contain $X(\bar c_\varepsilon)$ contributes at most $|\tilde a-a|\,W\le4n\varepsilon\cdot\frac14$ each, as $W\le\frac14$ at the mixed grandchildren. Since $2k+2\le n$,
\[
\Big|\frac{3D_\varepsilon(S^*)}{\varepsilon}-\big(-2\vartheta+3D^{(1)}\big)\Big|\le5n(2k+2)\varepsilon+5nk\varepsilon+2nk\varepsilon\le9n^2\varepsilon .
\]
Since $n_B(S^*)=2$, we have $U^\varepsilon_1(S^*)-U^\varepsilon_2(S^*)=\big(L^Y_\varepsilon(S^*)-L^X_\varepsilon(S^*)\big)/2=-D_\varepsilon(S^*)/2$, so
\[
\Big|\frac{U^\varepsilon_1(S^*)-U^\varepsilon_2(S^*)}\varepsilon-\frac{2\vartheta-3D^{(1)}}6\Big|\le\frac32n^2\varepsilon .
\]

It remains to show that $3D^{(1)}=2\cJ(c)$. Write $W_c(I),W_{\bar c}(I)$ for the quantities $W_q(I)$ of \Cref{sec:C-compare} with $q=c$ and $q=\bar c$, and $W_{c\bar c}(I)=W(\{X(c),X(\bar c)\}\cup\mathsf Y_I)$. The contexts entering~\eqref{eq:D1} are $S^*_0-X(c)=\{X(\bar c)\}\cup\mathsf Y_{[2k]}$, $S^*_0-X(\bar c)=\{X(c)\}\cup\mathsf Y_{[2k]}$, $S^*_0-Y(b_j)=\{X(c),X(\bar c)\}\cup\mathsf Y_{[2k]\setminus\{j\}}$, $S^*_0-X(c)-Y(b_j)=\{X(\bar c)\}\cup\mathsf Y_{[2k]\setminus\{j\}}$ and $S^*_0-X(\bar c)-Y(b_j)=\{X(c)\}\cup\mathsf Y_{[2k]\setminus\{j\}}$. By \Cref{lem:Dsign} and~\eqref{eq:scalar}, the larger load is the $X$-load at $\{X(c),X(\bar c)\}\cup\mathsf Y_{[2k]\setminus\{j\}}$ ($\mathsf d=b_j>0$, three buyers of class $B$) and the $Y$-load at $\{X(q)\}\cup\mathsf Y_{[2k]}$ ($q\in\{c,\bar c\}$, $\mathsf d=q-\|\bm b\|_1<0$, three buyers of class $B$):
\[
3W_{c\bar c}([2k]\setminus\{j\})=1-c\,W_{\bar c}([2k]\setminus\{j\})-\bar c\,W_c([2k]\setminus\{j\}),\qquad 3W_q([2k])=1-\sum_jb_jW_q([2k]\setminus\{j\}).
\]
Substituting the first identity into~\eqref{eq:D1} gives
\[
3D^{(1)}=3\big[W_c([2k])+W_{\bar c}([2k])\big]-2k+\sum_j\Big[(2c+b_j-cb_j)W_{\bar c}([2k]\setminus\{j\})+(2\bar c+b_j-\bar cb_j)W_c([2k]\setminus\{j\})\Big].
\]
By the second identity, $\sum_jb_jW_q([2k]\setminus\{j\})=1-3W_q([2k])$ for $q\in\{c,\bar c\}$, so
\begin{align*}
3D^{(1)}&=2-2k-\|\bm b\|_1+2c\Big[\tfrac32W_{\bar c}([2k])+\sum_jW_{\bar c}([2k]\setminus\{j\})\Big]+2\bar c\Big[\tfrac32W_c([2k])+\sum_jW_c([2k]\setminus\{j\})\Big]\\
&=2\cJ(c),
\end{align*}
and the displayed estimate of the proposition follows, as $2n^2\varepsilon\le\tau/50$ by~\eqref{eq:Deps}.

Suppose that $\vartheta\ge\cJ(c)+\tau$, and write $U_1=U^\varepsilon_1(S^*)$ and $U_2=U^\varepsilon_2(S^*)$. Then $(U_1-U_2)/\varepsilon\ge\tau/3-\tau/50>\tau/4$, so $U_1-U_2>\eta$. If $U_1\le U_2/\kappa$, then $\bm w_\varepsilon(S^*)=(U_1,U_2)$ and $\bm w_\varepsilon(S^*)_1-\bm w_\varepsilon(S^*)_2=U_1-U_2>\eta$. Otherwise $\bm w_\varepsilon(S^*)=(U_2/\kappa,U_2)$, and $\bm w_\varepsilon(S^*)_1-\bm w_\varepsilon(S^*)_2=U_2(\kappa^{-1}-1)\ge\varepsilon U_2\ge\frac38\varepsilon\ge\eta$ by~\eqref{eq:Dbounds} and $\tau\le1$. The case $\vartheta\le\cJ(c)-\tau$ is symmetric.
\end{proof}

\subsubsection{Revealing the comparison}\label{app:D-reveal}

We first verify the claims of \Cref{sec:C-reveal} about the priors. Fix valid data, a threshold $\vartheta$ with $|\vartheta|\le2n$ and $\tau\in(0,1]$. The coefficients $c$, $\bar c_\varepsilon$ and $b_1,\dots,b_{2k}$ lie in $[b_0,nb_0]$, so the slopes $\beta_1>\dots>\beta_J$ lie in $[1-nb_0,1-b_0]$. The hypotheses of \Cref{lem:real} hold for $\kappa=1-\varepsilon$, $\underline\beta=\frac12$ and $g=\frac{b_0}2$: $\kappa\ge\frac12$, $\beta_J\ge1-nb_0\ge\frac12$, and $\beta_1\le1-b_0\le\kappa-g$ because $\varepsilon\le b_0/2$. In the indexing of \Cref{lem:real}, the classes $X_j,Y_j$ with $j\ge1$ have coefficient $1-\beta_j\in\{c,\bar c_\varepsilon,b_1,\dots,b_{2k}\}$, and $X_0,Y_0$ have the coefficient $\varepsilon$, which differs from all others; so $S^*$ is a context of $n-2$ single-item classes. Since $1-\kappa\beta_j\le(1-\beta_j)+\varepsilon\le nb_0+\varepsilon$ for $j\ge1$ and $1-\kappa^2\le2\varepsilon$, we have $n\chi\le\frac14$. The greedy recursion~\eqref{eq:greedy} at a sub-multiset of $S^*$ involves only the classes in $S^*$, so $\bm w^*_n$ coincides with $\bm w_\varepsilon$ on the sub-multisets of $S^*$.

Next, \Cref{lem:revcurve} and $\ptie=d/(1-\theta(1-d))$ give $\rev(\ell_0)-\rev(1)=d(p-\ptie)/\ptie$ for every $p\in[d/2,3d/2]$. For $p=p(H)$ as in~\eqref{eq:pH},
\[
p(H)-\ptie=\frac{\ptie(1-\ptie)}{4H-4+\ptie},\qquad 1-p(H)=\frac{4(H-1)(1-\ptie)}{4H-4+\ptie},
\]
so $\rev(\ell_0)-\rev(1)=d(1-p(H))/(4(H-1))$, which is~\eqref{eq:pH}. For $H\ge2$, moreover, $0<p(H)-\ptie\le\ptie/(2H)$, so $d<\ptie<p(H)\le\frac54\ptie<\frac32d$, since $\ptie<d(1+2\theta)$.

\begin{lemma}[Two prices]\label{lem:hinges}
For every rational $H\ge H_0$, with $F_H$ formed from $F_{p(H)}$, there are rationals $e_1,e_2>0$ such that, with $\zeta=(H-\ell_0)/(H-1)>1$ and $A_H=(1+\zeta)e_1+e_2$,
\[
\Gamma_H(\bm r)=\min\big\{A_Hr_1,\ A_Hr_2,\ (e_1+e_2)\,r_1+e_1r_2,\ e_1r_1+(e_1+e_2)\,r_2\big\}\qquad\big(\bm r=(r_1,r_2)\ge\bm 0\big).
\]
In particular, $\Gamma_H(\bm r)=A_Hr_1$ if and only if $r_2\ge\zeta r_1$.
\end{lemma}

\begin{proof}
Let $H\ge H_0$, and write $p=p(H)$, $R_U=\rev(1)$, $R_L=\rev(\ell_0)$ and $u=H-1$, so that $H-\ell_0=u+d=\zeta u$. Put $m_0=R_U/u$ and $m_1=(R_L-R_U)/d$; by~\eqref{eq:pH}, $m_1=m_0/4$.

We first show that
\[
\cF_r(T)=\min\big\{m_0T,\ R_Ur+m_1(T-ur),\ R_Lr\big\}\qquad(r,T\ge0).
\]
The three functions on the right are linear in $(r,T)$, and their minimum is $m_0T$ for $T\le ur$, $R_Ur+m_1(T-ur)$ for $ur\le T\le\zeta ur$, and $R_Lr$ for $T\ge\zeta ur$; mixtures of the prices $1$ and $\ell_0$ attain these values. Conversely, for $\lambda$ feasible in~\eqref{eq:frontier}, it suffices to show that $\rev(z)\le\min\{m_0(H-z),\,R_U+m_1(H-z-u),\,R_L\}$ for every $z\in\supp F$, since summing these bounds with the weights $\lambda_z$ and using $\sum_z\lambda_z\le r$, $\sum_z\lambda_z(H-z)\le T$ and $R_U-m_1u>0$ gives the claim. As a function of $H-z$, this minimum equals $m_0(H-z)$ for $H-z\le u$ and $R_L$ for $H-z\ge\zeta u$. For $z=1$ and $z=\ell_0$, equality holds in the bound. For $z=\ell_j$ with $j\ge1$, we have $H-z>\zeta u$, and $\rev(\ell_j)<R_L$ by \Cref{lem:revcurve}. For $z=c_j$, we have $H-z<u$, and $\rev(c_j)\le R_U-q_0/2$ by \Cref{lem:revcurve}, so $m_0(H-c_j)-\rev(c_j)\ge\big(q_0u/2-3dR_U\big)/u\ge0$, because $u\ge1$, $R_U\le1$ and $6d\le q_0$ by (N2).

Let $\hat\pi=\pi/(1-\pi)$. We have $m_0<\hat\pi<m_0+m_1=\frac54m_0$. Indeed, $\hat\pi/m_0=\underline\gamma(H-1)/\big((H-\underline\gamma)(1-p)\big)$ is less than $\underline\gamma/(1-p)\le1/(1-2d)<\frac54$; and it exceeds $1$ if and only if $H(\underline\gamma-1+p)>\underline\gamma p$, which holds because $\underline\gamma-1+p\ge p-0.26\,d\ge p/2$ (as $\underline\gamma\ge1-\frac d4-\sigma\ge1-0.26\,d$ by \Cref{lem:real}(a) and $\sigma\le\frac d{100}$, and $p>\ptie\ge d$) and $H\ge2$.

By symmetry let $r_1\le r_2$. The function $\cF_{r_1}(T)+\cF_{r_2}(T)-\hat\pi T$ is piecewise linear in $T$. If $r_2\ge\zeta r_1$, that is, $\zeta ur_1\le ur_2$, its slopes on the consecutive intervals between $0$, $ur_1$, $\zeta ur_1$, $ur_2$ and $\zeta ur_2$ are $2m_0-\hat\pi>0$, $m_0+m_1-\hat\pi>0$, $m_0-\hat\pi<0$, $m_1-\hat\pi<0$, and then $-\hat\pi$; so its maximum is attained at $T=\zeta ur_1$, and $\Gamma_H(\bm r)=R_Lr_1+(m_0-\hat\pi)\zeta ur_1$. If $r_2<\zeta r_1$, the slopes on the intervals between $0$, $ur_1$, $ur_2$, $\zeta ur_1$ and $\zeta ur_2$ are $2m_0-\hat\pi>0$, $m_0+m_1-\hat\pi>0$, $2m_1-\hat\pi<0$, $m_1-\hat\pi<0$, and then $-\hat\pi$; so the maximum is attained at $T=ur_2$, and $\Gamma_H(\bm r)=R_U(r_1+r_2)+m_1u(r_2-r_1)-\hat\pi ur_2$.

Put $e_1=u(m_0+m_1-\hat\pi)>0$ and $e_2=u(\hat\pi-2m_1)>0$. Since $\zeta u=u+d$, $um_0=R_U$ and $dm_1=R_L-R_U$, we have $A_H=(1+\zeta)e_1+e_2=R_L+\zeta u(m_0-\hat\pi)$ and $e_1+e_2=u(m_0-m_1)$. So the two values above are $A_Hr_1$ if $r_2\ge\zeta r_1$, and $(e_1+e_2)r_1+e_1r_2$ if $r_2<\zeta r_1$. For $r_1\le r_2$, we have $A_Hr_2\ge A_Hr_1$ and $e_1r_1+(e_1+e_2)r_2\ge(e_1+e_2)r_1+e_1r_2$; and $A_Hr_1\le(e_1+e_2)r_1+e_1r_2$ if and only if $\zeta e_1r_1\le e_1r_2$, that is, $r_2\ge\zeta r_1$. Hence the minimum of the four linear functions is $\Gamma_H(\bm r)$, and it equals $A_Hr_1$ exactly when $r_2\ge\zeta r_1$. The case $r_1>r_2$ is symmetric; there $\Gamma_H(\bm r)\le A_Hr_2<A_Hr_1$, and $r_2<\zeta r_1$.
\end{proof}

\begin{lemma}[Sides of the menus at $S^*$]\label{lem:signs}
Let $H\ge H_0$ be rational with $\zeta-1<\eta$, let $M$ be an optimal mechanism for $(F_H;n,2)$, where $F_H$ is formed from $F_{p(H)}$, and let $\bm y$ be the capacity table of its base restriction. If $\vartheta\ge\cJ(c)+\tau$, then every $(i,\bm v_{-i})\in\cA_{n,S^*}$ satisfies $1-y_{i2}(\bm v_{-i})\ge\zeta\big(1-y_{i1}(\bm v_{-i})\big)$; if $\vartheta\le\cJ(c)-\tau$, it satisfies $1-y_{i1}(\bm v_{-i})\ge\zeta\big(1-y_{i2}(\bm v_{-i})\big)$.
\end{lemma}

\begin{proof}
Write $\bm r=\one-\bm y_i(\bm v_{-i})$ for $(i,\bm v_{-i})\in\cA_{n,S^*}$, and $\bar{\bm r}=\one-\bm w^*_n(S^*)$. By \Cref{thm:structure}(b) and~(c), $\E_{n,S^*}\bm r=\bar{\bm r}$ and $\E_{n,S^*}\Gamma_H(\bm r)=\Gamma_H(\bar{\bm r})$. Suppose that $\vartheta\ge\cJ(c)+\tau$. By \Cref{prop:Dimbalance} and the identification of $\bm w^*_n$ with $\bm w_\varepsilon$ on $S^*$, we have $\bar r_2-\bar r_1=\bm w^*_n(S^*)_1-\bm w^*_n(S^*)_2\ge\eta$, so $\bar r_2-\zeta\bar r_1\ge\eta-(\zeta-1)>0$, and $\Gamma_H(\bar{\bm r})=A_H\bar r_1$ by \Cref{lem:hinges}. The linear function $\bm r\mapsto A_Hr_1$ is at least $\Gamma_H$ everywhere, so
\[
\E_{n,S^*}\Gamma_H(\bm r)\le\E_{n,S^*}A_Hr_1=A_H\bar r_1=\Gamma_H(\bar{\bm r})=\E_{n,S^*}\Gamma_H(\bm r).
\]
Since every $(i,\bm v_{-i})\in\cA_{n,S^*}$ has positive probability, $\Gamma_H(\bm r)=A_Hr_1$, that is, $r_2\ge\zeta r_1$, for all of them. The case $\vartheta\le\cJ(c)-\tau$ is symmetric.
\end{proof}

\phantomsection
\addcontentsline{toc}{section}{References} 
\bibliographystyle{alphaurl}
\bibliography{refs}

\end{document}